\documentclass[reqno]{amsart}

\usepackage{amsfonts,latexsym,enumerate}
\usepackage{amsmath}
\usepackage{amscd}
\usepackage{float,amsmath,amssymb,mathrsfs,bm,multirow,graphics}
\usepackage[dvips]{graphicx}
\usepackage[percent]{overpic}
\usepackage{amsaddr}
\usepackage[numbers,sort&compress]{natbib}
\usepackage{tikz}
\usepackage{epstopdf}
\usepackage{color}
\usepackage{xcolor} 

\newcommand{\R}{{\Bbb R}}

\newcommand{\C}{{\Bbb C}}

\newcommand{\diag}{\text{\upshape diag\,}}

\newtheorem{RH problem}{RH problem}

\newtheorem{theorem}{Theorem}[section]
\newtheorem{proposition}[theorem]{Proposition}
\newtheorem{lemma}[theorem]{Lemma}

\newtheorem{remark}[theorem]{Remark}

\numberwithin{equation}{section}

\begin{document}
	
	\title[Miura transformations and large-time asymptotic behaviors]
	{Miura transformations and large-time behaviors of the Hirota-Satsuma equation}
	
	\author{Deng-Shan Wang, Cheng Zhu, Xiaodong Zhu}
	\address{Laboratory of Mathematics and Complex Systems (Ministry of Education), School of Mathematical Sciences, Beijing Normal University, Beijing 100875, China}
	\email{dswang@bnu.edu.cn}

	
	\date{March 26, 2024.}
	
	
	
\begin{abstract}

The good Boussinesq equation has several modified versions such as the modified Boussinesq equation, Mikhailov-Lenells equation and Hirota-Satsuma equation. This work builds the full relations among these equations by Miura transformation and invertible linear transformations and draws a pyramid diagram to demonstrate such relations. The direct and inverse spectral analysis shows that the solution of Riemann-Hilbert problem for Hirota-Satsuma equation has simple pole at origin, the solution of Riemann-Hilbert problem for the good Boussinesq equation has double pole at origin, while the solution of Riemann-Hilbert problem for the modified Boussinesq equation and Mikhailov-Lenells equation doesn't have singularity at origin. Further, the large-time asymptotic behaviors of the Hirota-Satsuma equation with Schwartz class initial value is studied by Deift-Zhou nonlinear steepest descent analysis. In such initial condition, the asymptotic expressions of the Hirota-Satsuma equation and good Boussinesq equation \textcolor{blue}{away from the origin are derived and it is shown that the leading term of asymptotic formulas matches well the direct numerical simulations.}

\end{abstract}
	
	\maketitle
	
	
	

	\section{Introduction}\label{sec1}
	
Starting from the B\"{a}cklund transformation of the Boussinesq equation \cite{Boussinesq1872,Zakharov1973,Deift1982,Kaup-1980}, Hirota and Satsuma \cite{Hirota-Satsuma-1977} initially proposed the Hirota-Satsuma equation
\begin{equation}\label{modified Boussinesq equation}
		u_{tt}+\frac{1}{3}(u_{xxx}-\frac{2}{3}u^2u_x-2u_x\partial_x^{-1}u_t)_x=0,
\end{equation}
where $\partial_x^{-1}$ is the inverse operator of operator $\partial=\partial/\partial x$,  which is defined as $(\partial_x^{-1}g)(x)=\int_{-\infty}^{x}g(y)dy$ with $\partial \partial_x^{-1}=\partial_x^{-1}\partial=1$. Introducing special constrain to $u_t$ in \cite{Geng-2003,Geng-2007}, the Hirota-Satsuma equation (\ref{modified Boussinesq equation}) is equivalent to
\begin{equation}\label{HS eq}
\begin{aligned}
			&u_{t}=-u_{xx}+\frac{2}{3}uu_x+2v_x,\\
			&v_{t}=-\frac{2}{3}u_{xxx}+\frac{2}{3}uu_{xx}+\frac{2}{3}u_xv-\frac{2}{3}uv_x+v_{xx},
\end{aligned}
\end{equation}
which is a completely integrable system with Lax pair
\begin{equation}\label{lax pair}
		\begin{cases}
			\psi_x=\tilde{L}\psi,\\
			\psi_t=\tilde{Z}\psi,
		\end{cases}
	\end{equation}
where
	\begin{equation*}
		\tilde{L}=\begin{pmatrix}
			-\frac{1}{3}u & -1 & 0 \\
			-v & \frac{2}{3}u & -\lambda\\
			-1 & 0 & -\frac{1}{3}u
		\end{pmatrix},\quad
		\tilde{Z}=\frac{1}{3}\begin{pmatrix}
			v-u_x-\frac{1}{3}u^2 & -u & 3\lambda \\
			2u_{xx}-3v_x-uv+3\lambda & \frac{2}{3}u^2+v & -\lambda u\\
			2u & 3 & u_x-2v-\frac{1}{3}u^2
		\end{pmatrix},
	\end{equation*}
in which $\lambda=k^3$ is the spectral parameter.
\par
The Hirota-Satsuma equation (\ref{HS eq}) is also called the Hirota-Satsuma type modified Boussinesq equation connecting with the good Boussinesq equation   \cite{Charlier-Lenells-2021,Charlier-Lenells-Wang-2021,Wang-Zhu-2022,Charlier-Lenells-2024}
\begin{equation}\label{good boussinesq}
		w_{tt}+\frac{1}{3}(w_{xx}+4w^2)_{xx}=0,
	\end{equation}
also written as
	\begin{equation}\label{gb eq}
		\begin{aligned}
			&w_{t}=h_{x},\\
			&h_{t}+\frac{1}{3}w_{xxx}+\frac{4}{3}(w^{2})_{x}=0,
		\end{aligned}
	\end{equation}
by a Miura transformation that will be listed below. In fact, the authentic modified equation of the good Boussinesq equation (\ref{gb eq}) is the following modified Boussinesq equation \cite{Wang-Zhu-2022}
\begin{equation}\label{mb eq}
		\begin{aligned}
			&p_{t}=2(pq)_{x}+q_{xx},\\
			&q_{t}=-\frac{1}{3}p_{xx}+\frac{2}{3}pp_{x}-2qq_{x},
		\end{aligned}
	\end{equation}
which was firstly given by Fordy and Gibbons \cite{Fordy-Gibbons-1981} and is equivalent to the Mikhailov-Lenells equation
\begin{equation}\label{L eq}
		\begin{aligned}
			&\mathrm{i}z_{t}+\frac{\sqrt{3}}{3}z_{xx}+2\mathrm{i}rr_{x}=0,\\
			&\mathrm{i}r_{t}-\frac{\sqrt{3}}{3}r_{xx}+2\mathrm{i}zz_{x}=0,
		\end{aligned}
\end{equation}
which was firstly proposed by Mikhailov \cite{Mikhailov-1981} and then by Lenells himself \cite{Lenells-2012}. Actually, the modified Boussinesq equation (\ref{mb eq}) is related with Mikhailov-Lenells equation (\ref{L eq}) by invertible linear transformations
\begin{equation}\label{linear-transformations-1}
		\begin{cases}		p=\frac{1}{2}(-\frac{3}{2}+\frac{\sqrt{3}\mathrm{i}}{2})z-\frac{1}{2}(\frac{3}{2}+\frac{\sqrt{3}\mathrm{i}}{2})r,\\			q=-\frac{1}{2}(\frac{1}{2}+\frac{\sqrt{3}\mathrm{i}}{2})z+\frac{1}{2}(-\frac{1}{2}+\frac{\sqrt{3}\mathrm{i}}{2})r,
		\end{cases}
	\end{equation}
and
	\begin{equation}\label{linear-transformations-2}
		\begin{cases}
			z=-\frac{1}{3}(\frac{3}{2}+\frac{\sqrt{3}\mathrm{i}}{2})p+(-\frac{1}{2}+\frac{\sqrt{3}\mathrm{i}}{2})q,\\
			r=\frac{1}{3}(-\frac{3}{2}+\frac{\sqrt{3}\mathrm{i}}{2})p-(\frac{1}{2}+\frac{\sqrt{3}\mathrm{i}}{2})q.
		\end{cases}
	\end{equation}
That is to say, the modified Boussinesq equation (\ref{mb eq}) is equivalent to Mikhailov-Lenells equation (\ref{L eq}). See Figure \ref{pyramid}.
\par
In the past years, much work \cite{Zakharov1973,Deift1982,Kaup-1980,Charlier-Lenells-2021,Charlier-Lenells-Wang-2021,
Wang-Zhu-2022,Charlier-Lenells-2024} has been done to explore the initial-value problem of the good Boussinesq equation (\ref{gb eq}) because it is an important nonlinear equation describing rich physical phenomena in elastic beam \cite{Falkovich-1983} and nonlinear dielectric materials \cite{Turitsyn-1993}.
Especially, it is worth pointing out that Charlier and Lenells \cite{Charlier-Lenells-2024} opened a way to build the correspondence of Riemann-Hilbert problems for a nonlinear integrable equation and its modified version and then establish Miura-type transformation between them, which provides an approach to derive Miura transformations. Moreover, it has been demonstrated that the solution of  Riemann-Hilbert problem for good Boussinesq equation (\ref{gb eq}) has double pole at origin $k=0$, which brings lots of trouble in nonlinear steepest descent analysis. However, recent work indicates that the solution of Riemann-Hilbert problem for the modified Boussinesq equation (\ref{mb eq}) \cite{Wang-Zhu-2022} and Mikhailov-Lenells equation (\ref{L eq}) \cite{Charlier-2021} doesn't have singularity at $k=0$. In this direction, it is meaningful to investigate the direct and inverse spectral problem of the Hirota-Satsuma equation (\ref{HS eq}) and it will be shown below that the solution of Riemann-Hilbert problem for equation (\ref{HS eq}) has simple pole at origin $k=0$. See Table 1 for details.
\par
After the Hirota-Satsuma equation was proposed, people carried out series of studies to explore its integrability and exact solutions. Quispel, Nijhoff, Capel \cite{Quispel-1982} and Clarkson \cite{Clarkson-1989} respectively investigated the similarity reductions of Hirota-Satsuma equation by introducing proper similarity variables. Geng \cite{Geng-1988} proposed the Lax pair of equation (\ref{HS eq}) and gave its exact solutions by Darboux transformation. Geng and his coauthors \cite{Geng-2003,Geng-2007} presented the Hirota-Satsuma hierarchy by nonlinearization approach of Lax pair and derived the algebraic-geometric solution of Hirota-Satsuma hierarchy by introducing a trigonal curve.
As far as we know, the full relations among Hirota-Satsuma equation (\ref{HS eq}), good Boussinesq equation (\ref{gb eq}) and the modified equations (\ref{mb eq})-(\ref{L eq}) are not given, moreover, the large-time asymptotics of Hirota-Satsuma equation (\ref{HS eq}) with initial-value condition remains unsolved.
\par
\par
	\begin{figure}[htp]\label{pyramid}
		\centering
		\begin{tikzpicture}
			\node at (0,0){H-S};
			\draw[blue,very thick] (0,0) ellipse (0.4 and 0.3);
			\node at (0,2.4) {G-B};
			\draw[blue,very thick] (0,2.4) ellipse (0.4 and 0.3);
			\node at (-2.53,-1) {M-B};
			\draw[blue,very thick] (-2.53,-1) ellipse (0.4 and 0.3);
			\node at (2.53,-1) {M-L};
			\draw[blue,very thick] (2.53,-1) ellipse (0.4 and 0.3);
			\draw[->,blue,very thick] (0,0.4)--(0,2.1);
			\draw[<-,blue,very thick] (-0.5,0)--(-2.1,-0.8);
			\draw[<-,blue,very thick] (0.5,0)--(2.1,-0.8);
			\draw[<->,very thick] (-2,-1)--(2,-1) node[midway, above] {(\ref{linear-transformations-2})} node[midway,below] {(\ref{linear-transformations-1})};
			\draw[<->,blue,very thick] (-2,-1)--(2,-1);
			\draw[->,blue,very thick] (-2.4,-0.6)--(-0.2,2.1);
			\draw[->,blue,very thick] (2.4,-0.6)--(0.2,2.1);
            \draw (-0.25,0.5) node[above] {\rotatebox{90}{(\ref{HS-to-GB})}};
			\draw (1.45,-0.68) node[above] {\rotatebox{330}{(\ref{ML-to-HS})}};
			\draw (-1.4,-0.68) node[above] {\rotatebox{30}{(\ref{mGB-to-HS})}};
			\draw (1.6,0.3) node[above] {\rotatebox{310}{(\ref{ML-to-GB})}};
			\draw (-1.6,0.3) node[above] {\rotatebox{50}{(\ref{mGB-to-GB})}};
		\end{tikzpicture}
		\caption{Relations among Hirota-Satsuma equation (\ref{HS eq}), good Boussinesq equation (\ref{gb eq}), modified Boussinesq equation (\ref{mb eq}) and Mikhailov-Lenells equation (\ref{L eq}), which are abbreviated as H-S, G-B, M-B and M-L, respectively.}
	\end{figure}
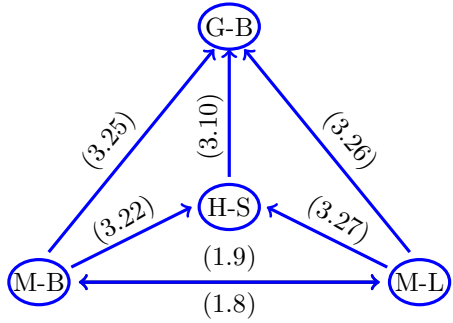

\par
The Riemann-Hilbert approach is a powerful tool to study the initial-value problems of integrable systems \cite{A-Constantin-2006,A-Constantin-2007,A-Constantin-2010} from real physics, especially the issues above.
The present work aims to solve these problems by Riemann-Hilbert formulation and nonlinear steepest descent analysis proposed by Deift and Zhou \cite{Deift-Zhou} and then developed by many scholars \cite{Bertola-2013}-\cite{Girotti-2023}. To to specific, Section \ref{Section-2} first analyzes the inverse spectral problem via the Lax pair (\ref{lax pair}) to formulate the Riemann-Hilbert problem for function $\mathcal{M}(x,t,k)$ of Hirota-Satsuma equation (\ref{HS eq}), then Section \ref{Section-3} builds the relationship between $\mathcal{M}(x,t,k)$ and the solutions to Riemann-Hilbert problems of equations (\ref{gb eq})-(\ref{L eq}) to derive the Miura transformations among these equations (see Figure \ref{pyramid}), and finally, after deforming the Riemann-Hilbert problem for function $\mathcal{M}(x,t,k)$, Section \ref{Section-4} formulates the large-time asymptotics of Hirota-Satsuma equation (\ref{HS eq}) in term of reflection coefficient.
\par
In what follows, we list the main results of this paper. Summarizing Theorems
\ref{RHP between gb and hs}-\ref{RHP between gb and mb}, a graph of relations among Hirota-Satsuma equation (\ref{HS eq}), good Boussinesq equation (\ref{gb eq}) and the modified equations (\ref{mb eq})-(\ref{L eq}) is plotted in Figure \ref{pyramid}, which demonstrates the pyramid relationship of Hirota-Satsuma equation (\ref{HS eq}) with equations (\ref{gb eq})-(\ref{L eq}) by Miura transformation and invertible linear transformations. It is observed that the top of pyramid is the good Boussinesq equation (\ref{gb eq}), the bottoms are the modified equations (\ref{mb eq})-(\ref{L eq}), while the center of pyramid is Hirota-Satsuma equation (\ref{HS eq}). It is observed that the good Boussinesq equation (\ref{gb eq}) can be derived from the Hirota-Satsuma equation (\ref{HS eq}), the modified Boussinesq equation (\ref{mb eq}) and Mikhailov-Lenells equation (\ref{L eq}) \cite{Charlier-Lenells-2024} by the Miura transformations, respectively.
\par
\textcolor{blue}{Introduce the following $3\times3$-matrix valued eigenfunctions (see (\ref{Volterra-x}) and (\ref{xa-integral})) by
$$
	\begin{aligned}
		Y_{1}(x, k) &= I - \int_{x}^{\infty} e^{(x-s) \widehat{\mathcal{L}(k)}}(U Y_{1})(s, k) \, ds, \\
		Y^{A}_{1}(x, k) &= I + \int_{x}^{\infty} e^{-(x-s) \widehat{\mathcal{L}(k)}} (U^{T} Y_{1}^{A})(s, k) \, ds
	\end{aligned}
$$
with
$$
U(x,t;k)=\frac{u(x,t)}{3}\begin{pmatrix}
	0 & \omega^2 & \omega\\
	\omega & 0 & \omega^2\\
	\omega^2 & \omega & 0
\end{pmatrix}-\frac{v(x,t)}{3k}\begin{pmatrix}
\omega^2 & 1 & \omega\\
1 & \omega & \omega^2\\
\omega & \omega^2 & 1
\end{pmatrix},
$$
where $e^{\widetilde{\mathcal{L}}}$ is an operator defined by $e^{\widetilde{\mathcal{L}}}U:=e^{{\mathcal{L}}}Ue^{{-\mathcal{L}}}$, and $U^T$ denotes the transpose of matrix $U$. Moreover, the scattering matrices $s(k)$ and $s^{A}(k)$, as given in (\ref{s(k)-integral}) and (\ref{sa-integral}), can be expressed as 
$$
\begin{aligned}
s(k)&=I-\int_{-\infty}^{\infty} e^{-x \widehat{\mathcal{L}(x)}}\left(U Y_{1}\right)\left(x, k\right) dx,\\
s^{A}(k)&=I+\int_{-\infty}^{\infty} e^{x \widehat{\mathcal{L}(k)}}\left(U^{T} Y_{1}^{A}\right)\left(x, k\right) dx.
\end{aligned}
$$
As a result, one can define reflection coefficients $\rho_{1}(k)$ and $\rho_{2}(k)$ in (\ref{r1r2}),i.e.,
$$
	\begin{cases}
		\rho_{1}(k)=\frac{s_{12}(k)}{s_{11}(k)},\qquad & k\in(-\infty,0),\\
		\rho_{2}(k)=\frac{s_{12}^{A}(k)}{s_{11}^{A}(k)},\qquad & k\in(0,\infty).
	\end{cases}
$$}
\par
For Schwartz class initial value $\{u_{0},v_{0}\}$ of the Hirota-Satsuma equation (\ref{HS eq}), assume the solitons under this initial condition are absent, that is the elements $s_{11}(k)$ and $s^A_{11}(k)$ are nonzero for \(k \in \bar{D}_4 \setminus \{0\}\) and \(k \in \bar{D}_1 \setminus \{0\}\), respectively. Then the large-time asymptotics for the Hirota-Satsuma equation (\ref{HS eq}) and good Boussinesq equation (\ref{gb eq}) with Schwartz class initial values can be formulated in the theorems as follow.
\par

\begin{figure}\label{rho1-rho2-HS}
	\centering
	\includegraphics[width=8cm]{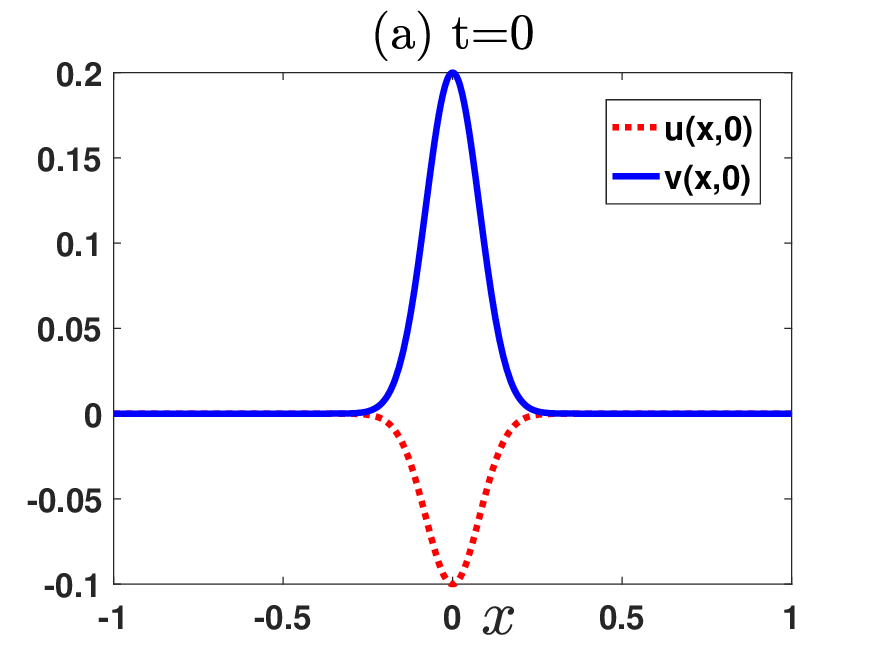}
    \includegraphics[width=7cm]{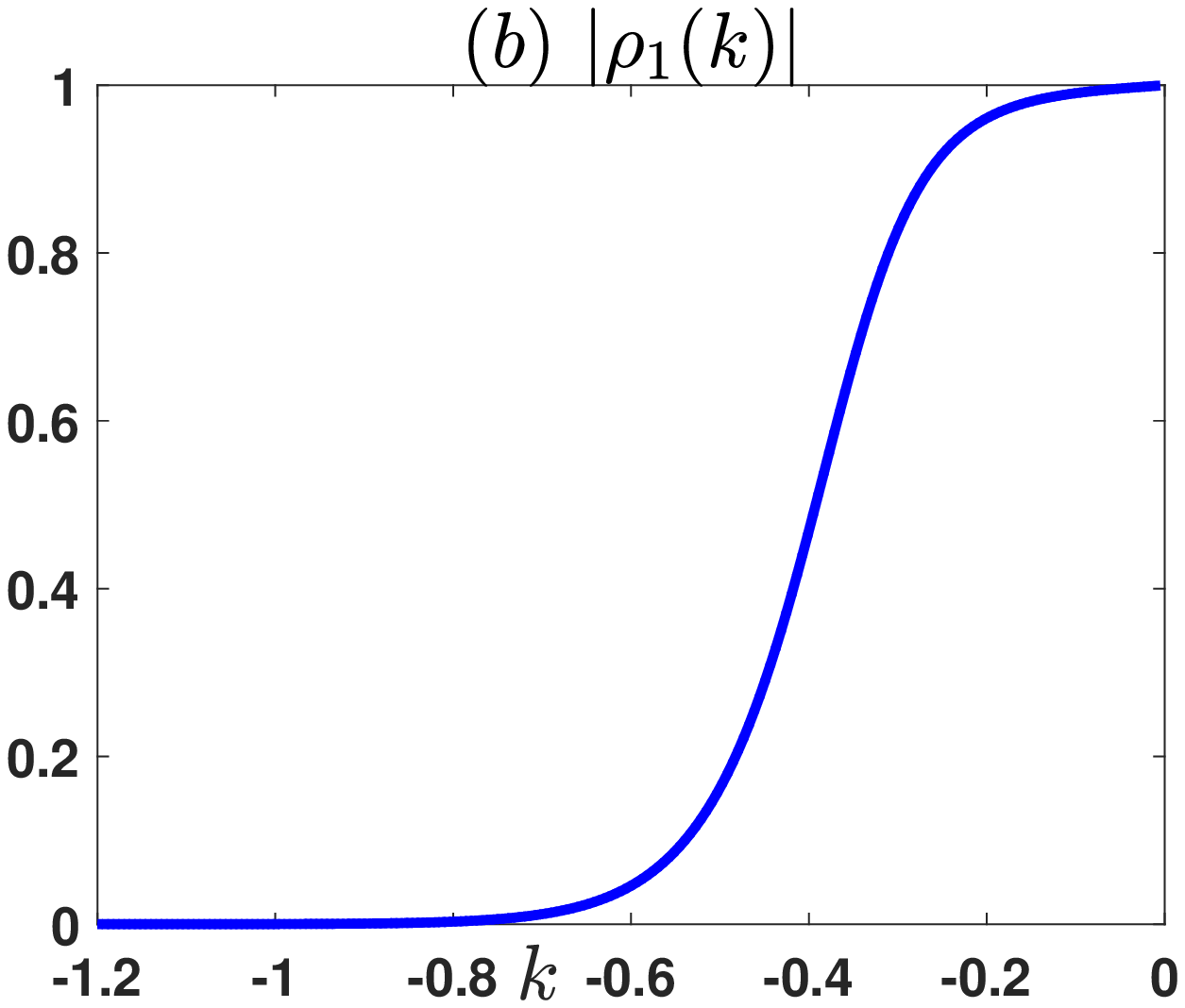}
    \includegraphics[width=7cm]{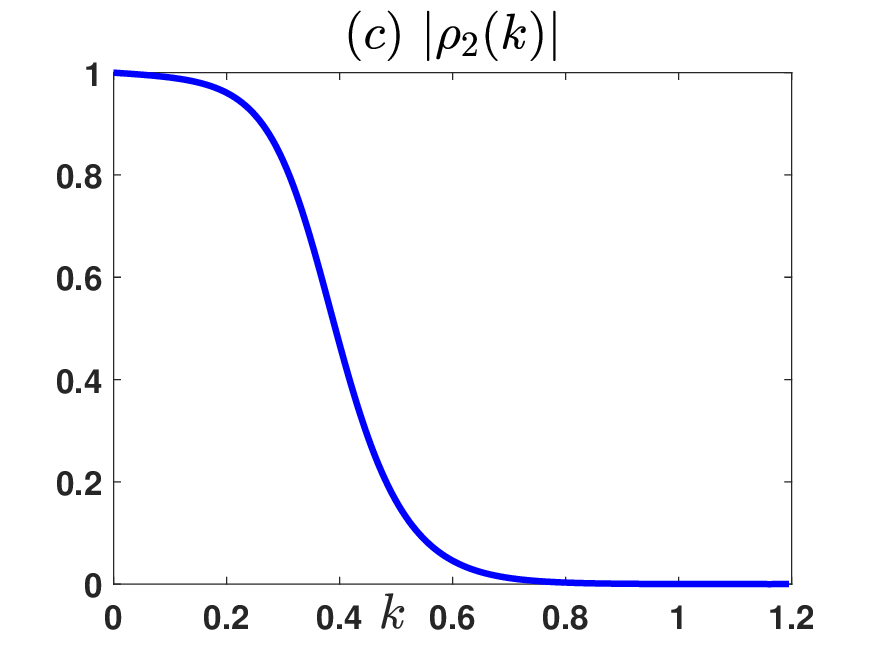}
	\caption{{\protect\small (a) The images of initial condition $u(x, 0)$ and $v(x, 0)$ in (\ref{initial-value}). (b) The absolute value of reflection coefficient $\rho_1(k)$ for $k<0$. (c) The absolute value of reflection coefficient $\rho_2(k)$ for $k>0$. }}
\end{figure}

\begin{figure}\label{utvt=100duibi}
	\centering
	\includegraphics[width=13cm]{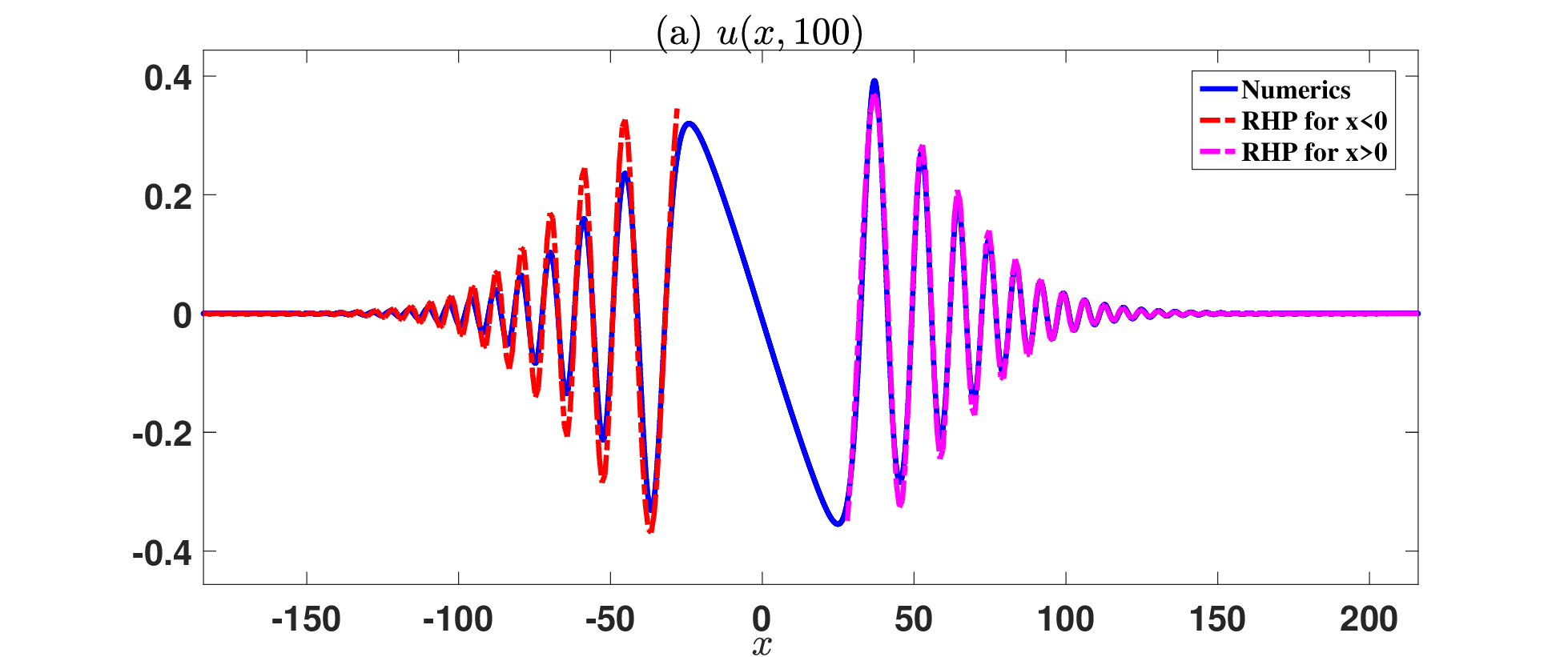}
    \includegraphics[width=13cm]{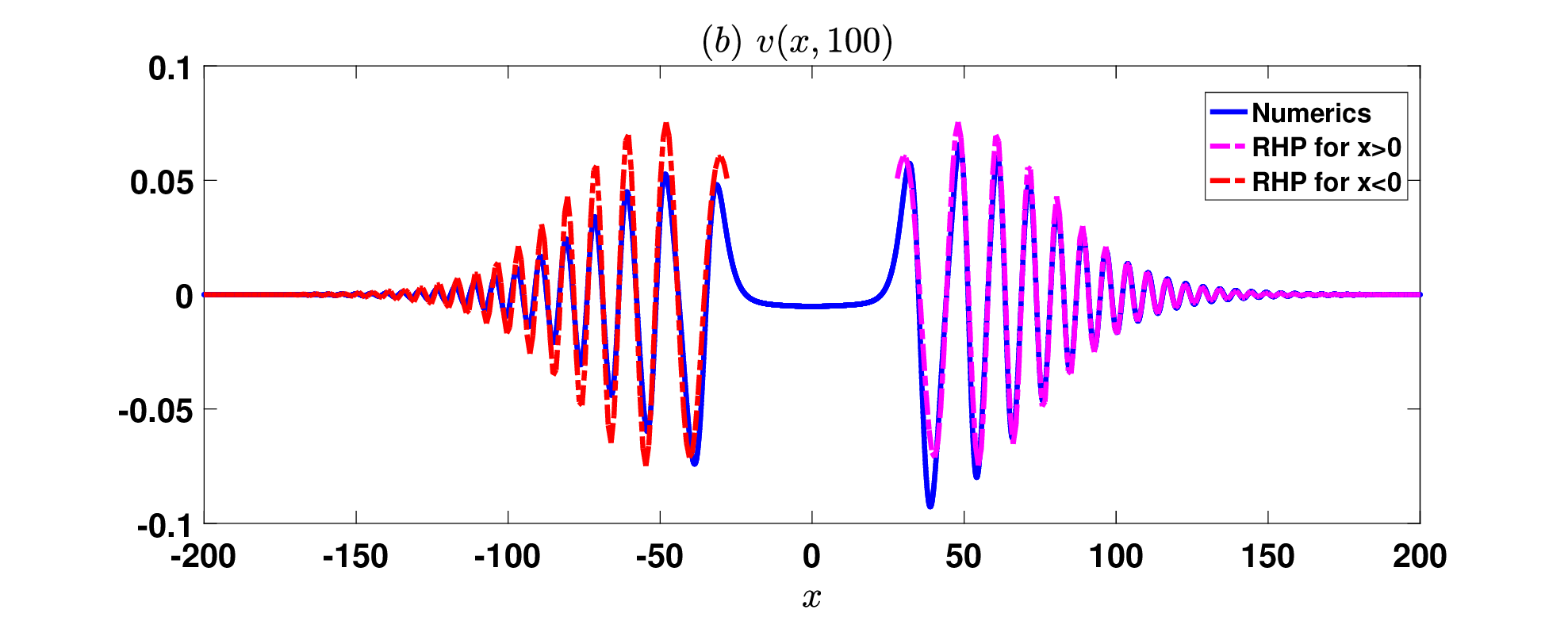}
	\caption{{\protect\small The comparisons of asymptotic expressions given in Theorem \ref{large-time theorem-HS} and full numerical simulations of the Hirota-Satsuma equation (\ref{HS eq}) with initial condition (\ref{initial-value}) at time $t=100$. }}
\end{figure}

\begin{theorem}\label{large-time theorem-HS}
		Let $\{u(x,t),v(x,t)\}$ be Schwartz class solution of the Hirota-Satsuma equation (\ref{HS eq}) with initial value $u_{0}=u(x,0),v_{0}=v(x,0)\in\mathcal{S}(\mathbb{R})$ and assume the elements $s_{11}(k)\neq 0$ and $s^A_{11}(k)\neq 0$ for \(k \in \bar{D}_4 \setminus \{0\}\) and \(k \in \bar{D}_1 \setminus \{0\}\), respectively, and
		\begin{equation*}
			\lim\limits_{k\to 0}ks_{11}(k)\neq 0,\quad 	\lim\limits_{k\to 0}ks_{11}^{A}(k)\neq 0,
		\end{equation*}
such that the functions $\rho_{1}(k)$ and $\rho_{2}(k)$ in (\ref{r1r2}) satisfy $\rho_{1}(0)=\rho_{2}(0)=\omega$, $|\rho_{1}(k)|<1$ for $k<0$ and $|\rho_{2}(k)|<1$ for $k>0$, \textcolor{blue}{and denote $\mathcal{K}$ as a compact subset of $(0,\infty )$}, 
then the large-time solution of the Hirota-Satsuma equation (\ref{HS eq}) holds
		\begin{enumerate}[$(1)$]
			\item \textcolor{blue}{For $x/t> 0$ and $x/t\in\mathcal{K}$}, as $t\to\infty$, the asymptotic behaviors of $u(x,t)$ and $v(x,t)$ are
			\begin{equation}\label{asympotics for k more than 0}
				\begin{aligned}
					u(x,t)&=-\frac{3^{5/4}\sqrt{2\mu}}{\sqrt{t}}\cos
					\big(\mu\ln(6\sqrt{3}tk_{0}^{2})+\frac{1}{\pi}\int^{-\infty}_{k_0}\ln{\frac{|s-k_{0}|}{|s-\omega k_{0}|}}d\textcolor{blue}{\ln(1-|\rho_1(s)|^2)}\\
					&-\sqrt{3}tk_0^{2}+\frac{3\pi}{4}-\arg{p}-\arg{\Gamma(i\mu)}\big)+O(\ln{(t)}/t),\\
					v(x,t)&=\frac{3^{5/4}\times \sqrt{2\mu}k_0}{\sqrt{t}}\sin\big(\mu\ln(6\sqrt{3}tk_{0}^{2})+\frac{1}{\pi}\int^{-\infty}_{k_0}\ln{\frac{|s-k_{0}|}{|s-\omega k_{0}|}}d\textcolor{blue}{\ln(1-|\rho_1(s)|^2)}\\
					&-\sqrt{3}tk_0^{2}+\frac{19\pi}{12}-\arg{p}-\arg{\Gamma(i\mu)}\big)+O(\ln{(t)}/t),
				\end{aligned}
			\end{equation}
			where $\mu=-\frac{1}{2\pi}\ln(1-|\rho_1(k_0)|^2)$ and $p=\rho_{1}(k_0)$ with $k_0=-x/(2t).$
		\end{enumerate}
		\begin{enumerate}[$(2)$]
			\item \textcolor{blue}{For $x/t< 0$ and $|x|/t \in \mathcal{K}$,} as $t\to\infty$, the asymptotic behaviors of $u(x,t)$ and $v(x,t)$ are
			\begin{equation}\label{asympotics for k less than 0}
				\begin{aligned}
					u(x,t)&=\frac{3^{5/4}\sqrt{2\tilde{\mu}}}{\sqrt{t}}\cos					\big(\tilde{\mu}\ln(6\sqrt{3}tk_{0}^{2})+\frac{1}{\pi}\int^{\infty}_{k_0}\ln{\frac{|s-k_{0}|}{|s-\omega k_{0}|}}d\textcolor{blue}{\ln(1-|\rho_2(s)|^2)}\\					&-\sqrt{3}tk_0^{2}+\frac{3\pi}{4}-\arg{\tilde{p}}-\arg{\Gamma(i\tilde{\mu})}\big)+O(\ln{(t)}/t),\\
					v(x,t)&=\frac{3^{5/4}\times \sqrt{2\tilde{\mu}}k_0}{\sqrt{t}}\sin\big(\tilde{\mu}\ln(6\sqrt{3}tk_{0}^{2})+\frac{1}{\pi}\int^{\infty}_{k_0}\ln{\frac{|s-k_{0}|}{|s-\omega k_{0}|}}d\textcolor{blue}{\ln(1-|\rho_2(s)|^2)}\\					&-\sqrt{3}tk_0^{2}+\frac{19\pi}{12}-\arg{\tilde{p}}-\arg{\Gamma(i\tilde{\mu})}\big)+O(\ln{(t)}/t).
				\end{aligned}
			\end{equation}
			where $\tilde{\mu}=-\frac{1}{2\pi}\ln(1-|\rho_2(k_0)|^2)$ and $\tilde{p}=\rho_{1}(-k_0)$.
		\end{enumerate}
	\end{theorem}
\begin{proof}
See Section \ref{Section-4}.
\end{proof}
\textcolor{blue}{\begin{remark}
For any positive constant \( M > 1 \) and \( M \leq {|x|}/{t} < \infty \), assume that the spectral functions \(\rho_j\), \(j = 1, 2\), satisfy the conditions of Theorem \ref{large-time theorem-HS}. Then, the large-time solutions of the Hirota-Satsuma equation in (\ref{asympotics for k more than 0}) and (\ref{asympotics for k less than 0}) still hold, except that the error term \(O\left(\frac{\ln t}{t}\right)\) is replaced by \(O\left(\frac{1}{|x|^N} + C_N\left({|x|}/{t}\right)\frac{\ln x}{|x|}\right)\) for any fixed \(N \geq 1\), where \(C_N(\eta)\) decays rapidly as \(\eta \to \infty\) for each \(N\). In particular, since \(\rho_{1}(-k)\) and \(\rho_{2}(k)\) vanish to arbitrary order as \(k \to \infty\), it follows that
\[
\begin{aligned}
	u(x,t) &= O\left(\frac{1}{|x|^N} + \frac{C_N(|x|/t)}{\sqrt{|x|}}\right),\ 
	v(x,t) &= O\left(\frac{1}{|x|^N} + \frac{C_N(|x|/t)}{\sqrt{|x|}}\right),
\end{aligned}
\]
as \(|x| \to \infty\), which is consistent with the fact that \(u\) and \(v\) are Schwartz class solutions.
\end{remark}
\begin{proof}
	The proof of this remark follows a similar procedure as outlined in Section \ref{Section-4}. For detailed steps, refer to \cite{Teschl-2009}.
\end{proof}}
The large-time asymptotic behavior of the equation
(\ref{good boussinesq}) with Schwartz class initial value has been studied in Charlier, Lenells and Wang \cite{Charlier-Lenells-Wang-2021}. However, thanks to the Miura transformation (\ref{HS-to-GB}) below, a similar asymptotic result can be given in the theorem as follows.

\begin{theorem}\label{large-time theorem-GB}
	Let $w(x,t)$ be a solution in the Schwartz class of the good Boussinesq equation (\ref{good boussinesq}) with initial value $w_{0}=w(x,0)\in\mathcal{S}(\mathbb{R})$. Assume the conditions stated in Theorem \ref{large-time theorem-HS} are satisfied.
	 Then the Miura transformation $w=-\frac{1}{6}u^2-\frac{1}{2}v$ and Theorem \ref{large-time theorem-HS} yield the large-time solution of the good Boussinesq equation (\ref{good boussinesq}) for $x/t>0$ and $x/t\in \mathcal{K}$, as follows
	\begin{equation}\label{Asympototics-GB}
		\begin{aligned}
			w(x,t)&=-\frac{3^{5/4}\times \sqrt{\mu}k_0}{\sqrt{2t}}\sin\big(\mu\ln(6\sqrt{3}tk_{0}^{2})+\frac{1}{\pi}
			\int^{-\infty}_{k_0}\ln{\frac{|s-k_{0}|}{|s-\omega k_{0}|}}d\textcolor{blue}{\ln(1-|\rho_1(s)|^2)}\\
			&-\sqrt{3}tk_0^{2}+\frac{19\pi}{12}-\arg{p}-\arg{\Gamma(i\mu)}\big)+O(\ln{(t)}/t),\quad t\to\infty, \\
		\end{aligned}
	\end{equation}
where $\mu(k_0)=-\frac{1}{2\pi}\ln(1-|\rho_1(k_0)|^2)$ and $p=\rho_{1}(k_0)$ with $k_0=-x/(2t).$
\end{theorem}

In order to check the validity of the long-time asymptotics of the Hirota-Satsuma equation (\ref{HS eq}) in Theorem \ref{large-time theorem-HS}, take the initial value of the form
\begin{equation}\label{initial-value}
u(x,0)=-0.1e^{-x^2/20},\quad v(x,0)=0.2e^{-x^2/20},
\end{equation}
which are Gaussian wave packets shown in Figure \ref{rho1-rho2-HS}(a). Under this initial condition,
Figure \ref{rho1-rho2-HS}(b)-(c) demonstrate the modulus of the functions $\rho_1(k)$ and $\rho_2(k)$, which are reflection coefficients. It is seen that $|\rho_1(0)|=1, \lim\limits_{k\to-\infty}\rho_1(k)=0,  |\rho_1(k)|<1$ for $k<0$ and $|\rho_2(0)|=1, \lim\limits_{k\to+\infty}\rho_2(k)=0,  |\rho_2(k)|<1$ for $k>0$. Figure \ref{utvt=100duibi} shows the comparison of asymptotic expressions given in Theorem \ref{large-time theorem-HS} with numerical simulations with initial condition (\ref{initial-value}) at $t=100$, in which the dashed pink and red lines correspond to the asymptotic expressions in (\ref{asympotics for k more than 0}) and (\ref{asympotics for k less than 0}) while the blue solid lines represent the numerical results. Expectedly, the large-time asymptotic solutions can approximate the numerical solutions within certain error range.
\textcolor{blue}{\begin{remark}
For the region near the origin, we observe that when \(\frac{|x|}{\sqrt{t}} \leq M\), where \(M\) is a positive constant, the modified Boussinesq equation \eqref{mb eq} exhibits the following asymptotic behavior as \(t \to \infty\):
\begin{equation}\label{painleve for mb}
	\begin{aligned}
		p(x, t) &= \frac{\sqrt{3}}{4 \sqrt{t}} \frac{\left(p_{\mathrm{IV}}^{\prime}(y) + \frac{2}{3}\right)}{p_{\mathrm{IV}}(y)} + \mathcal{O}(t^{-1}), \\
		q(x, t) &= \frac{\sqrt{3}}{4 \sqrt{t}} \left(p_{\mathrm{IV}}(y) + \frac{2}{3} y\right) + \mathcal{O}(t^{-1}),
	\end{aligned}
\end{equation}
where $y = -\frac{\sqrt{3} x}{2 \sqrt{t}}$ and $p_{\mathrm{IV}}(y)$ satisfies the Painlevé IV equation:
\begin{equation}\label{Painleve}
	\frac{p_{\mathrm{IV}}^{\prime}(y)^2}{2 p_{\mathrm{IV}}(y)} + 4 y p_{\mathrm{IV}}^2(y) + \frac{3 p_{\mathrm{IV}}^3(y)}{2} + 2 y^2 p_{\mathrm{IV}}(y) - \frac{2}{9 p_{\mathrm{IV}}(y)} = p_{\mathrm{IV}}^{\prime\prime}(y).
\end{equation}
Consequently, by applying the Miura transformation from (\ref{mGB-to-HS}), we derive the asymptotics for Hirota-Satsuma equation (\ref{modified Boussinesq equation}) near the origin.Numerical results for \(u\) are shown in Figure \ref{u100vsP4}, with the dashed green line corresponding to the numerical solution of the Painlevé IV equation after applying the Miura transformation \eqref{mGB-to-HS}. This suggests that the region near the origin can be described by the solution in \eqref{Painleve}. A rigorous proof of \eqref{painleve for mb} will be provided in a forthcoming paper.
	 \begin{figure}
	 	\centering
	 	\includegraphics[width=13cm]{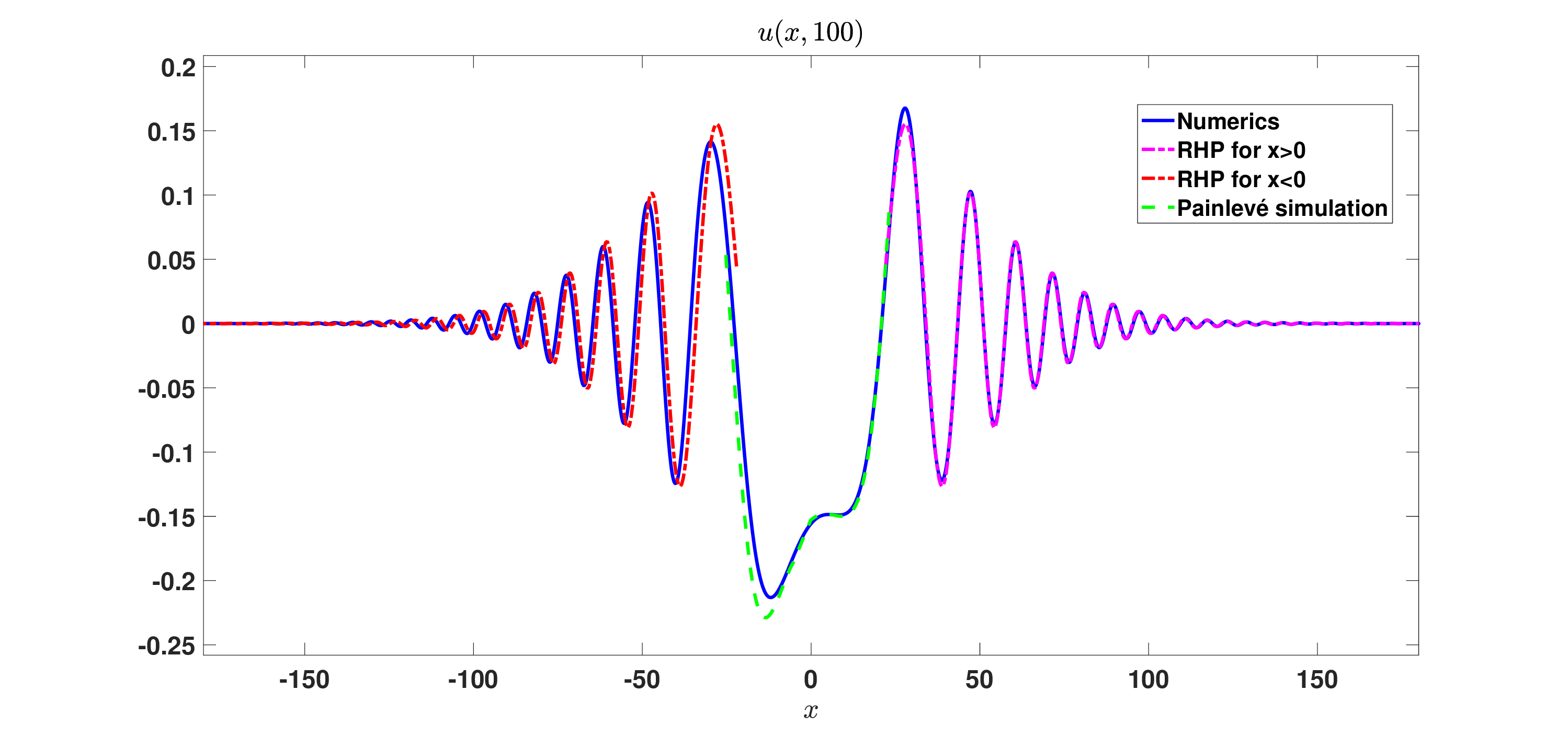}
	 	\caption{{\protect\small The comparisons of RHP reulst and Painlev$\acute{e}$ simulations with full numerical simulations of the Hirota-Satsuma equation (\ref{HS eq}) with initial condition $u(x,0)=-0.6 e^{-\frac{x^2}{20}}, v(x,0)=-0.08e^{-\frac{x^2}{10}}+0.02xe^{-\frac{x^2}{20}}$ at time $t=100$. }}
	 	\label{u100vsP4}
	 \end{figure}	 
\end{remark}}
	
	\section{Direct and inverse spectral problem}\label{Section-2}
	
	This section mainly builds the direct and inverse spectral theory of Lax pair (\ref{lax pair}) and proposes the Riemann-Hilbert problem associated with the solution $u$ and $v$ of the Hirota-Satsuma equation (\ref{HS eq}).
	
	\subsection{The modified Lax pair} For $k\in\mathbb{C}\setminus\{0\}$, define $\psi=P\phi$, where
	
	\begin{equation}\label{gauge-transform}
		P(k)=\begin{pmatrix}
			\omega k & \omega^2k  & k\\
			\omega^2 k^2 & \omega k^2 & k^2 \\
			1 & 1 & 1
		\end{pmatrix},\quad \omega=e^{\frac{2i\pi}{3}},
	\end{equation}
	with $\det(P(k))=3(\omega^2-\omega)k^3$, then the Lax pair (\ref{lax pair}) is converted into
	\begin{equation}\label{lax-phi}
		\begin{cases}
			\phi_x=L\phi,\\
			\phi_t=Z\phi,
		\end{cases}
	\end{equation}
	with
	$$
	L=P^{-1}\tilde{L}P, \quad Z=P^{-1}\tilde{Z}P.
	$$
	Rewrite $L$ and $Z$ as
	\begin{equation}\label{L,Z}
		L=\mathcal{L}+U,
		\quad Z=\mathcal{Z}+V,
	\end{equation}
	with $\mathcal{L}=-kJ$ and $\mathcal{Z}=k^2J^2$, where $J$ is a diagonal matrix defined by
	\begin{equation}\label{J}
		J=\begin{pmatrix}
			\omega & 0 & 0 \\
			0 & \omega^2 & 0\\
			0 & 0 & 1
		\end{pmatrix},
	\end{equation}
	and $U, V$ take the forms
	\begin{subequations}\label{U,V}
		\begin{align}
			& U(x,t,k)=U^{(0)}(x,t)+\frac{1}{k}U^{(-1)}(x,t), \label{UUU}\\
			& V(x,t,k)=kV^{(1)}(x,t)+V^{(0)}(x,t)+\frac{1}{k}V^{(-1)}(x,t),\label{VVV}
		\end{align}
	\end{subequations}
	where
	\begin{align*}
		U^{(0)}&=\frac{u}{3}\begin{pmatrix}
			0 & \omega^2 & \omega\\
			\omega & 0 & \omega^2\\
			\omega^2 & \omega & 0
		\end{pmatrix},\quad
		U^{(-1)}=-\frac{v}{3}\begin{pmatrix}
			\omega^2 & 1 & \omega\\
			1 & \omega & \omega^2\\
			\omega & \omega^2 & 1
		\end{pmatrix},\\
		V^{(1)}&=\frac{u}{3}\begin{pmatrix}
			0 & \omega^2 & 1\\
			\omega & 0 & 1\\
			\omega & \omega^2 & 0
		\end{pmatrix},\quad
		V^{(-1)}=\frac{2u_{xx}-uv-3v_x}{9}\begin{pmatrix}
			\omega^2 & 1 & \omega\\
			1 & \omega & \omega^2\\
			\omega & \omega^2 & 1
		\end{pmatrix},\\
		V^{(0)}&=\frac{1}{9}\begin{pmatrix}
			0 & \omega^2 u^2+(1-\omega)u_x-3v & \omega u^2+(\omega+2)u_x-3v\\
			\omega u^2-3v+(\omega+2)u_x & 0 & \omega^2 u^2+(1-\omega)u_x-3v\\
			\omega^2 u^2-3v+(1-\omega)u_x & \omega u^2-3v+(\omega+2)u_x & 0
		\end{pmatrix}.
	\end{align*}
	\par
	Direct calculations show that $L$ and $Z$ satisfy the $\mathbb{Z}_{2}$ and $\mathbb{Z}_{3}$ symmetries below, respectively
	\begin{equation}\label{symetric-conjugate}
		L(k)=\mathcal{B}\overline{L(\bar{k})}\mathcal{B},\quad {\rm and}\quad
		Z(k)=\mathcal{B}\overline{Z(\bar{k})}\mathcal{B},\quad
		\mathcal{B}=\begin{pmatrix}
			0 & 1 & 0 \\
			1 & 0 & 0 \\
			0 & 0 & 1
		\end{pmatrix},
	\end{equation}	
	\begin{equation}\label{symmetric-ratation}
		L(k)=\mathcal{A}L(\omega k)\mathcal{A}^{-1},\quad {\rm and}\quad
		Z(k)=\mathcal{A}Z(\omega k)\mathcal{A}^{-1},\quad
		\mathcal{A}=\begin{pmatrix}
			0 & 0 & 1 \\
			1 & 0 & 0 \\
			0 & 1 & 0
		\end{pmatrix}.
	\end{equation}
	\par
	The decomposition (\ref{U,V}) and the transformation $\phi=Ye^{\mathcal{L}x+\mathcal{Z}t}$ turn the Lax pair (\ref{lax-phi}) into the modified Lax pair
	\begin{equation}\label{lax-pde}
		\begin{cases}
			Y_x-[\mathcal{L},Y]=UY,\\
			Y_t-[\mathcal{Z},Y]=VY.
		\end{cases}
	\end{equation}

	\subsection{The eigenfunctions $Y_{1}$ and $Y_{2}$}
	Fixing time $t$, we are particularly interested in the $x$-part of the modified Lax pair (\ref{lax-pde}), i.e.,
	\begin{equation}\label{lax-x}
		Y_{x}-[\mathcal{L},Y]=UY.
	\end{equation}
	For convenience, rewrite $\mathcal{L}=\diag\{l_1,l_2,l_3\}$ with $l_j=-k\omega^j$~$(j=1,2,3)$ and $\mathcal{Z}=\diag\{z_1,z_2,z_3\}$ with $z_j=k^2\omega^{2j}$~$(j=1,2,3)$.
	Then the complex $k$-plane is decomposed into six open sets $\{\Omega_n\}_1^6$ in Fig. \ref{Figure-jump}, where we have set $S=\{k\in\mathbb{C}|\arg{k}\in(-\frac{\pi}{3},\frac{\pi}{3})\}$, which denotes the interior of $\bar{\Omega}_{6}\cup\bar{\Omega}_{1}$.
	
	\begin{figure}[htp]
		\centering
		\begin{tikzpicture}[>=latex]
			\draw[<->,very thick] (-2,0)node[below]{${\Gamma_4}$} to (2,0)node[above]{${\Gamma_1}$};
			\draw[very thick] (-4,0) to (4,0)node[right]{$\Gamma$};
			\draw[<->,very thick] (-1,-1.732)node[right=2mm]{${\Gamma_5}$} to (1,1.732)node[left=2mm]{${\Gamma_2}$};
			\draw[very thick] (-2,-1.732*2) to (2,1.732*2);
			\draw[<->,very thick] (-1,1.732)node[left]{${\Gamma_3}$} to (1,-1.732)node[right]{${\Gamma_6}$};
			\draw[very thick] (-2,1.732*2) to (2,-1.732*2);
			\node at (-.4,-.3){$0$};
			\node at (3,1.8){$\Omega_{1}$};
			\node at (0,2.8){$\Omega_{2}$};
			\node at (-3,1.8){$\Omega_{3}$};
			\node at (-3,-1.8){$\Omega_{4}$};
			\node at (0,-2.8){$\Omega_{5}$};
			\node at (3,-1.8){$\Omega_{6}$};
		\end{tikzpicture}
		\caption{The jump contour $\Gamma=\cup_{i=1}^{6}\Gamma_i$ and six open sets $\Omega_{j},j=1,\dots,6$.}
		\label{Figure-jump}
	\end{figure}
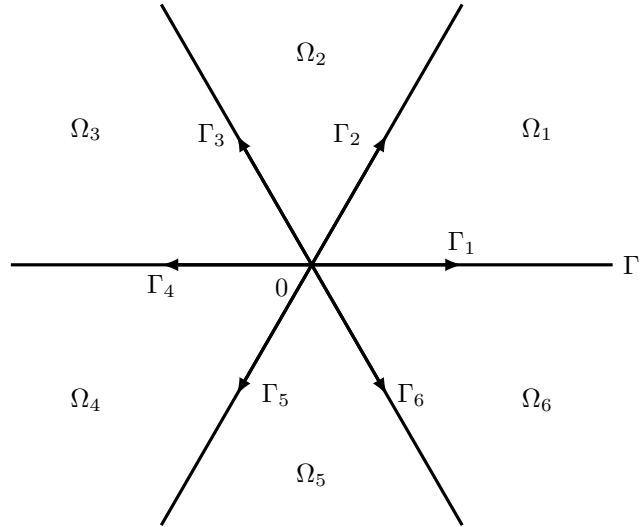
	\par
	Define two three-order matrix-valued eigenfunctions $Y_{1}$ and $Y_{2}$ as the solutions to the Volterra integral equations
	\begin{subequations}\label{Volterra-x}
		\begin{align}
			Y_{1}(x, k)=I-\int_{x}^{\infty} e^{\left(x-s\right) \widehat{\mathcal{L}(k)}}(U(s,k) Y_{1}\left(s, k\right)) ds,\\
			Y_{2}(x, k)=I+\int_{-\infty}^{x} e^{\left(x-s\right) \widehat{\mathcal{L}(k)}}(U(s,k) Y_{2}\left(s, k\right)) ds,
		\end{align}
	\end{subequations}
	where $e^{\left(x-s\right) \widehat{\mathcal{L}(k)}}(UY)=e^{\left(x-s\right)\mathcal{L}(k)}UYe^{-\left(x-s\right)\mathcal{L}(k)}$. Letting $\mathcal{S} \left(\mathbb{R}\right)$ be the Schwartz space, the properties of the eigenfunctions $Y_{1}$ and $Y_{2}$ are given in the next proposition. \\

	\begin{proposition}\label{property of x}
Denote $\Xi=(\omega^2\bar{S},\omega\bar{S},\bar{S})$ and $\hat{\Xi}=(-\omega^2\bar{S},-\omega\bar{S},-\bar{S})$.	If $u_{0},v_{0}\in\mathcal{S} \left(\mathbb{R}\right)$, the matrix-valued solutions $Y_{1}(x,k)$ and $Y_{2}(x,k)$ of (\ref{lax-x}) exhibit the following characteristics:
		\begin{enumerate}[$(1)$]
			\item The function $Y_{1}(x,k)$ is well-defined for $x\in \mathbb{R}$ and for $k$ belonging to the region $\Xi$ excluding zero. Additionally, $Y_{1}(\cdot,k)$ is smooth for any $k$ within this region excluding zero.
		\end{enumerate}
		\begin{enumerate}[$(2)$]
			\item The function $Y_{2}(x,k)$ is well-defined for $x\in \mathbb{R}$ and for $k$ belonging to the region $\hat{\Xi}$ excluding zero. Additionally, $Y_{2}(\cdot,k)$ is smooth for any $k$ within this region excluding zero.
		\end{enumerate}
		\begin{enumerate}[$(3)$]
			\item For given $x\in\mathbb{R}$, $Y_{1}(x,\cdot)$ is analytic for any $k\in\Xi\setminus\{0\}$.
		\end{enumerate}
		\begin{enumerate}[$(4)$]
			\item For given $x\in\mathbb{R}$, $Y_{2}(x,\cdot)$ is analytic for any $k\in\hat{\Xi}\setminus\{0\}$.
		\end{enumerate}
		\begin{enumerate}[$(5)$]
			\item The partial derivative $\frac{\partial^{i}Y_{1}}{\partial k^{i}}(x,\cdot)$ can extend continuously to $k\in \Xi\setminus\{0\}$ for any $x\in\mathbb{R}$ and each integer $i\ge1$.
		\end{enumerate}
		\begin{enumerate}[$(6)$]
			\item The partial derivative $\frac{\partial^{i}Y_{2}}{\partial k^{i}}(x,\cdot)$ can extend continuously to $k\in\hat{\Xi}\setminus\{0\}$ for any $x\in\mathbb{R}$ and each integer $i\ge1$.
		\end{enumerate}
		\begin{enumerate}[$(7)$]
			\item For any integer $n\ge1$, there exist two bounded, smooth, positive functions $g_{1}(x)$ and $g_{2}(x)$ defined on $x\in\mathbb{R}$, which decay rapidly as $x$ approaches $\infty$ and $-\infty$, respectively, such that
			\begin{subequations}
				\begin{align}
					&\left|\frac{\partial^{i}}{\partial k^{i}}(Y_{1}(x,k)-I)\right|\le g_{1}(x),\quad k\in\Xi\setminus\{0\},\quad  &i=0,1,2\cdots,n,\\
					&\left|\frac{\partial^{i}}{\partial k^{i}}(Y_{2}(x,k)-I)\right|\le g_{2}(x),\quad k\in\hat{\Xi}\setminus\{0\},\quad  &i=0,1,2\cdots,n.
				\end{align}
			\end{subequations}
		\end{enumerate}
		\begin{enumerate}[$(8)$]
			\item For any $x\in\mathbb{R}$, the eigenfunctions $Y_{1}$ and $Y_{2}$ also satisfy the $\mathbb{Z}_{2}$ and $\mathbb{Z}_{3}$ symmetries
			\begin{subequations}\label{symmetry-X}
				\begin{align}
					&Y_{1}(x,k)=\mathcal{A}Y_{1}(x,\omega k)\mathcal{A}^{-1}=\mathcal{B}\overline{Y_{1}(x,\bar{k})}\mathcal{B},\quad k\in\Xi\setminus\{0\},\\
					&Y_{2}(x,k)=\mathcal{A}Y_{2}(x,\omega k)\mathcal{A}^{-1}=\mathcal{B}\overline{Y_{2}(x,\bar{k})}\mathcal{B},\quad k\in\hat{\Xi}\setminus\{0\}.
				\end{align}
			\end{subequations}
		\end{enumerate}
		\begin{enumerate}[$(9)$]
			\item Suppose $u_{0},v_{0}$ are supported compactly. For every $x\in\mathbb{R}$ and $k\in\mathbb{C}\setminus \{0\}$, both $Y_{1}(x,k)$ and $Y_{2}(x,k)$ are analytic, and their determinants are equal to $1$: $\det{Y_{1}(x,k)}=\det{Y_{2}(x,k)}=1$.
		\end{enumerate}
	\end{proposition}

	\subsection{Asymptotics of the eigenfunctions $Y_{1}$ and $Y_{2}$ as $k\to\infty$}
	The differential equation (\ref{lax-x}) allows the formal solutions
	\begin{equation*}
		\begin{aligned}
			Y_{1}^{formal}(x,k)=I+\frac{Y_1^{(1)}(x)}{k}+\frac{Y_{1}^{(2)}(x)}{k^2}+\dots,\\
			Y_{2}^{formal}(x,k)=I+\frac{Y_{2}^{(1)}(x)}{k}+\frac{Y_{2}^{(2)}(x)}{k^2}+\dots,
		\end{aligned}
	\end{equation*}
	in which the coefficients satisfy
	\begin{equation*}\label{x-8}
		\lim_{x\to+\infty}Y_{1}^{(j)}(x)=\lim_{x\to-\infty}Y_{2}^{(j)}(x)=0,  \quad j\ge1.
	\end{equation*}
	For convenience, rewrite the matrix $U(x,k)$ as
	\begin{equation*}
		U(x,k)=U_{0}(x)+\frac{U_{1}(x)}{k},
	\end{equation*}
	with
	\begin{equation*}
		U_{0}(x)=U^{(0)}(x,0), \qquad U_{1}(x)=U^{(-1)}(x,0).
	\end{equation*}
	Inserting
	\begin{equation*}
		Y_{1}=I+\frac{Y_{1}^{(1)}(x)}{k}+\frac{Y_{1}^{(2)}(x)}{k^2}+\cdots
	\end{equation*}
	into the equation (\ref{lax-x}) yields
	\begin{equation*}
		\begin{cases}
			[J,Y_{1}^{(j+1)}]=-\partial_{x}(Y_{1}^{(j)})^{(o)}+(U_0Y_{1}^{(j)}+U_{1}Y_{1}^{(j-1)})^{(o)},\\
			\partial_{x}(Y_{1}^{(j+1)})^{(d)}=(U_0Y_{1}^{(j+1)}+U_{1}Y_{1}^{(j)})^{(d)},
		\end{cases}
	\end{equation*}
	in which $Y^{(d)}$ denotes the diagonal parts of the matrix $Y$ and $Y^{(o)}$ denotes the off-diagonal parts. Some calculations give the first coefficient of $Y_{1}$ of the form
	\begin{equation*}
		Y_1^{(1)}(x)=-\int_{\infty}^{x}\frac{u^2_0(x')+3v_0(x')}{9}\begin{pmatrix}
			\omega^2 & 0 &0\\
			0 & \omega  & 0\\
			0 & 0 & 1
		\end{pmatrix}dx'-\frac{u_0}{9}\begin{pmatrix}
			0 & -\omega(\omega+2) & \omega(\omega+2)\\
			\omega+2 &    0  &- (\omega+2)\\
			-\omega^2(\omega+2) & \omega^2(\omega+2) &0
		\end{pmatrix}.
	\end{equation*}

	\begin{proposition}\label{kto-8}
		Let $u_{0},v_{0}\in\mathcal{S}(\mathbb{R})$, then when $k\to{\infty}$, $Y_{1}$ and $Y_{2}$ approximate to $Y_{1}^{formal}$ and $Y_{2}^{formal}$, respectively. Furthermore, assuming $p\ge0$ is an integer, define two functions
		\begin{subequations}\label{X,P}
			\begin{align}
				Y_{1,p}(x, k):&=I+\frac{Y_{1}^{(1)}(x)}{k}+\cdots+\frac{Y_{1}^{(p)}(x)}{k^{p}},\\
				Y_{2,p}(x, k):&=I+\frac{Y_{2}^{(1)}(x)}{k}+\cdots+\frac{Y_{2}^{(p)}(x)}{k^{p}}.
			\end{align}
		\end{subequations}
		Then there exist two bounded, smooth, positive functions $g_{1}(x)$ and $g_{2}(x)$ defined for $x\in\mathbb{R}$, which decay rapidly as $x$ approaches $\infty$ and $-\infty$, respectively, such that
		\begin{subequations}
			\begin{align}
				&\left|\frac{\partial^{i}}{\partial k^{i}}\left(Y_{1}-Y_{1,p}\right)\right| \leq \frac{g_{1}(x)}{|k|^{p+1}}, \quad  k \in \Xi,\quad |k| \geq 2, \\
				&\left|\frac{\partial^{i}}{\partial k^{i}}\left(Y_{2}-Y_{2,p}\right)\right| \leq \frac{g_{2}(x)}{|k|^{p+1}}, \quad  k \in\hat{\Xi},\quad |k| \geq 2,
			\end{align}
		\end{subequations}
		for any $x\in\mathbb{R}$ and $i=0,1,2,\cdots .$
	\end{proposition}

	\subsection{Asymptotics of the eigenfunctions $Y_{1}$ and $Y_{2}$ as $k\to0$}
	The eigenfunctions $Y_{1}$ and $Y_{2}$ have a first-order pole at $k=0$ because the function $U(x,t,k)$ in (\ref{UUU}) has a first-order pole, which is caused by the gauge transformation (\ref{gauge-transform}). In such case, the following proposition holds near the origin $k=0$.

	\begin{proposition}\label{k to 0}
		Let $u_{0},v_{0}\in\mathcal{S}(\mathbb{R})$ and integer $p\ge0$. There are third-order matrix-valued functions $C_{i}^{(l)}(x)$ for $i=1,2; l=-1,0,\dots ,p$, which satisfy
		\begin{enumerate}[$(1)$]
			\item For each $x\in\mathbb{R}$, there exist two smooth functions $g_{1}(x)>0$ and $g_{2}(x)>0$ which decay rapidly as $x\to{\infty}$ and  $x\to{-\infty}$ respectively, such that
			\begin{subequations}
				\begin{align}
					\left| \frac{\partial^{j}}{\partial k^{j}}(Y_{1}-I-\sum_{l=-1}^{p}C_{1}^{(l)}(x)k^l)\right| \le g_{1}(x)|k|^{p+1-j},\quad k \in \Xi,\quad |k|\le\frac{1}{2},\\
					\left| \frac{\partial^{j}}{\partial k^{j}}(Y_{2}-I-\sum_{l=-1}^{p}C_{2}^{(l)}(x)k^l)\right| \le g_{2}(x)|k|^{p+1-j},\quad k \in \hat{\Xi}, \quad |k|\le\frac{1}{2},
				\end{align}
			\end{subequations}
			where $j\ge0$ is an integer.
		\end{enumerate}
		\begin{enumerate}[$(2)$]
			\item The functions $C_{1}^{(l)}$ and $C_{2}^{(l)}$ are smooth and decay rapidly as $x\to{\infty}$ and  $x\to{-\infty}$ respectively for $l=-1,0,1,2,\cdots$.
		\end{enumerate}
		\begin{enumerate}[$(3)$]
			\item The leading coefficients are given by
			\begin{align}
				C_{i}^{(-1)}(x)=&\alpha_{i}(x)\begin{pmatrix}
					\omega^2 & 1 & \omega \\
					1 & \omega & \omega^{2} \\
					\omega & \omega^2 & 1
				\end{pmatrix},\label{c-1}	\\
				C_{i}^{(0)}(x)=&\beta_{i}(x)\begin{pmatrix}
					1 & \omega^2  & \omega \\
					\omega & 1 & \omega^{2} \\
					\omega^2 & \omega & 1
				\end{pmatrix}+\gamma_{i}(x)\begin{pmatrix}
					1 & \omega &\omega^2  \\
					\omega^2 & 1 &\omega  \\
					\omega & \omega^2 & 1
				\end{pmatrix}+\delta_{i}(x)\begin{pmatrix}
					1  & 1 & 1\\
					1 & 1 & 1 \\
					1 & 1 & 1
				\end{pmatrix}-I,\label{c0}
			\end{align}
			where functions $\alpha_{i}(x), \beta_{i}(x), \gamma_{i}(x),\delta_{i}(x), i=1,2$ are real-valued. Moreover, they satisfy
			\begin{align*}
				&|\alpha_{i}(x)|\le (1+|x|)g_{i}(x),\\
				&|\beta_{i}(x)|+|\gamma_{i}(x)|\le (1+|x|)^{2}g_{i}(x),\\
				&|\delta_{i}(x)-\frac{1}{3}| \le (1+|x|)^{2}g_{i}(x),
			\end{align*}
			for $x\in\mathbb{R}$ and $i=1,2$.
		\end{enumerate}
	\end{proposition}

	\begin{proof}
		Let us define $\tilde{Y}(x,k)=P(k)Y_{1}(x,k)$ that satisfies
		\begin{equation}\label{tilde x}
			\tilde{Y}(x, k)=P(k)-\int_{x}^{\infty} P(k) e^{\left(x-s\right) \widehat{\mathcal{L}}}\left(P(k)^{-1} \tilde{\mathrm{U}}\left(s\right) \tilde{Y}\left(s, k\right)\right) ds,
		\end{equation}
		for $x \in \mathbb{R}$ and $ k \in \Xi \backslash\{0\}$, where the matrix $\tilde{U}$ is of the form
		\begin{equation}\label{tilde u}
			\tilde{U}:=PUP^{-1}=\begin{pmatrix}
				-\frac{1}{3}u_{0} & 0 & 0 \\
				-v_{0} & \frac{2}{3}u_{0} & 0\\
				0 & 0 & -\frac{1}{3}u_{0}
			\end{pmatrix}.
		\end{equation}
		\par
		Firstly, we prove that the first two lines in $\tilde{Y}(x, 0)$ are zero. Assume that $\mathcal{P}\left(x, s, k\right):=P(k) e^{\left(x-s\right) \mathcal{L}(k)} P(k)^{-1}$. By computing the behavior at the zero, we can get $k=0$ is a removable singularity of $\mathcal{P}$. In other word, $\mathcal{P}$ is analytic at $k=0$. Now the kernel of Volterra integral function $\tilde{Y}(x, k)$ loses the singularity at $k=0$. Thus following the same procedure in the Proposition \ref{property of x}, $\tilde{Y}(x, k)$ is analytic at $k=0$. Rewrite the Taylor expansion of $\tilde{Y}(x, k)$ at $k=0$ as
		\begin{equation*}
			\tilde{Y}(x, k)=\tilde{Y}(x, 0)+\partial_{k} \tilde{Y}(x, 0) k+\frac{1}{2} \partial_{k}^{2} \tilde{Y}(x, 0) k^{2}+\cdots,\quad k\to0,\quad k \in \Xi,
		\end{equation*}
		in which the coefficients are smooth of $x$ in the real domain. As $x$ approaches $+\infty$, the derivatives $\partial_{k}^{j} \tilde{Y}(x, k)$ converge rapidly to $\partial_{k}^{j} P(k)$ for each $j \geq 0$ and $k \in \Xi$. Let us analyze the behavior of $\partial_{k}^{j} \tilde{Y}(x, 0)~ (j=0,1,2)$, as $x \rightarrow-\infty$. It can be calculated that
		\begin{equation*}
			\mathcal{P}\left(x, s, 0\right)=\begin{pmatrix}
				1 & s-x & 0\\
				0 & 1 & 0\\
				s-x & \frac{(s-x)^2}{2} & 1
			\end{pmatrix}.
		\end{equation*}
		By means of the equation (\ref{tilde u}), it follows
		\begin{equation*}
			\mathcal{P}\left(x, s, 0\right) \tilde{U}\left(s\right)=\begin{pmatrix}
				-\frac{u_0(s)}{3}-v_0(s)(s-x) & \frac{2u_0(s)}{3}(s-x) & 0\\
				-v_0(s) & \frac{2}{3}u_0(s) & 0\\
				-\frac{u_0(s)}{3}(s-x)-\frac{v_0(s)}{2}(s-x)^2 & \frac{u_0(s)}{3}(x-s)^2 & -\frac{u_0(s)}{3}
			\end{pmatrix},
		\end{equation*}
		and further
		\begin{equation*}
			\mathcal{P}\left(x, s, 0\right) \tilde{U}\left(s\right)P(0)=
			\begin{pmatrix}
				0 & 0 & 0\\
				0 & 0 & 0\\
				-\frac{u_0(s)}{3} & -\frac{u_0(s)}{3} & -\frac{u_0(s)}{3}
			\end{pmatrix}.
		\end{equation*}
		\par
		Therefore, we have proved that the first two lines in $\tilde{Y}(x, 0)$ are zero by analyzing the volterra equation (\ref{tilde x}) and found the behavior of $\tilde{Y}(x,0)$ as $x\to -\infty$, i.e.,
		\begin{equation*}
			\tilde{Y}(x,0)=\begin{pmatrix}
				O(x) & O(x) & O(x)\\
				O(1) & O(1) & O(1) \\
				O(x^2) & O(x^2) & O(x^2)
			\end{pmatrix},\qquad \text{as}\quad x\to -\infty.
		\end{equation*}
		\par
		Similarly, it is estimated that as $x\to-\infty$
		\begin{equation*}
			\partial_{k}\tilde{Y}(x,0)=\begin{pmatrix}
				O(x^2) & O(x^2) & O(x^2)\\
				O(x) & O(x) & O(x) \\
				O(x^3) & O(x^3) & O(x^3)
			\end{pmatrix},\quad \partial^2_{k}\tilde{Y}(x,0)=\begin{pmatrix}
				O(x^3) & O(x^3) & O(x^3)\\
				O(x^2) & O(x^2) & O(x^2) \\
				O(x^4) & O(x^4) & O(x^4)
			\end{pmatrix}.
		\end{equation*}
		\par
		Next, we prove that $Y_{1}(x,k)$ has only first-order pole at $k=0$. To do so, it is obvious that $P(k)^{-1}$ has the form
		\begin{equation*}\label{p-1}
			P(k)^{-1}=\frac{P^{(-2)}}{k^2}+\frac{P^{(-1)}}{k}+P^{(0)},
		\end{equation*}
		where
		\begin{equation*}
			P^{(-2)}=\frac{1}{3}\begin{pmatrix}
				0 & \omega & 0\\
				0 & \omega^2 & 0\\
				0 & 1 & 0
			\end{pmatrix},\quad
			P^{(-1)}=\frac{1}{3}\begin{pmatrix}
				\omega^2 & 0 & 0\\
				\omega & 0 & 0\\
				1 & 0 & 0
			\end{pmatrix},\quad
			P^{(0)}=\frac{1}{3}\begin{pmatrix}
				0 & 0 & 1\\
				0 & 0 & 1\\
				0 & 0 & 1\\
			\end{pmatrix},
		\end{equation*}
		which indicates that $Y_{1}(x,k)=P^{-1}(k)\tilde{Y}(x,k)$ possesses at most one double pole located at $k=0$. So set
		\begin{equation*}
			Y_{1}(x,k)=\frac{C_{1}^{(-2)}(x)}{k^2}+\frac{C_{1}^{(-1)}(x)}{k}+I+C_{1}^{(0)}(x)+C_{1}^{(1)}(x) k+\cdots, \quad    k \in \Xi,
		\end{equation*}
		as $k\to 0$. Direct calculations show that
		\begin{equation*}
			\begin{aligned}
				C_{1}^{(-2)}(x)&=P^{(-2)} \tilde{Y}(x, 0)=0,\\
				C_{1}^{(-1)}(x)&=P^{(-2)} \partial_{k} \tilde{Y}(x, 0)+P^{(-1)} \tilde{Y}(x, 0)=P^{(-2)} \partial_{k} \tilde{Y}(x, 0),\\
				C_{1}^{(0)}(x)&=\frac{1}{2}P^{(-2)} \partial_{k}^2\tilde{Y}(x, 0)+P^{(-1)} \partial_{k}\tilde{Y}(x, 0)+P^{(0)} \tilde{Y}(x, 0)-I.
			\end{aligned}
		\end{equation*}
		Assume that
		\begin{equation*}
			\partial_{k} \tilde{Y}(x, 0)=\begin{pmatrix}
				a_{11}(x) & a_{12}(x) & a_{13}(x)\\
				a_{21}(x) & a_{22}(x) & a_{23}(x)\\
				a_{31}(x) & a_{32}(x) & a_{33}(x)
			\end{pmatrix},
		\end{equation*}
		where the functions $a_{ij}(x)~(i,j=1,2,3)$ are complex.
		Due to the special form of matrix $P^{(-2)}$, it follows
		\begin{equation*}
			C_{1}^{(-1)}(x)=\frac{1}{3}\begin{pmatrix}
				\omega a_{21} & \omega a_{22} & \omega a_{23}\\
				\omega^{2} a_{21} & \omega^{2} a_{22} & \omega^{2} a_{23}\\
				a_{21} & a_{22} & a_{23}
			\end{pmatrix}.
		\end{equation*}
		The symmetries in (\ref{symmetry-X}) indicate that $a_{21}=\omega^{2}a_{22}=\omega a_{23}$ and $a_{23}=\bar{a}_{23}$. Thus we have completed the proof of (\ref{c-1}) with $\alpha_{1}(x)=\frac{1}{3}a_{23}$. The proof of $C_{1}^{(0)}$ in (\ref{c0}) and the asymptotics for $Y_2$ can also be given similarly.
		
	\end{proof}

	\subsection{The scaterring data $s(k)$}
	The functions $Y_{1}$ and $Y_{2}$ are linearly dependent as both of them are the solutions of the equation (\ref{lax-x}). Therefore, for $u_{0},v_{0}\in\mathcal{S}(\mathbb{R})$, one can set
	\begin{equation}\label{s(k) definition}
		Y_{1}(x,k)=Y_{2}(x,k)e^{x\widehat{\mathcal{L}(k)}}s(k),
	\end{equation}
	for $x\in\mathbb{R}$ and $k\in\mathbb{C}\setminus \{0\}$.
	From the properties of the Volterra integral equations (\ref{Volterra-x}) at $x=-\infty$, it follows
	\begin{equation}\label{s(k)-integral}
		s(k)=I-\int_{-\infty}^{\infty} e^{-x \widehat{\mathcal{L}(x)}}\left(U Y_{1}\right)\left(x, k\right) dx.
	\end{equation}
	Recalling the Proposition \ref{property of x}, the following proposition can be given at once.
	\begin{proposition}\label{property of sk}
		Amussing $u_{0},v_{0}\in\mathcal{S}(\mathbb{R})$,
		the scaterring matrix $s(k)$ satisfies
		\begin{enumerate}[$(1)$]
			\item The matrix $s(k)$ is well-defined and continuous on
			\begin{equation}\label{s(k)-area}
				k\in\begin{pmatrix}
					\omega^2\bar{S} & \mathbb{R}^- & \omega\mathbb{R}^-\\
					\mathbb{R}^- & \omega\bar{S} & \omega^2\mathbb{R}^-\\
					\omega\mathbb{R}^- & \omega^2\mathbb{R}^- & \bar{S}
				\end{pmatrix} \backslash\{0\}.
			\end{equation}
			For example, the entry $s_{11}$ is well-defined and continuous for $k \in \omega^2\bar{S} \backslash\{0\}$.
		\end{enumerate}
		\begin{enumerate}[$(2)$]
			\item The spectral function $s(k)$ admits the symmetries
			\begin{equation}\label{s(k)-symmetry}
				s(k)=\mathcal{A} s(\omega k) \mathcal{A}^{-1}=\mathcal{B} \overline{s(\bar{k})} \mathcal{B}.
			\end{equation}
		\end{enumerate}
		\begin{enumerate}[$(3)$]
			\item The arbitrary-order derivative $\partial_{k}^{j} s(k)~(j>0)$ is continuous for $k$ in (\ref{s(k)-area}).
		\end{enumerate}
		\begin{enumerate}[$(4)$]
			\item For $k$ in (\ref{s(k)-area}) and  $k$ tends to infinity, $s(k)$ converges to the identity matrix. Moreover, for any integer $N\ge1$ and $k$ in (\ref{s(k)-area}), we have
			\begin{equation}
				\left|\partial_{k}^{j}\left(s(k)-I-\sum_{j=1}^{N} \frac{s_{j}}{k^{j}}\right)\right|=O\left(k^{-N-1}\right),\quad k\to\infty,\quad  j=0,1,2,\cdots,N,
			\end{equation}
			where $\{s_{j}\}_{1}^{\infty}$ are diagonal matrices.
		\end{enumerate}

		\begin{enumerate}[$(5)$]
			\item For $k$ in (\ref{s(k)-area}) and $k\to0$, it is seen that
			\begin{equation}
				s(k)=\frac{s^{(-1)}}{k}+s^{(0)}+s^{(1)} k+\ldots,
			\end{equation}
			where
			\begin{equation}\label{s(-1)}
				s^{(-1)}=\mathrm{s}^{(-1)}\left(\begin{array}{ccc}
					\omega^2  & 1  & \omega \\
					1 & \omega&\omega^2 \\
					\omega& \omega^2 & 1
				\end{array}\right), \quad \mathrm{s}^{(-1)}:=-\int_{\mathbb{R}}(\frac{2}{3}\alpha_1(x)u_0-v_{0}\gamma_{1}(x))dx.
			\end{equation}
		\end{enumerate}
		
		\begin{enumerate}[$(6)$]
			\item If $u_{0}(x),v_{0}(x)$ are compactly supported, the function $s(k)$ is analytic for $k \in \mathbb{C} \backslash\{0\}$ and $\det{s}=1$.
		\end{enumerate}
	\end{proposition}
	\begin{proof}
		Property (1) can be obtained by direct calculation where it is ensured that the exponential factors are bounded. The properties (2) and (3) are obvious from the properties of $Y_{1}$ in Proposition \ref{property of x}. Recalling that $Y_{1}$ has good properties in Proposition \ref{kto-8}, for $k$ in (\ref{s(k)-area}), we have
		\begin{equation*} s(k)=I-\sum_{j=0}^{p}\frac{1}{k^{j}}\int_{\mathbb{R}}e^{-x\widehat{\mathcal{L}(k)}}\big(U(x,k)Y_{1}^{(p)}(x)\big)dx+O(k^{-p-1}),
		\end{equation*}
		as $k\to\infty$. Note that $U(\cdot,k)\in\mathcal{S}(\mathbb{R})$ and $Y_{1}^{(p)}(x)$ and their derivatives are bounded. Then the integration by parts applied to the off-diagonal elements of $s(k)$ shows that  $\{s_{j}\}_{1}^{\infty}$ are diagonal matrices as $k\to \infty$. With $\tilde{U}$ and $\tilde{Y}$ defined in the proof of Proposition \ref{k to 0}, rewrite $s(k)$ as
		\begin{equation*}
			\begin{aligned}
				s(k)&=I-\int_{-\infty}^{\infty} e^{-x \widehat{\mathcal{L}(k)}}\big(P(k)^{-1}\tilde{U} (x)  \tilde{Y}(x, k)\big) d x\\
				&=I-\int_{-\infty}^{\infty}e^{-x \widehat{\mathcal{L}(k)}}\big[(\frac{P^{(-2)}}{k^2}+\frac{P^{(-1)}}{k}+P^{(0)})\tilde{U}\left(\tilde{Y}(x, 0)+\partial_{k} \tilde{Y}(x, 0) k+\cdots\right)\big] dx\\
				&=I-\frac{1}{k^2}\int_{-\infty}^{\infty}e^{-x \widehat{\mathcal{L}(k)}}P^{(-2)}\tilde{U}\tilde{Y}(x, 0) dx\\
				&-\frac{1}{k}\int_{-\infty}^{\infty}e^{-x \widehat{\mathcal{L}(k)}}[P^{(-2)}\tilde{U}\partial_{k}\tilde{Y}(x, 0)+P^{(-1)}\tilde{U}\tilde{Y}(x, 0)] dx+\cdots\\
				&=I-\frac{1}{k}\int_{-\infty}^{\infty}e^{-x \widehat{\mathcal{L}(k)}}[P^{(-2)}\tilde{U}\partial_{k}\tilde{Y}(x, 0)] dx+\cdots
				\\
				&=\frac{s^{(-1)}}{k}+s^{(0)}+s^{(1)}k+\cdots,
			\end{aligned}
		\end{equation*}
		where we have used the fact that only the last line in the matrix $\tilde{Y}(x,0)$ is not zero, and the coefficient $s^{(-1)}$ is
		\begin{equation}\label{s-1-s-1}
			s^{(-1)}=-\int_{-\infty}^{\infty}e^{-x \widehat{\mathcal{L}(k)}}[P^{(-2)}\tilde{U}\partial_{k}\tilde{Y}(x, 0)] dx.
		\end{equation}
		In fact, the coefficient $s^{(-1)}$ takes the following form
		\begin{equation}\label{s-1}
			s^{(-1)}=-\int_{-\infty}^{\infty}P^{(-2)}\tilde{U}\partial_{k}\tilde{Y}(x, 0) dx.
		\end{equation}
		The diagonal elements in the right-hand side of (\ref{s-1-s-1}) is the same as the right-hand side of (\ref{s-1}). Next it needs to show that off-diagonal elements in (\ref{s-1-s-1}) and (\ref{s-1}) are also equal. For simplicity, only the $(12)$ entry of $s^{(-1)}$ in (\ref{s-1-s-1}) and (\ref{s-1}) is verified to the same (the other off-diagonal elements can be estimated similarly).
		As $l_1(k)-l_2(k)=-\sqrt{3}\mathrm{i}k$, it follows
		\begin{equation*}
			\big|\int_{\mathbb{R}}(e^{-x(l_1-l_2)}-1)f(x)dx\big| \le \int_{\mathbb{R}}|\sqrt{3}\mathrm{i}kxf(x)|dx=O(k), \qquad k\to 0,
		\end{equation*}
		for any $f\in\mathcal{S}(\mathbb{R})$. In this sense, one can replace $e^{-x(l_1-l_2)}$ with $1$. Then the calculation of (\ref{s(-1)}) is direct.
		\par
		Providing $u_0(x),v_0(x)$ are compactly supported, the integral in (\ref{s(k)-integral}) is convergent for $k\in\mathbb{C}\setminus\{0\}$. Therefore, $s(k)$ is well-defined and analytic for $k\in\mathbb{C}\setminus\{0\}$. Furthermore, as $\det{Y_{1}}=\det{Y_{2}}=1$, we have $\det{s}=1$ from (\ref{s(k) definition}).
	\end{proof}

	\subsection{The adjoint eigenfunctions $Y_1^{A}(x,k)$ and $Y_2^{A}(x,k)$}

	For a third order matrix $V$ which has a unit determinant, define its cofactor matrix $V^{A}=(V^{-1})^{T}$ in the following form
	\begin{equation*}
		V^A=\begin{pmatrix}
			m_{11} (V) & -m_{12}(V) & m_{13}(V) \\
			-m_{12}(V)  & m_{22}(V)  & -m_{23}(V)  \\
			m_{31}(V)  & m_{32}(V)  & m_{33}(V)
		\end{pmatrix},
	\end{equation*}
	where $m_{ij}(V)$ represents the $(ij)$th minor of the matrix $V$.
	\par
	Supposing $u_0,v_0\in\mathcal{S}(\mathbb{R})$ have compact support, the entries of matrices $Y_{1}(x,k),Y_{2}(x,k)$ and $s(k)$ will be defined for $k\in\mathbb{C}\setminus\{0\}$. From the identity $(Y_{1}^{A})_{x}=-Y_{1}^{A}(Y_{1}^{T})_{x}Y_{1}^{A}$ and (\ref{lax-x}), it's not difficult to find that
	\begin{equation}\label{xa-ode}
		(Y^{A})_{x}+[\mathcal{L},Y^{A}]=-U^{T}Y^{A}.
	\end{equation}
	Since $Y_{1}^{A}\to I$ as $x\to+\infty$ and $Y_{2}^{A}\to I$ as $x\to-\infty$, it is observed that $Y_{1}^{A}$ and $Y_{2}^{A}$ satisfy
	\begin{subequations}\label{xa-integral}
		\begin{align}
			Y^{A}_{1}(x,k)=I+\int_{x}^{\infty}e^{-(x-s)\widehat{\mathcal{L}(k)}}(U^{T}Y_{1}^{A})(s,k)ds,\\
			Y^{A}_{2}(x,k)=I-\int^{x}_{-\infty}e^{-(x-s)\widehat{\mathcal{L}(k)}}(U^{T}Y_{1}^{A})(s,k)ds.
		\end{align}
	\end{subequations}
	The next proposition on $Y_{1}^{A}(x,k)$ and $Y_{2}^{A}(x,k)$ is similar to Proposition \ref{property of x}.
	
	\begin{proposition}\label{property of xa}
		Assuming $u_{0},v_{0}\in\mathcal{S} \left(\mathbb{R}\right)$, the matrix-valued solutions $Y_{1}^{A}(x,k)$ and $Y_{2}^{A}(x,k)$ of (\ref{xa-ode}) satisfy the characteristics as follow:
		\begin{enumerate}[$(1)$]
			\item The function $Y_{1}^{A}(x,k)$ is well-defined for $x\in \mathbb{R}$ and for $k$ belonging to the region $\hat{\Xi}$ excluding zero. Additionally,
			$Y_{1}^{A}(\cdot,k)$ is smooth for any $k$ within this region excluding zero and satisfies (\ref{xa-ode}).
		\end{enumerate}
		\begin{enumerate}[$(2)$]
			\item The function $Y_{2}^{A}(x,k)$ is well-defined for $x\in \mathbb{R}$ and for $k$ belonging to the region $\Xi$ excluding zero. Additionally, $Y_{2}^{A}(\cdot,k)$ is smooth for any $k$ within this region excluding zero and satisfies (\ref{xa-ode}).
		\end{enumerate}
		\begin{enumerate}[$(3)$]
			\item For given $x\in\mathbb{R}$, $Y_{1}^{A}(x,\cdot)$ is continuous and analytic for any $k\in\hat{\Xi}\setminus\{0\}$.
		\end{enumerate}
		\begin{enumerate}[$(4)$]
			\item For given $x\in\mathbb{R}$, $Y_{2}^{A}(x,\cdot)$ is continuous and analytic for any $k\in\Xi\setminus\{0\}$.
		\end{enumerate}
		\begin{enumerate}[$(5)$]
			\item The partial derivative $\frac{\partial^{i}Y_{1}^{A}}{\partial k^{i}}(x,\cdot)$ can extend continuously to $k\in\hat{\Xi}\setminus\{0\}$ for any $x\in\mathbb{R}$ and each integer $i\ge1$.
		\end{enumerate}
		\begin{enumerate}[$(6)$]
			\item The partial derivative $\frac{\partial^{i}Y_{2}^{A}}{\partial k^{i}}(x,\cdot)$ can extend continuously to $k\in\Xi\setminus\{0\}$ for any $x\in\mathbb{R}$ and each integer $i\ge1$.
		\end{enumerate}
		\begin{enumerate}[$(7)$]
			\item For any integer $n\ge1$, there exist two bounded, smooth, positive functions $g_{1}(x)$ and $g_{2}(x)$ defined on  $x\in\mathbb{R}$, which decay rapidly as $x$ approaches ${\infty}$ and ${-\infty}$, respectively, such that
			\begin{subequations}
				\begin{align}
					&\left|\frac{\partial^{i}}{\partial k^{i}}(Y_{1}^{A}(x,k)-I)\right|\le g_{1}(x),\quad k\in\hat{\Xi}\setminus\{0\},\quad  &i=0,1,2\cdots,n,\\
					&\left|\frac{\partial^{i}}{\partial k^{i}}(Y_{2}^{A}(x,k)-I)\right|\le g_{2}(x),\quad k\in\Xi\setminus\{0\},\quad  &i=0,1,2\cdots,n.
				\end{align}
			\end{subequations}
		\end{enumerate}
		\begin{enumerate}[$(8)$]
			\item For any $x\in\mathbb{R}$, $Y_{1}^{A}$ and $Y_{2}^{A}$ also satisfy the $\mathbb{Z}_{2}$ and $\mathbb{Z}_{3}$ symmetries
			\begin{subequations}\label{symmetry-XA}
				\begin{align}
					&Y_{1}^{A}(x,k)=\mathcal{A}Y_{1}^{A}(x,\omega k)\mathcal{A}^{-1}=\mathcal{B}\overline{Y_{1}^{A}(x,\bar{k})}\mathcal{B},\quad k\in\hat{\Xi}\setminus\{0\},\\
					&Y_{2}^{A}(x,k)=\mathcal{A}Y_{2}^{A}(x,\omega k)\mathcal{A}^{-1}=\mathcal{B}\overline{Y_{2}^{A}(x,\bar{k})}\mathcal{B},\quad k\in\Xi\setminus\{0\}.
				\end{align}
			\end{subequations}
		\end{enumerate}
		\begin{enumerate}[$(9)$]
			\item Suppose $u_{0},v_{0}$ are supported compactly. For every $x\in\mathbb{R}$ and $k\in\mathbb{C}\setminus \{0\}$, both $Y_{1}^{A}(x,k)$ and $Y_{2}^{A}(x,k)$ are analytic, and their determinants are equal to  $1$ : $\det{Y_{1}^{A}(x,k)}=\det{Y_{2}^{A}(x,k)}=1$.
		\end{enumerate}
	\end{proposition}

	\begin{proposition}\label{XA as k to 8}
		Assuming $u_0,v_0\in\mathcal{S}(\mathbb{R})$ and $Y_{1,p}^{A}$ and $Y_{2,p}^{A}$ are the cofactor matrices of $Y_{1,p}$ and $Y_{2,p}$
		in (\ref{X,P}), $Y_{1}^{A}$ and $Y_{2}^{A}$ will converge to $Y_{1,p}^{A}$ and $Y_{2,p}^{A}$ for any  positive integer $p$. Furthermore, there exist two bounded, smooth, positive functions $g_{1}(x)$ and $g_{2}(x)$ of $x\in\mathbb{R}$ which rapidly as $x$ approaches ${\infty}$ and ${-\infty}$, respectively, such that
		\begin{subequations}
			\begin{align}
				&\left|\frac{\partial^{i}}{\partial k^{i}}\left(Y_{1}^{A}-Y_{1,p}^{A}\right)\right| \leq \frac{g_{1}(x)}{|k|^{p+1}}, \quad  k \in\hat{\Xi},\quad |k| \geq 2, \\
				&\left|\frac{\partial^{i}}{\partial k^{i}}\left(Y_{2}^{A}-Y_{2,p}^{A}\right)\right| \leq \frac{g_{2}(x)}{|k|^{p+1}}, \quad  k \in\Xi,\quad |k| \geq 2.
			\end{align}
		\end{subequations}
	\end{proposition}

	\begin{proposition}\label{XA as k to 0}
		Let $u_{0},v_{0}\in\mathcal{S}(\mathbb{R})$ and integer $p\ge0$. There are third-order matrix-valued functions $D_{i}^{(l)}(x)~(i=1,2,l=-1,0,\dots ,p)$, which satisfy
		\begin{enumerate}[$(1)$]
			\item For each $x\in\mathbb{R}$, there exist smooth positive functions $g_{1}(x)$ and $g_{2}(x)$  which rapidly as $x$ approaches ${\infty}$ and ${-\infty}$, respectively, such that
			\begin{subequations}
				\begin{align}
					&\left| \frac{\partial^{j}}{\partial k^{j}}(Y_{1}^{A}-I-\sum_{l=-1}^{p}D_{1}^{(l)}(x)k^l)\right| \le g_{1}(x)|k|^{p+1-j},\quad k \in\hat{\Xi},\\
					&\left| \frac{\partial^{j}}{\partial k^{j}}(Y_{2}^{A}-I-\sum_{l=-1}^{p}D_{2}^{(l)}(x)k^l)\right| \le g_{2}(x)|k|^{p+1-j},\quad k \in \Xi,
				\end{align}
			\end{subequations}
			where $j\ge0$ is an integer and $|k|\le\frac{1}{2}$.
		\end{enumerate}
		\begin{enumerate}[$(2)$]
			\item The functions $D_{1}^{(l)}$ and $D_{2}^{(l)}$ are smooth and decay rapidly as $|x|\to{\infty}$ for $l=-1,0,1,2,\cdots$.
		\end{enumerate}
		\begin{enumerate}[$(3)$]
			\item The leading coefficients are given by
			\begin{align}
				D_{i}^{(-1)}(x)=&\tilde{\alpha}_{i}\begin{pmatrix}
					\omega^2 & 1 & \omega \\
					1 & \omega & \omega^2 \\
					\omega & \omega^2 & 1\\
				\end{pmatrix},\label{d-1}	\\
				D_{i}^{(0)}(x)=&\tilde{\delta}_{i}\begin{pmatrix}
					1 & 1 & 1 \\
					1 & 1 & 1 \\
					1 & 1 & 1 \\
				\end{pmatrix}+\tilde{\gamma}_{i}\begin{pmatrix}
					1  & \omega^2 & \omega\\
					\omega & 1 & \omega^2\\
					\omega^2 & \omega & 1\\
				\end{pmatrix}+\tilde{\beta}_{i}\begin{pmatrix}
					1 & \omega & \omega^2\\
					\omega^2 & 1 & \omega \\
					\omega & \omega^2 & 1 \\
				\end{pmatrix}-I,\label{c0}
			\end{align}
			where functions $\tilde\alpha_{i}(x), \tilde\beta_{i}(x), \tilde\gamma_{i}(x),\tilde\delta_{i}(x)~ (i=1,2)$ are real-valued. Moreover, they satisfy
			\begin{align*}
				&|\tilde{\alpha}_{i}(x)|\le (1+|x|)g_{i}(x),\\
				&|\tilde{\gamma}_{i}(x)|+|\tilde{\beta}_{i}(x)|\le (1+|x|)^{2}g_{i}(x),\\
				&|\tilde{\delta}_{i}(x)-\frac{1}{3}| \le (1+|x|)^{2}g_{i}(x),
			\end{align*}
			for $x\in\mathbb{R}$ and $i=1,2$.
		\end{enumerate}
	\end{proposition}

	\begin{proof}
		Define $\tilde{Y}_{1}^A:=k^2P^AY_{1}^A$ for $x\in\mathbb{R}$ and $k \in\hat{\Xi}$, which satisfies
		\begin{equation*}
			\tilde{Y}_{1}^A(x,k)=k^2P^A(k)+\int_{x}^{\infty}P^Ae^{-(x-s)\widehat{\mathcal{L}(k)}}(P(k)^T\tilde{U}(s)^T\tilde{Y}_{1}^A(s,k))ds.
		\end{equation*}
		The following proof bears similarity to that in Proposition \ref{k to 0}.

	\end{proof}

	Since both $Y_{1}^{A}$ and $Y_{2}^{A}$ are the solutions of (\ref{xa-ode}), there is a spectral function $s^{A}(k)$ such that
	\begin{equation*}
		Y^{A}_{1}(x,k)=Y_{2}^{A}(x,k)e^{-x\widehat{\mathcal{L}(k)}}s^{A}(k),
	\end{equation*}
	for all $x\in \mathbb{R}$.
	Reminding (\ref{xa-integral}) and letting $x\to -\infty$, it follows
	\begin{equation}\label{sa-integral}
		s^{A}(k)=I+\int_{-\infty}^{\infty} e^{x \widehat{\mathcal{L}(k)}}\left(U^{T} Y_{1}^{A}\right)\left(x, k\right) dx.
	\end{equation}

	\begin{proposition}
		Assuming $u_{0},v_{0}\in\mathcal{S}(\mathbb{R})$,
		the scattering matrix $s^{A}(k)$ satisfies
		\begin{enumerate}[$(1)$]
			\item The matrix $s^{A}(k)$ is continuous on
			\begin{equation}\label{sA(k)-area}
				k\in\begin{pmatrix}
					-\omega^2\bar{S} & \mathbb{R}^+ & \omega\mathbb{R}^+\\
					\mathbb{R}^+ & -\omega\bar{S} & \omega^2\mathbb{R}^+\\
					\omega\mathbb{R}^+ & \omega^2\mathbb{R}^+ & -\bar{S}
				\end{pmatrix} \backslash\{0\}.
			\end{equation}
			It means that $s_{11}^{A}$  is continuous for $k \in -\omega^2\bar{S} \backslash\{0\}$, etc.
		\end{enumerate}
		\begin{enumerate}[$(2)$]
			\item The spectral function $s^{A}(k)$ allows the symmetries
			\begin{equation}\label{sA(k)-symmetry}
				s^{A}(k)=\mathcal{A} s^{A}(\omega k) \mathcal{A}^{-1}=\mathcal{B} \overline{s^{A}(\bar{k})} \mathcal{B}.
			\end{equation}
		\end{enumerate}
		\begin{enumerate}[$(3)$]
			\item The arbitrary-order derivative $\partial_{k}^{j} s^{A}(k)~(j>0)$ is well-defined and continuous for $k$ in (\ref{sA(k)-area}).
		\end{enumerate}
		
		\begin{enumerate}[$(4)$]
			\item For $k$ in (\ref{sA(k)-area}) and $k$ tends to infinity, $s^{A}(k)$ converges to the identity matrix. Moreover, for any integer $N\ge1$ and $k$ in (\ref{sA(k)-area}), we have
			\begin{equation}
				\left|\partial_{k}^{j}\left(s^{A}(k)-I-\sum_{j=1}^{N} \frac{s^{A}_{j}}{k^{j}}\right)\right|=O\left(k^{-N-1}\right),\quad k\to\infty,\quad  j=0,1,2,\cdots,N,
			\end{equation}
			where $\{s^{A}_{j}\}_{1}^{\infty}$ are diagonal matrices.
		\end{enumerate}
		
		\begin{enumerate}[$(5)$]
			\item For $k$ in (\ref{sA(k)-area}) and $k\to0$, it is seen that
			\begin{equation}
				s^{A}(k)=\frac{s^{A(-1)}}{k}+s^{A(0)}+s^{A(1)} k+\ldots,
			\end{equation}
			where
			\begin{equation}\label{sA(-1)}
				s^{A(-1)}=\mathrm{s}^{A(-1)}\left(\begin{array}{ccc}
					\omega^2  & 1 &\omega \\
					1 & \omega & \omega^2 \\
					\omega  &\omega^2 &1
				\end{array}\right), \quad \mathrm{s}^{A(-1)}:=-\int_{\mathbb{R}}(v_{0}\tilde\beta_{1}(x)+\frac{1}{3}\tilde\alpha_1(x)u_0)dx.
			\end{equation}
		\end{enumerate}
		
		\begin{enumerate}[$(6)$]
			\item If $u_{0}(x),v_{0}(x)$ are compactly supported, the function $s^{A}(k)$ is analytic for $k \in \mathbb{C} \backslash\{0\}$ and $\det{s^{A}}=1$.
		\end{enumerate}
		
	\end{proposition}
	
	With the definitions of scattering matrices $s(k)$ and $s^{A}(k)$ in mind, it's time to define the reflection coefficients $\rho_{1}(k)$ and $\rho_{2}(k)$, which are of the following forms
	\begin{equation}\label{r1r2}
		\begin{cases}
			\rho_{1}(k)=\frac{s_{12}(k)}{s_{11}(k)},\qquad & k\in(-\infty,0),\\
			\rho_{2}(k)=\frac{s_{12}^{A}(k)}{s_{11}^{A}(k)},\qquad & k\in(0,\infty).
		\end{cases}
	\end{equation}
	
	The properties of $\rho_{1}$ and $\rho_{2}$	are stated in the theorem below.
	
	\begin{theorem}[Properties of $\rho_{1}$ and $\rho_{2}$]
		Assume  $u_0,v_0\in \mathcal{S}(\mathbb{R})$. Furthermore, suppose that $s_{11}(k)$ and $s_{11}^{A}(k)$ are nonzero respectively for $k\in\bar{\Omega}_{4}\setminus \{0\}$ and $k\in\bar{\Omega}_{1}\setminus \{0\}$ and
		\begin{equation*}
			\lim\limits_{k\to 0}k \cdot s_{11}(k)\neq 0,\quad 	\lim\limits_{k\to 0}k \cdot s_{11}^{A}(k)\neq 0.
		\end{equation*}
		Therefore, the reflection coefficients $\rho_{1}$ and $\rho_{2}$ have the properties as follow:
		\begin{enumerate}[$(1)$]
			\item $\rho_{1}$ and $\rho_{2}$ are smooth for $k\in(-\infty ,0)$ and $k\in(0,\infty)$, respectively.
		\end{enumerate}
		
		\begin{enumerate}[$(2)$]
			\item There exist power series expansions
			\begin{subequations}
				\begin{align}
					\rho_{1}(k)=\rho_{1}(0)+\rho_{1}^{\prime}(0)k+\frac{1}{2}\rho_{1}^{\prime \prime}(0)k^{2}+\cdots,\quad k\to 0,\quad k<0,\\
					\rho_{2}(k)=\rho_{2}(0)+\rho_{2}^{\prime}(0)k+\frac{1}{2}\rho_{2}^{\prime \prime}(0)k^{2}+\cdots,\quad k\to 0,\quad k>0,
				\end{align}
			\end{subequations}
			where
			\begin{equation*}
				\rho_{1}(0)= \rho_{2}(0)=\omega .
			\end{equation*}

		\end{enumerate}
		
	\end{theorem}

	\subsection{The Riemann-Hilbert problem}
	This subsection constructs the Riemann-Hilbert problem of the Hirota-Satsuma equation (\ref{HS eq}) based on inverse spectral theory. In the regions $\Omega_{n}~(n=1,2,\dots 6)$, denote $\mathcal{M}_{n}(x,k)$ as $\mathcal{M}(x,k)$. In fact, the eigenfunctions $\mathcal{M}_{n}(x, k)$ are three-order matrix-valued solutions of (\ref{lax-x}) defined by
	\begin{equation}\label{Mn}
		\left(\mathcal{M}_{n}\right)_{i j}(x, k)=\delta_{i j}+\int_{\gamma_{i j}^{n}}\left(e^{\left(x-s\right) \widehat{\mathcal{L}(k)}}\left(U \mathcal{M}_{n}\right)\left(s, k\right)\right)_{i j} d s, \quad i, j=1,2,3,
	\end{equation}
	where
	\begin{equation*}\label{gamma-ij}
		\gamma_{i j}^{n}=\left\{\begin{array}{ll}
			(-\infty, x), & \operatorname{Re} l_{i}(k)<\operatorname{Re} l_{j}(k), \\
			(+\infty, x), & \operatorname{Re} l_{i}(k) \geq \operatorname{Re} l_{j}(k),
		\end{array} \quad \text { for } \quad k \in \Omega_n\right.,\quad n=1,\dots,6.
	\end{equation*}
	We define $	\gamma_{i j}^{n}$ in this way in order to make sure that the integral in  (\ref{Mn}) is bounded for $k\in\Omega_{n}$ and $x^{\prime}\in	\gamma_{i j}^{n}$. From the equality (\ref{Mn}), one can extend $\mathcal{M}_{n}$ to the boundary of $\Omega_{n}$ continuously. Let us record the set of zeros of the Fredholm determinants related to (\ref{Mn}) as $Z$. In the following proposition, it will be shown that $\mathcal{M}_{n}$ are well-defined for $k\in\bar{\Omega}_{n}\setminus \mathcal{Z}$, where $\mathcal{Z}=Z\cup\{0\}$.
	
	\begin{proposition}\label{property of mn}
		Assuming the initial values $u_0,v_0\in\mathcal{S}(\mathbb{R})$, the matrix-valued functions $\{\mathcal{M}_{n}\}_{1}^{6}$ of (\ref{lax-x}) defined in (\ref{Mn}) satisfy:

		\begin{enumerate}[$(1)$]
			\item For  $x\in\mathbb{R}$ and $k\in\bar{\Omega}_{n}\setminus \mathcal{Z}$, $\mathcal{M}_{n}(x,k)$ is well-defined. Furthermore, $\mathcal{M}_{n}(\cdot,k)$ is smooth for any $k\in\bar{\Omega}_{n}\setminus \mathcal{Z}$ and solves the equation (\ref{Mn}).
		\end{enumerate}
		
		\begin{enumerate}[$(2)$]
			\item For every $x\in\mathbb{R}$, $\mathcal{M}_{n}(x,k)$ is analytic for $k\in\bar{\Omega}_{n}\setminus \mathcal{Z}$ and continuous for $k\in\Omega_{n}\setminus \mathcal{Z}$.
		\end{enumerate}
		
		\begin{enumerate}[$(3)$]
			\item For any $\epsilon > 0$ and $\text{dist}(k,\mathcal{Z})\ge \epsilon $, there exist a  constant $C(\epsilon)$ such that
			\begin{equation*}
				|\mathcal{M}_{n}(x,k)|\le C(\epsilon),
			\end{equation*}
			for $x\in\mathbb{R}$ and $k\in\bar{\Omega}_{n}$.
		\end{enumerate}
		
		\begin{enumerate}[$(4)$]
			\item The partial derivatives $\frac{\partial^{i}\mathcal{M}_{n}}{\partial k^{i}}(x,\cdot)~(i\geq1)$ can be extended continuously to $\bar{\Omega}_{n}\setminus \mathcal{Z}$ for any $x\in\mathbb{R}$ and $i=1,2,\dots$.
		\end{enumerate}
		
		\begin{enumerate}[$(5)$]
			\item For  $x\in\mathbb{R}$ and $k\in\bar{\Omega}_{n}\setminus \mathcal{Z}$, there always exists $\det{\mathcal{M}_{n}(x,k)}=1$.
		\end{enumerate}
		
		\begin{enumerate}[$(6)$]
			\item Recalling the definition of $\mathcal{M}(x,k)$ which equals $\mathcal{M}_{n}(x,k)$ for $k\in\Omega_{n}$, the matrix-valued function $\mathcal{M}(x,k)$ is sectionally analytic and satisfies
			\begin{equation}\label{symmetry of M}
				\mathcal{A}\mathcal{M}(x,\omega k)\mathcal{A}^{-1}=\mathcal{M}(x,k)=\mathcal{B}\overline{\mathcal{M}(x,\bar{k})}\mathcal{B}, \quad k\in\mathbb{C}\setminus \mathcal{Z}.
			\end{equation}
		\end{enumerate}

	\end{proposition}

	\begin{lemma}[Asymptotics of $\mathcal{M}(x,k)$ as $k\to \infty$]
		Assume $u_0,v_0\in\mathcal{S}(\mathbb{R})$ and $u_0,v_0 \not\equiv 0$. Reminding the definition of $Y_{1,p}$ appearing in (\ref{X,P}) for an integer $p\ge1$, for $x\in\mathbb{R}$ and $k\in\mathbb{C}\setminus\Gamma$, there are two constants $\delta >0$ and $C>0$ such that
		\begin{equation*}
			|\mathcal{M}(x,k)-Y_{1,p}(x,k)|\le \frac{C}{|k|^{p+1}}, \qquad |k|\ge\delta.
		\end{equation*}
	\end{lemma}

	Assume $u_0,v_0\in\mathcal{S}(\mathbb{R})$ are supported compactly. Since $Y_1, Y_2$ and $\mathcal{M}_n$ are all the solutions of the equation (\ref{lax-x}), there are matrices $S_{n}(k),T_{n}(k)~(n=1,2, \dots ,6)$ such that
	\begin{equation*}
		\begin{aligned}
			\mathcal{M}_n(x,k)&=Y_{2}(x,k)e^{x\widehat{\mathcal{L}(k)}}S_{n}(k)\\
			&=Y_{1}(x,k)e^{x\widehat{\mathcal{L}(k)}}T_{n}(k),   \qquad x\in\mathbb{R},k\in\bar{\Omega}_n\setminus\mathcal{Z}.
		\end{aligned}
	\end{equation*}
	Here $T_{n}(k)$ and $S_{n}(k)$ can be accurately expressed as follow:
	\begin{equation*}
		\begin{cases}
			T_n(k)=\lim\limits_{x\to\infty}e^{-x\widehat{\mathcal{L}(k)}}M_{n}(x,k),\\
            S_n(k)=\lim\limits_{x\to-\infty}e^{-x\widehat{\mathcal{L}(k)}}M_{n}(x,k),\\
			s(k)=S_{n}(k)T_{n}(k)^{-1}.
		\end{cases}\quad k\in\bar{\Omega}_{n}\setminus\{0\},
	\end{equation*}
	Now we only provide the matrices $S_1$ and $S_6$
	\begin{equation*}
		S_{1}(k)=\begin{pmatrix}
			\frac{1}{m_{11}(s)} & \frac{m_{21}(s)}{s_{33}} & s_{13}\\
			0 & \frac{m_{11}(s)}{s_{33}} & s_{23}\\
			0 & 0 & s_{33}
		\end{pmatrix},\quad
		S_{6}(k)=\begin{pmatrix}
			\frac{m_{22}(s)}{s_{33}}& 0 & s_{13}\\
			\frac{m_{12}(s)}{s_{33}} & \frac{1}{m_{22}(s)} & s_{23}\\
			0 & 0 & s_{33}
		\end{pmatrix}.
	\end{equation*}
	\par
	
	\begin{lemma}
		Assume $u_0,v_0\in\mathcal{S}(\mathbb{R})$, then the matrices  $\mathcal{M}_1$ and $\mathcal{M}_{1}^{A}$ can be expressed by the entries of $Y_1,Y_2,Y_1^{A},Y_2^A,s$ and $s^A$ as follows:
		\begin{equation}
			\mathcal{M}_1=\begin{pmatrix}
				\frac{[Y_{2}]_{11}}{s_{11}^{A}} & \frac{[Y_1^{A}]_{31}[Y_{2}^{A}]_{23}-[Y_{1}^{A}]_{21}[Y_{2}^{A}]_{33}}{s_{33}} & [Y_1]_{13}\\
				\frac{[Y_{2}]_{21}}{s_{11}^{A}} & \frac{[Y_{1}^{A}]_{11}[Y_{2}^{A}]_{33}-[Y_{1}^{A}]_{31}[Y_{2}^{A}]_{13}}{s_{33}} & [Y_1]_{23}\\
				\frac{[Y_{2}]_{31}}{s_{11}^{A}} & \frac{[Y_{1}^{A}]_{21}[Y_{2}^{A}]_{13}-[Y_{1}^{A}]_{11}[Y_{2}^{A}]_{23}}{s_{33}} & [Y_1]_{33}\\
			\end{pmatrix},\mathcal{M}_{1}^{A}=\begin{pmatrix}
				[Y_1^A]_{11} & \frac{[Y_{2}]_{31}[Y_{1}]_{23}-[Y_{2}]_{21}[Y_{1}]_{33}}{s_{11}^{A}} & \frac{[Y_{2}]_{13}}{s_{33}}\\
				[Y_1^A]_{21} & \frac{[Y_{1}]_{11}[Y_{2}]_{33}-[Y_{2}]_{31}[Y_{1}]_{13}}{s_{11}^{A}} & \frac{[Y_{2}]_{23}}{s_{33}}\\
				[Y_1^A]_{31} & \frac{[Y_{1}]_{21}[Y_{2}]_{13}-[Y_{2}]_{11}[Y_{1}]_{23}}{s_{11}^{A}} & \frac{[Y_{2}]_{33}}{s_{33}}\\
			\end{pmatrix}
		\end{equation}
		for $x\in\mathbb{R}$ and $k\in\Omega_{1}\setminus\mathcal{Z}$, 	where $[Y_{1}]_{11}$ denotes the $(11)$ entry of $Y_{1}$.
	\end{lemma}

	Thus the jump condition for $\mathcal{M}$ on the cut $\Gamma_1\setminus \{Z\}$ is formulated as
	\begin{equation}
		\mathcal{M}_{1}=\mathcal{M}_{6}\mathsf{V}_{1}, \quad k\in \Gamma_1\setminus \{Z\},
	\end{equation}
	where
	\begin{equation*}
		\mathsf{V}_{1}(x,0,k)=e^{x\widehat{\mathcal{L}(k)}}[S_{6}(k)^{-1}S_{1}(k)]=e^{x\widehat{\mathcal{L}(k)}}\begin{pmatrix}
			1-|\rho_2(k)|^2 & -\rho_2^*(k) & 0\\
			\rho_2(k) & 1    & 0\\
			0 &  0   & 1
		\end{pmatrix}.
	\end{equation*}
	\par	
	\par
	\par	
	Recalling the time evolutions in (\ref{lax-pde}) and the special forms of $\mathcal{L}$ and $\mathcal{Z}$, define
	\begin{equation}\label{theta ij}
		\Theta_{ij}(x,t,k)=(l_{i}-l_{j})x+(z_{i}-z_{j})t,
	\end{equation}
for $i,j=1,2,3.$
	\par
	So one can get the jump condition of the function $\mathcal{M}(x,t,k)$ on the jump contour $\Gamma$ in Figure \ref{Figure-jump}, which is stated in the proposition as follows.
	\begin{proposition}\label{jump proposition}
		Assuming  $u_0,v_0\in\mathcal{S}(\mathbb{R})$, the sectionally analytic matrix-valued function $\mathcal{M}(x,t,k)$ meets the jump condition
		\begin{equation}\label{jump condition for M}
			\mathcal{M}_{+}(x,t,k)=\mathcal{M}_{-}(x,t,k)\mathsf{V}(x,t,k), \quad k\in\Gamma\setminus\{0\},
		\end{equation}
		where the jump matrices $\mathsf{V}(x,t,k)=\mathsf{V}_{j}(x,t,k)$ for $k\in\Gamma_{j}~(j=1,2,\cdots,6)$ are listed below
		\begin{equation}\label{jump matrix v}
			\begin{aligned}
				\mathsf{V}_1(x,t,k)&=\begin{pmatrix}
					1-|\rho_2(k)|^2 & -\rho_2^*(k)e^{-\Theta_{21}} & 0\\
					\rho_2(k)e^{\Theta_{21}}  & 1    & 0\\
					0 &  0   & 1
				\end{pmatrix},
				&\mathsf{V}_2(x,t,k)=\begin{pmatrix}
					1 & 0 & 0\\
					0 & 1    & -\rho_1(\omega k)e^{-\Theta_{32}}\\
					0 &  \rho_1^*(\omega k)e^{\Theta_{32}}   & 1-\rho_1(\omega k)\rho_{1}^{*}(\omega k)
				\end{pmatrix},\\
				\mathsf{V}_3(x,t,k)&=\begin{pmatrix}
					1 & 0 & \rho_2(\omega^2 k)e^{-\Theta_{31}}\\
					0 & 1    & 0\\
					-\rho_2^*(\omega^2k)e^{\Theta_{31}} &  0   & 1-\rho_2(\omega^2 k)\rho_2^{*}(\omega^2 k)
				\end{pmatrix},
				&\mathsf{V}_4(x,t,k)=\begin{pmatrix}
					1 & -\rho_1(k)e^{-\Theta_{21}} & 0\\
					\rho_1^*(k)e^{\Theta_{21}} & 1-|\rho_1(k)|^2    & 0\\
					0 & 0   & 1
				\end{pmatrix},\\
				\mathsf{V}_5(x,t,k)&=\begin{pmatrix}
					1 & 0 & 0\\
					0 & 1-\rho_2(\omega k) \rho_2^{*}(\omega k)  & -\rho_2^*(\omega k)e^{-\Theta_{32}}\\
					0 &  \rho_2(\omega k)e^{\Theta_{32}}   & 1
				\end{pmatrix},
				&\mathsf{V}_6(x,t,k)=\begin{pmatrix}
					1-\rho_1(\omega^2k)\rho_1^{*}(\omega^2k) & 0 & \rho_1^*(\omega^2k)e^{-\Theta_{31}}\\
					0 & 1    & 0\\
					-\rho_1(\omega^2k)e^{\Theta_{31}} &  0   & 1
				\end{pmatrix}.
			\end{aligned}
		\end{equation}
	\end{proposition}

	Now, it's ready to construct the Riemann-Hilbert problem corresponding to the initial value condition of the Hirota-Satsuma equation (\ref{HS eq}).
	
	\begin{RH problem}[RH problem for the function $\mathcal{M}(x,t,k)$]\label{RH problem for M}
		The matrix-valued function $\mathcal{M}(x,t,k)$ satisfies the properties below
		\begin{enumerate}[$(a)$]
			\item $\mathcal{M}(x,t,k)$ is analytic for $k\in\mathbb{C}\setminus \Gamma$.
		\end{enumerate}
		
		\begin{enumerate}[$(b)$]
			\item For $k\in\Gamma\setminus\{0\}$, $\mathcal{M}(x,t,k)$ meets the jump condition in Proposition \ref{jump proposition}.
		\end{enumerate}
		
		\begin{enumerate}[$(c)$]
			\item If $k\not\in\Gamma$ and $k\to\infty$, the following expansion for function $\mathcal{M}(x,t,k)$ holds
			\begin{equation}\label{hs k to infty}
				\mathcal{M}(x,t,k)=I+\frac{\mathcal{M}_{\infty}^{(1)}(x,t)}{k}+O(\frac{1}{k^{2}}),
			\end{equation}
			where $\mathcal{M}_{\infty}^{(1)}$ is independent of $k$ and satisfies
			\begin{equation}\label{hs k to infty,M1}
				[\mathcal{M}_{\infty}^{(1)}]_{12}+[\mathcal{M}_{\infty}^{(1)}]_{13}=0.
			\end{equation}
		\end{enumerate}
		
		\begin{enumerate}[$(d)$]
			\item There are three-order matrices $\{\mathcal{M}_{1}^{(l)}(x,t)\}_{l=-1}^{\infty}$  independent of $k$ and satisfying
			\begin{equation}\label{hs k to 0}
				\mathcal{M}(x,t,k)=\frac{\mathcal{M}_{1}^{(-1)}(x,t)}{k}+\mathcal{M}_{1}^{(0)}(x,t)+\mathcal{M}_{1}^{(1)}(x,t)k+O(k^2),
			\end{equation}
			as $k\to 0$ and $k\in\Omega_{1} $. Moreover, there is a real-valued function $\alpha_{\mathcal{M}}(x,t)$ such that
			\begin{equation}
				\mathcal{M}_{1}^{(-1)}(x,t)=\alpha_{\mathcal{M}}(x,t)\begin{pmatrix}
					0 & 0 & \omega\\
					0 & 0 & \omega^{2}\\
					0 & 0 & 1
				\end{pmatrix}.
			\end{equation}
		\end{enumerate}		
		
		\begin{enumerate}[$(e)$]
			\item $\mathcal{M}(x,t,k)$ follows the $\mathbb{Z}_{2}$ and $\mathbb{Z}_{3}$ symmetries
			\begin{equation}\label{symmetry of hs}
\mathcal{A}\mathcal{M}(x,t,\omega k)\mathcal{A}^{-1}=\mathcal{M}(x,t,k)=\mathcal{B}\overline{\mathcal{M}(x,t,\bar{k})}\mathcal{B},\qquad  k\in\mathbb{C}\setminus\Gamma.
			\end{equation}
		\end{enumerate}
	\end{RH problem}
	Assuming the matrix function $\mathcal{M}(x,t,k)$ satisfies the RH problem \ref{RH problem for M} and $\{u,v\}$  is a Schwartz class solution of the Hirota-Satsuma equation (\ref{HS eq}), then $\{u,v\}$ can be reconstructed via
	\begin{equation}\label{reconstructed-1}
		\begin{cases}
			u(x,t)=3(1-\omega^{2})\lim\limits_{k\to\infty}\mathcal{M}_{13}(x,t,k),\\
			v(x,t)=-3\frac{\partial}{\partial_{x}}\lim\limits_{k\to\infty}\big(\mathcal{M}_{33}(x,t,k)-1\big)-\frac{u^2(x,t)}{3}.
		\end{cases}
	\end{equation}
	\par
	Because the function $\mathcal{M}(x,t,k)$ has first-order singularity at $k=0$, introduce two new functions $\mathcal{N}(x,t,k)=(\omega,\omega^2,1)\mathcal{M}(x,t,k)$ and $\mathcal{Q}(x,t,k)=(1,1,1)\mathcal{M}(x,t,k)$ to remove this singularity. Then the Riemann-Hilbert problems for $\mathcal{N}(x,t,k)$ and $\mathcal{Q}(x,t,k)$ are listed as follow.
	
	\begin{RH problem}[RH problems for $\mathcal{N}(x,t,k)$ and $\mathcal{Q}(x,t,k)$]\label{RH problem for N,Q}
		$\mathcal{N}(x,t,k)$ and $\mathcal{Q}(x,t,k)$ are two $1\times 3$-row-vector valued functions satisfying
		\begin{enumerate}[$(a)$]
			\item $\mathcal{N}(x,t,k)$ and $\mathcal{Q}(x,t,k)$ are analytic for $k\in\mathbb{C}\setminus\Gamma$
		\end{enumerate}
		
		\begin{enumerate}[$(b)$]
			\item For $k\in\Gamma\setminus\{0\}$, $\mathcal{N}(x,t,k)$ and $\mathcal{Q}(x,t,k)$ satisfy the jump conditions
			\begin{equation}
				\mathcal{N}_{+}(x,t,k)=\mathcal{N}_{-}(x,t,k)\mathsf{V}(x,t,k),\quad \mathcal{Q}_{+}(x,t,k)=\mathcal{Q}_{-}(x,t,k)\mathsf{V}(x,t,k).
			\end{equation}
		\end{enumerate}
		
		\begin{enumerate}[$(c)$]
			\item $\mathcal{N}(x,t,k)=(\omega,\omega^2,1)+O(k^{-1})$  and $\mathcal{Q}(x,t,k)=(1,1,1)+O(k^{-1})$, as $k\to\infty$.
		\end{enumerate}
		
		\begin{enumerate}[$(d)$]
			\item  $\mathcal{N}(x,t,k)=O(1)$  and $\mathcal{Q}(x,t,k)=O(1)$, as $k\to 0$.
		\end{enumerate}
	\end{RH problem}
	Now the RH problems for $\mathcal{N}(x,t,k)$ and $\mathcal{Q}(x,t,k)$ are regular at the origin $k=0$. Furthermore, the reconstruction formula in (\ref{reconstructed-1}) can be rewritten via the equalities
	\begin{equation}\label{fanjie u v}
		\begin{cases}
			u(x,t)=3\lim\limits_{k\to\infty}k\big[\mathcal{N}_{3}(x,t,k)-\mathcal{Q}_{3}(x,t,k)\big],\\
			v(x,t)=3\frac{\partial}{\partial x}\lim\limits_{k\to\infty}k\big[\mathcal{N}_{3}(x,t,k)-1\big]-\frac{u^2(x,t)}{3}.
		\end{cases}
	\end{equation}

	\section{\bf{Miura transformations}}\label{Section-3}
	
	This section builds the Miura transformations among the Hirota-Satsuma equation (\ref{HS eq}), good Boussinesq equation (\ref{gb eq}) and the modified Boussinesq equation (\ref{mb eq}) based on the relationships of their RH problems. Firstly, the correspondence between the Hirota-Satsuma equation (\ref{HS eq}) and the good Boussinesq equation (\ref{gb eq}) is considered in detail, then the similar analysis will be conducted in the correspondence between the Hirota-Satsuma equation (\ref{HS eq}) and the modified Boussinesq equation (\ref{mb eq}).

	\subsection{Relations with the good Boussinesq equation (\ref{gb eq})}
	Here, the RH problem for the good Boussinesq equation (\ref{gb eq}) is given directly by modifying the jump conditions in the former work \cite{Charlier-Lenells-2021,Charlier-Lenells-Wang-2021}.
	
	\begin{RH problem}[RH problem for the good Boussinesq equation]\label{RH problem for gb}
		The matrix-valued function ${M}(x,t,k)$ satisfies the characteristics below
		\begin{enumerate}[$(a)$]
			\item ${M}(x,t,k)$ is analytic for $k\in\mathbb{C}\setminus \Gamma$.
		\end{enumerate}
		
		\begin{enumerate}[$(b)$]
			\item For $k\in\Gamma\setminus\{0\}$, ${M}(x,t,k)$ meets the jump condition \begin{equation}
				M_{+}(x,t,k)=M_{-}(x,t,k)\mathsf{V}(x,t,k),
			\end{equation}
			where $\mathsf{V}(x,t,k)=\mathsf{V}_{j}(x,t,k)~(j=1,2,\cdots,6)$ are defined in (\ref{jump matrix v}).
		\end{enumerate}
		
		\begin{enumerate}[$(c)$]
			\item If $k\not\in\Gamma$ and $k\to\infty$,  ${M}(x,t,k)$ behaves
			\begin{equation}\label{gb k to infty}
				{M}(x,t,k)=I+\frac{{M}_{\infty}^{(1)}(x,t)}{k}+O(\frac{1}{k^{2}}),
			\end{equation}
			where ${M}_{\infty}^{(1)}$ is independent of $k$ and satisfies
			\begin{equation}\label{M1 infty}
				[{M}_{\infty}^{(1)}]_{12}=[{M}_{\infty}^{(1)}]_{13}=0.
			\end{equation}
		\end{enumerate}
		
		\begin{enumerate}[$(d)$]
			\item There are three-order matrices $\{{M}_{1}^{(l)}(x,t)\}_{l=-2}^{\infty}$  independent of $k$ and satisfying
			\begin{equation}\label{gb k to 0}
				{M}(x,t,k)=\frac{{M}_{1}^{(-2)}(x,t)}{k^{2}}+\frac{{M}_{1}^{(-1)}(x,t)}{k}+{M}_{1}^{(0)}(x,t)+{M}_{1}^{(1)}(x,t)k+O(k^2),
			\end{equation}
			as $k\to 0$ and $k\in\Omega_{1} $. Moreover, there is a real-valued function $\alpha_{M}(x,t)$ such that
			\begin{equation}\label{M-alpha-M}
				{M}_{1}^{(-2)}(x,t)=\alpha_{M}(x,t)\begin{pmatrix}
					0 & 0 & 1\\
					0 & 0 & 1\\
					0 & 0 & 1
				\end{pmatrix}.
			\end{equation}
		\end{enumerate}		
		
		\begin{enumerate}[$(e)$]
			\item ${M}(x,t,k)$ follows the $\mathbb{Z}_{2}$ and $\mathbb{Z}_{3}$ symmetries
			\begin{equation}\label{symmetry of gb}
\mathcal{A}{M}(x,t,\omega k)\mathcal{A}^{-1}={M}(x,t,k)=\mathcal{B}\overline{{M}(x,t,\bar{k})}\mathcal{B},\qquad  k\in\mathbb{C}\setminus\Gamma.
			\end{equation}
		\end{enumerate}
	\end{RH problem}
	
	Letting $w$ be a solution of the good Boussinesq equation (\ref{good boussinesq}) or (\ref{gb eq}) in Schwartz space, $w(x,t)$ can be reconstructed via the equality
	\begin{equation}\label{solution w}
		w(x,t)=\frac{3}{2}\frac{\partial}{\partial x}\lim\limits_{k\to\infty}\left(M_{33}(x,t,k)-1\right).
	\end{equation}
	
	The next theorem expounds the correspondences of RH problems and solutions between the Hirota-Satsuma equation (\ref{HS eq}) and the good Boussinesq equation (\ref{gb eq}).

	\begin{theorem}\label{RHP between gb and hs}
		Suppose the function ${M}(x,t,k)$ solves the RH problem \ref{RH problem for gb} and $\mathcal{M}(x,t,k)$ solves the RH problem \ref{RH problem for M} for $x\in\mathbb{R}$ and $t\in[0,\infty)$, then the following correspondences hold:
		\begin{enumerate}[$(1)$]
			\item There is a complex-valued matrix $A_{1}(x,t)$ that connects the RH problem \ref{RH problem for M} with the RH problem \ref{RH problem for gb} in such way
			\begin{equation}\label{M new definition}
				M(x,t,k)=(I+\frac{A_1(x,t)}{k})\mathcal{M}(x,t,k),
			\end{equation}
			where
			\begin{equation}\label{A1}
				A_{1}(x,t)=\frac{u(x,t)}{9}\begin{pmatrix}
					0 & -\omega(\omega+2) & \omega(\omega+2)\\
					\omega+2 & 0 & -(\omega+2)\\
					-\omega^2(\omega+2) & \omega^2(\omega+2) & 0\\
				\end{pmatrix}.
			\end{equation}
		\end{enumerate}
		\begin{enumerate}[$(2)$]
			\item The solution $w$ of the good Boussinesq equation (\ref{gb eq}) is related to the solution $\{u,v\}$ of the Hirota-Satsuma equation (\ref{HS eq}) via the Miura transformation	
	\begin{equation}\label{HS-to-GB}
		\begin{cases}
			w=-\frac{1}{6}u^2-\frac{1}{2}v,\\
			h=\frac{1}{3}u_{xx}-\frac{2}{27}u^{3}-\frac{1}{3}uv-\frac{1}{2}v_{x},
		\end{cases}
	\end{equation}
		\end{enumerate}
	\end{theorem}
	\begin{proof}
		Suppose the function $\mathcal{M}(x,t,k)$ satisfies the RH problem \ref{RH problem for M}, then the function ${M}(x,t,k)$ defined in the equation (\ref{M new definition}) will be proved to satisfy the RH problem \ref{RH problem for gb} in the subsequent analysis.
		\par
		The analyticity and  the jump condition of the function ${M}(x,t,k)$ follow from the analyticity and  the jump condition of the function $\mathcal{M}(x,t,k)$ obviously. The facts below
		\begin{equation}
			A_1(x,t)=\omega^2\mathcal{A}A_1(x,t)\mathcal{A}^{-1},\quad A_1(x,t)=\mathcal{B}\overline{A_1(x,t)}\mathcal{B},
		\end{equation}
		and the symmetries in (\ref{symmetry of M}) of the function $\mathcal{M}(x,t,k)$ imply the function ${M}(x,t,k)$ satisfies the symmetries in (\ref{symmetry of gb}).
		\par
		Take the asymptotic expression (\ref{hs k to infty}) of  $\mathcal{M}(x,t,k)$ into the equation (\ref{M new definition}), then ${M}(x,t,k)$ obeys the asymptotics
		\begin{equation}\label{new expression of M k to infty}
			\begin{aligned}
				M(x,t,k)&=\left(I+\frac{A_1(x,t)}{k}\right)\left(I+\frac{\mathcal{M}_{\infty}^{(1)}(x,t)}{k}+O(\frac{1}{k^{2}})\right)\\
				&=I+\frac{A_1(x,t)+\mathcal{M}_{\infty}^{(1)}(x,t)}{k}+O(\frac{1}{k^{2}}),\qquad k\to\infty.
			\end{aligned}
		\end{equation}
		It's easy to find that
		\begin{equation}\label{eq1}
			{M}_{\infty}^{(1)}(x,t)=A_1(x,t)+\mathcal{M}_{\infty}^{(1)}(x,t)
		\end{equation}
		and ${M}_{\infty}^{(1)}(x,t)$ satisfies the condition (\ref{M1 infty}).
		\par
		As $k\to0$ and $k\in\Omega_1$, the function $\mathcal{M}(x,t,k)$ permits the expansion (\ref{hs k to 0}) which leads to
		 \begin{equation}\label{new expression of M k to 0}
		 	\begin{aligned}
		 		M(x,t,k)&=\left(I+\frac{A_1(x,t)}{k}\right)\left(\frac{\mathcal{M}_{1}^{(-1)}(x,t)}{k}+\mathcal{M}_{1}^{(0)}(x,t)+\mathcal{M}_{1}^{(1)}(x,t)k+O(k^2)\right)\\
		 		&=\frac{A_1(x,t)\mathcal{M}_{1}^{(-1)}(x,t)}{k^2}+\frac{A_1(x,t)\mathcal{M}_{1}^{(0)}(x,t)+\mathcal{M}_{1}^{(-1)}(x,t)}{k}+\cdots,\quad k\to0.
		 	\end{aligned}
		 \end{equation}
		This implies that $\alpha_{M}=-\frac{u}{3}\alpha_{\mathcal{M}}$ in (\ref{M-alpha-M}) and
		\begin{equation}\label{eq2}
			{M}_{1}^{(-2)}(x,t)=A_1(x,t)\mathcal{M}_1^{(-1)}(x,t).
		\end{equation}
		\par
		The above analysis has proved that ${M}(x,t,k)$ defined in the equation (\ref{M new definition}) obeys the RH problem \ref{RH problem for gb}. Inserting the equality (\ref{eq1}) into the reconstructed formula (\ref{solution w}) and considering the equations (\ref{hs k to infty}) and (\ref{hs k to infty,M1}),  the Miura transformation (\ref{HS-to-GB}) can be achieved without much effort.
	\end{proof}

	\subsection{Relations with modified Boussinesq equation} Similarly, the correspondences of RH problems and solutions between the Hirota-Satsuma equation (\ref{HS eq}) and the modified Boussinesq equation (\ref{mb eq}) can also be formulated.

	\begin{RH problem}[RH problem for the modified Boussinesq equation (\ref{mb eq})]\label{RH problem for mb}
		The matrix-valued function $m(x,t,k)$ satisfies the properties below
		\begin{enumerate}[$(a)$]
			\item $m(x,t,k)$ is analytic for $k\in\mathbb{C}\setminus \Gamma$.
		\end{enumerate}
		
		\begin{enumerate}[$(b)$]
			\item For $k\in\Gamma$, $m(x,t,k)$ meets the jump condition \begin{equation}
				m_{+}(x,t,k)=m_{-}(x,t,k)\mathsf{V}(x,t,k)
			\end{equation}
			where $\mathsf{V}(x,t,k)=\mathsf{V}_{j}(x,t,k)~(j=1,2,\cdots,6)$ are defined in (\ref{jump matrix v}).
		\end{enumerate}
		
		\begin{enumerate}[$(c)$]
			\item If $k\not\in\Gamma$ and $k\to\infty$, $m(x,t,k)$ behaves
			\begin{equation}
				m(x,t,k)=I+\frac{{m}_{\infty}^{(1)}(x,t)}{k}+O(\frac{1}{k^{2}}).
			\end{equation}
		\end{enumerate}
		
		\begin{enumerate}[$(d)$]
			\item There are three-order matrces $\{m_{1}^{(l)}(x,t)\}_{l=0}^{\infty}$  independent of $k$ and satisfying
			\begin{equation}
				m(x,t,k)=m_{1}^{(0)}(x,t)+m_{1}^{(1)}(x,t)k+O(k^2),
			\end{equation}
			as $k\to 0$ and $k\in\Omega_{1} $.
		\end{enumerate}		
		
		\begin{enumerate}[$(e)$]
			\item $m(x,t,k)$ follows the $\mathbb{Z}_{2}$ and $\mathbb{Z}_{3}$ symmetries
			\begin{equation}\label{symmetry of mb}
\mathcal{A}m(x,t,\omega k)\mathcal{A}^{-1}=m(x,t,k)=\mathcal{B}\overline{m(x,t,\bar{k})}\mathcal{B},\qquad  k\in\mathbb{C}\setminus\Gamma.
			\end{equation}
		\end{enumerate}
	\end{RH problem}
	\begin{theorem}\label{RHP between hs and mgb}
		Suppose the function $\mathcal{M}(x,t,k)$ solves the RH problem \ref{RH problem for M} and ${m}(x,t,k)$ solves the RH problem \ref{RH problem for mb} for $x\in\mathbb{R}$ and $t\in[0,\infty)$, then the following correspondences hold:
		\begin{enumerate}[$(1)$]
			\item There is a complex-valued matrix $A_{2}(x,t)$ that connects the RH problem \ref{RH problem for M} with the RH problem \ref{RH problem for mb} in such way
			\begin{equation}
				\mathcal{M}(x,t,k)=(I+\frac{A_2(x,t)}{k})m(x,t,k),
			\end{equation}
			where
			\begin{equation}
				A_{2}(x,t)=-\frac{2}{3}p(x,t)\begin{pmatrix}
					\omega^2 & 1 & \omega\\
					1 & \omega  & \omega^2\\
					\omega & \omega^2 & 1\\
				\end{pmatrix}.
			\end{equation}
		\end{enumerate}
		\begin{enumerate}[$(2)$]
			\item The solution $\{u,v\}$ of the Hirota-Satsuma equation (\ref{HS eq}) is related to the solution $\{p,q\}$ of the modified Boussinesq equation (\ref{mb eq}) via the Miura transformation
			\begin{equation}\label{mGB-to-HS}
				\begin{cases}
					u=3(p-q),\\
					v=-2p^2+6pq+2p_x.
				\end{cases}
			\end{equation}
		\end{enumerate}
	\end{theorem}

\subsection{Relations between the good Boussinesq equation and the modified Boussinesq equation}
The next theorem which corresponds to the function ${M}(x,t,k)$ for RH problem \ref{RH problem for gb} and function ${m}(x,t,k)$ for RH problem \ref{RH problem for mb} is inferred from the Theorem \ref{RHP between gb and hs} and Theorem  \ref{RHP between hs and mgb} directly.
	\begin{theorem}\label{RHP between gb and mb}
	Suppose the function ${M}(x,t,k)$ solves the RH problem \ref{RH problem for gb} and ${m}(x,t,k)$ solves the RH problem \ref{RH problem for mb} for $x\in\mathbb{R}$ and $t\in[0,\infty)$, then the following correspondences hold:
	\begin{enumerate}[$(1)$]
		\item There are complex-valued matrices $A_{3}(x,t)$ and $A_{4}(x,t)$ that connect the RH problem \ref{RH problem for gb} with the RH problem \ref{RH problem for mb} in such way
		\begin{equation}
			{M}(x,t,k)=(I+\frac{A_3(x,t)}{k}+\frac{A_{4}(x,t)}{k^{2}})m(x,t,k),
		\end{equation}
		where
		\begin{equation}
			\begin{aligned}
			A_{3}(x,t)&=-\frac{2}{3}p\begin{pmatrix}
				\omega^2 & 1 & \omega\\
				1 & \omega  & \omega^2\\
				\omega & \omega^2 & 1\\
			\end{pmatrix}+\frac{p-q}{3}\begin{pmatrix}
			0 & -\omega(\omega+2) & \omega(\omega+2)\\
			\omega+2 & 0 & -(\omega+2)\\
			-\omega^2(\omega+2) & \omega^2(\omega+2) & 0\\
			\end{pmatrix},\\
			A_{4}(x,t)&=\frac{2p(p-q)}{3}\begin{pmatrix}
				\omega & \omega^2 & 1\\
				\omega & \omega^2 & 1\\
				\omega & \omega^2 & 1\\
			\end{pmatrix}.
			\end{aligned}
		\end{equation}
	\end{enumerate}
	\begin{enumerate}[$(2)$]
		\item The solution $w$ of the good Boussinesq equation (\ref{good boussinesq}) is related to the solution $\{p,q\}$ of the modified Boussinesq equation (\ref{mb eq}) via the Miura transformation
	\begin{equation}\label{mGB-to-GB}
		\begin{cases}
			w=-p_{x}-\frac{1}{2}p^{2}-\frac{3}{2}q^{2},\\
			h=-q_{xx}-3pq_{x}-qp_{x}-2qp^{2}+2q^{3}.
		\end{cases}
	\end{equation}

	\end{enumerate}
\end{theorem}
\par
Recently, Charlier and Lenells \cite{Charlier-2021} investigated the long-time asymptotics of the Mikhailov-Lenells equation (\ref{L eq}) based on its RH problem governed by the matrix-valued function $n(x,t,k).$ Following the procedure of Charlier and Lenells \cite{Charlier-2021}, without much effort,
one can build the relations between the good Boussinesq equation (\ref{good boussinesq}) and Mikhailov-Lenells equation (\ref{L eq}) equation, i.e., the connection between the functions $M(x,t,k)$ and $n(x,t,k)$ and the Miura transformation
\begin{equation}\label{ML-to-GB}
		\begin{cases}
			w=\frac{3+\sqrt{3}\mathrm{i}}{4}r_x+\frac{3-\sqrt{3}\mathrm{i}}{4}z_{x}-\frac{3}{2}zr,\\
			h=\frac{1-\sqrt{3}\mathrm{i}}{4}r_{xx}+\frac{1+\sqrt{3}\mathrm{i}}{4}z_{xx}
+\frac{\sqrt{3}\mathrm{i}(z+r)-3r}{2}r_{x}-\frac{\sqrt{3}\mathrm{i}(z+r)+3z}{2}z_{x}+r^{3}+z^{3}.
		\end{cases}
	\end{equation}
Moreover, one can also built the connection between the functions $\mathcal{M}(x,t,k)$ and $n(x,t,k)$ and the Miura transformation between the modified Boussinesq equation (\ref{mb eq}) and Mikhailov-Lenells equation (\ref{L eq}), that is
\begin{equation}\label{ML-to-HS}
		\begin{cases}		u=-(\frac{3}{2}+\frac{3\sqrt{3}\mathrm{i}}{2})z+(-\frac{3}{2}+\frac{3\sqrt{3}\mathrm{i}}{2})r,\\
			v=(\frac{3}{2}-\frac{3\sqrt{3}\mathrm{i}}{2})z^{2}+(\frac{3}{2}+\frac{3\sqrt{3}\mathrm{i}}{2})
r^{2}-3zr-\sqrt{3}\mathrm{i}z_{x}+\sqrt{3}\mathrm{i}r_{x}.
		\end{cases}
\end{equation}
\par
Summarizing the results of inverse spectral theory in Section \ref{Section-2} and Miura transformation in Section \ref{Section-3}, the main differences among Hirota-Satsuma equation (\ref{HS eq}), good Boussinesq equation (\ref{gb eq}), modified Boussinesq equation (\ref{mb eq}) and Mikhailov-Lenells equation (\ref{L eq}) are listed in Table 1.

\begin{table}\label{Table11}
		\begin{center}
			\caption{The differences among H-S equation (\ref{HS eq}), G-B equation (\ref{gb eq}), M-B equation (\ref{mb eq}) and M-L equation (\ref{L eq}).  }
			\begin{tabular}{|c|c|c|c|}
				\hline
				Equations & Properties of eigenfunctions & Reflection coefficients at $k=0$ & RH problems \\
				\hline
		H-S (\ref{HS eq}) & Simple pole at $k=0$ & $\rho_{1}(0)=\rho_{2}(0)=\omega$& $\mathcal{M}(x,t,k)$ \\
				\hline
	G-B (\ref{gb eq}) & Double pole at $k=0$ & $\rho_{1}(0)=\omega,\rho_{2}(0)=1$& $M(x,t,k)$			\\
				\hline
				M-B (\ref{mb eq}) & No singularity at $k=0$ & $|\rho_{1}(0)|<1,|\rho_{2}(0)|<1$& $m(x,t,k)$\\
				\hline
				M-L (\ref{L eq}) & No singularity at $k=0$ & $|\rho_{1}(0)|<1,|\rho_{2}(0)|<1$& $n(x,t,k)$\\
				\hline
			\end{tabular}
		\end{center}
	\end{table}

	\section{\bf{Large-time asymptotics}}\label{Section-4}
	
	This section explores the large-time asymptotics of the initial value problem of the Hirota-Satsuma equation (\ref{HS eq}) by deforming the RH problem \ref{RH problem for M} based on Deift-Zhou nonlinear steepest-descent strategy \cite{Deift-Zhou}.
	\par
	Recall the definition of $\Theta_{ij}$ in (\ref{theta ij}) and set
	\begin{equation*}\label{phi ij}
		\theta_{ij}=\Theta_{ij}(x,t,k)/{t}=(l_{i}-l_{j})\eta+(z_{i}-z_{j}), \quad 1\le j\le i\le 3,
	\end{equation*}
	where $\eta=\frac{x}{t}$. In what follows, we first assume $x>0$ and the case for $x<0$ can also be studied similarly. Taking $\frac{\partial\theta_{ij}}{\partial k}(\eta,k)=0$, three critical points are gotten
	\begin{equation*}
		k_{21}=-\frac{\eta}{2},\quad k_{31}=-\frac{\omega\eta}{2},\quad k_{32}=-\frac{\omega^2 \eta}{2}.
	\end{equation*}
	Set $k_{0}\equiv k_{21}<0$, then $k_{31}=\omega k_{0},k_{32}=\omega^{2}k_{0}$, which are displayed in Figure \ref{critical-points}. Furthermore, the sign signature of $\mathrm{Re}(\theta_{21})$ is shown in Figure \ref{signature-points}, where $\mathrm{Re}(\theta_{21})>0$ in pink shadowed region, while $\mathrm{Re}(\theta_{21})<0$ in white region.
\par

\begin{figure}[htp]
		\centering
		\begin{tikzpicture}[>=latex]
			\draw[very thick] (-4,0) to (4,0) node[right]{$\mathbb{R}$};
			\draw[very thick] (-2,-1.732*2) to (2,1.732*2);
			\draw[very thick] (-2,1.732*2) to (2,-1.732*2);
			\filldraw[black] (-1.6,0) node[below=1mm]{$k_{0}$} circle (1.5pt);
			\filldraw[black] (.8,1.732*0.8) node[right=1mm]{$\omega^{2}k_{0}$} circle (1.5pt);
			\filldraw[black] (.8,-1.732*0.8) node[right=1mm]{$\omega k_{0}$} circle (1.5pt);
		\end{tikzpicture}
		\caption{Three critical points $k_{0},\omega k_{0},\omega^{2}k_{0}$ in the complex $k$-plane $\mathbb{C}$.}\label{critical-points}
	\end{figure}
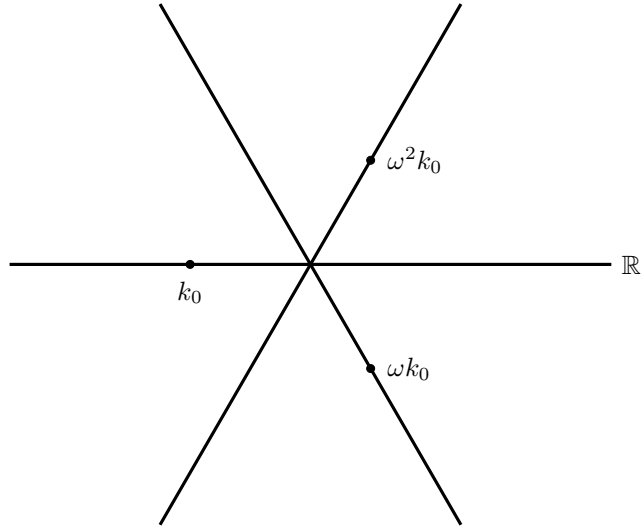

	\par
	In what follows, the RH problem \ref{RH problem for M} for $\mathcal{M}$ is deformed into RH problem for $\mathcal{M}^{(1)}$ firstly. Then the RH problem for $\mathcal{M}^{(1)}$ is deformed into RH problems for $\mathcal{M}^{(2)}$ and $\mathcal{M}^{(3)}$, respectively. Finally, a solvable model RH problem is obtained. One can write the process into the following form.
	
	\begin{RH problem}[RH problem for $\mathcal{M}^{(j)}$]\label{RHP for M}
		In the following deformations, we want to seek some three-order  matrices  $\mathcal{M}^{(j)}(x,t,k)~(j=1,2,3)$ satisfying the jump conditions
		\begin{equation*}
			\mathcal{M}^{(j)}_{+}(x,t,k)=\mathcal{M}_{-}^{(j)}(x,t,k)\mathsf{V}^{(j)}(x,t,k), \qquad k\in \Gamma^{(j)}~{\rm and}~j=1,2,3.
		\end{equation*}
	\end{RH problem}

	\begin{figure}[htp]
		\centering
		\begin{tikzpicture}[>=latex]
			\fill [purple!10!white] (-1.6,0) rectangle (4,1.732*2);
			\fill [purple!10!white] (-1.6,0) rectangle (-4,-1.732*2);
			\draw[very thick] (-4,0) to (4,0) node[right]{$\mathbb{R}$};
			\draw[very thick] (-2,-1.732*2) to (2,1.732*2);
			\draw[very thick] (-2,1.732*2) to (2,-1.732*2);
			\filldraw[black] (-1.6,0) node[below=1mm]{$k_{0}$} circle (1.5pt);
			\filldraw[black] (.8,1.732*0.8) node[right=1mm]{$\omega^{2}k_{0}$} circle (1.5pt);
			\filldraw[black] (.8,-1.732*0.8) node[right=1mm]{$\omega k_{0}$} circle (1.5pt);
			\node at (-3,-2){$\mathrm{Re}\theta_{21}>0$};
			\node at (-3,2){$\mathrm{Re}\theta_{21}<0$};
		\end{tikzpicture}
		\caption{The sign signature of $\mathrm{Re}(\theta_{21})$: $\mathrm{Re}(\theta_{21})>0$ (pink shadowed region) and $\mathrm{Re}(\theta_{21})<0$ (white region).}\label{signature-points}
	\end{figure}
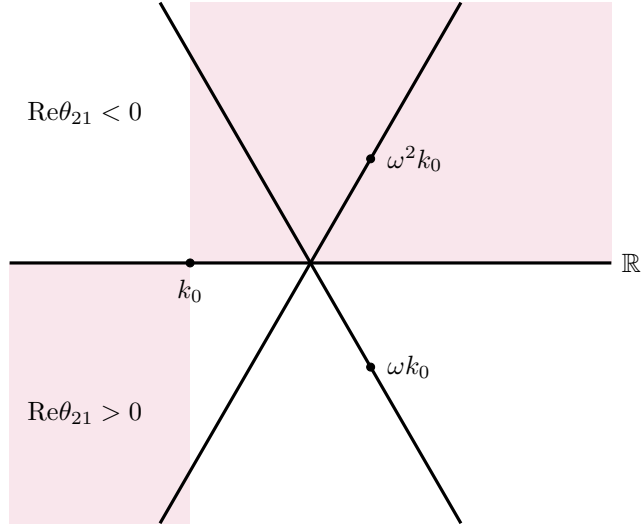

	\subsection{The first transformation: $\mathcal{M}\to \mathcal{M}^{(1)}$}
	
	The first transformation is aimed to eliminate the jumps across the sub-contours $\mathbb{R}^{-},\omega \mathbb{R}^{-}$ and $\omega^{2}\mathbb{R}^{-}$ of $\Gamma$ in Figure \ref{Figure-jump} aside from a small residue. Let us define
	\begin{equation*}\label{rho-hat}
		\hat{\rho}_{1}(k)=\frac{\rho_{1}(k)}{1-|\rho_{1}(k)|^{2}}, \qquad k<0.
	\end{equation*}
	Further, introduce four open sets $D_{j}~(j=1,2,3,4)$ shown in Figure \ref{figure D1,2,3,4} such that
	\begin{equation*}
		D_{1}\cup D_{3}=\{k|\mathrm{Re}(\theta_{21})>0\},\quad 	D_{2}\cup D_{4}=\{k|\mathrm{Re}(\theta_{21})<0\}.
	\end{equation*}

	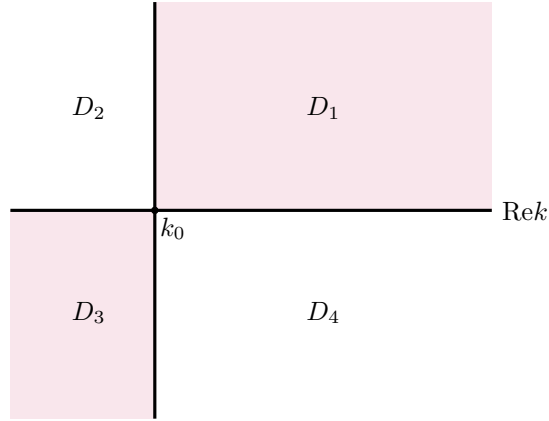
\begin{figure}[htp]
		\centering
		\begin{tikzpicture}[>=latex,scale=0.8]
			\fill [purple!10!white] (-1.6,0) rectangle (4,1.732*2);
			\fill [purple!10!white] (-1.6,0) rectangle (-4,-1.732*2);
			\draw[very thick] (-4,0) to (4,0) node[right]{$\mathrm{Re} {k}$};
			\draw[very thick] (-1.6,-1.732*2)--(-1.6,1.732*2);
			\filldraw[black] (-1.6,0) circle (1.5pt);
			\node at (-1.3,-.3){$k_{0}$};
			\node at (-2.7,-1.7){$D_{3}$};
			\node at (1.2,1.7){$D_{1}$};
			\node at (1.2,-1.7){$D_{4}$};
			\node at (-2.7,1.7){$D_{2}$};
		\end{tikzpicture}
		\caption{Four open sets $D_1, D_2, D_3, D_4$ in the complex $k$-plane $\mathbb{C}$.}
		\label{figure D1,2,3,4}
	\end{figure}
	\begin{lemma} Let $\mathcal{K}$ represent a compact subset of $(0,\infty )$ which is fixed.
		Factorize the scattering data $\rho_1,\rho_{2}$ and $\hat{\rho}_{1}(k)$ as follow:
		\begin{equation}
			\begin{aligned}
				\rho_{1}(k)&=\rho_{1,a}(x,t,k)+\rho_{1,r}(x,t,k), \quad k\in[k_{0},0],\\
				\rho_{2}(k)&=\rho_{2,a}(x,t,k)+\rho_{2,r}(x,t,k), \quad k\in[0,\infty),\\
				\hat{\rho}_{1}(k)&=\hat{\rho}_{1,a}(x,t,k)+\hat{\rho}_{1,r}(x,t,k), \quad k\in(-\infty,k_{0}),
			\end{aligned}		
		\end{equation}
		where the functions $\rho_{1,a},\rho_{1,r},\rho_{2,a},\rho_{2,r},\hat{\rho}_{1,a}$ and $\hat{\rho}_{1,r}$ satisfy the following characteristics:
		\begin{enumerate}[$(1)$]
			\item Assume $\eta\in\mathcal{K}$ and $t>0$. Then $\rho_{1,a}$ and $\rho_{2,a}$ are well-defined and continuous for $k\in\bar{D}_{4}$ and analytic for $k\in D_{4}$, and $\hat{\rho}_{1,a}$ is well-defined and continuous for $k\in \bar{D}_{3}$ and analytic for $k\in D_{3}$.
		\end{enumerate}
		\begin{enumerate}[$(2)$]
			\item For $\eta\in\mathcal{K}$ and $t>0$, $\rho_{1,a}, \rho_{2,a}$ and $\hat{\rho}_{1,a}$ have the properties below:\\
			\begin{subequations}
				\begin{align}
					&|\rho_{2,a}|\le\frac{C|k-\omega k_{0}|}{1+k^{2}}e^{\frac{t}{4}|\mathrm{Re}\theta_{21}(\eta,k)|},  &k\in\bar{D}_{4}, \\
					&|\partial_{x}^{j}(\rho_{2,a}(x,t,k)-\rho_{2}(0))|\le C|k|e^{\frac{t}{4}|\mathrm{Re}\theta_{21}(\eta,k)|},  &k\in\bar{D}_{4},\\
					&|\partial_{x}^{j}(\rho_{1,a}(x,t,k)-\rho_{1}(0))|\le C|k|e^{\frac{t}{4}|\mathrm{Re}\theta_{21}(\eta,k)|},  &k\in\bar{D}_{4},\\
					&|\partial_{x}^{j}(\rho_{1,a}(x,t,k)-\rho_{1}(k_0))|\le C|k-k_{0}|e^{\frac{t}{4}|\mathrm{Re}\theta_{21}(\eta,k)|},  &k\in\bar{D}_{4},\\
					&|\partial_{x}^{j}(\hat{\rho}_{1,a}(x,t,k)-\hat{\rho}_{1}(k_0))|\le C|k-k_{0}|e^{\frac{t}{4}|\mathrm{Re}\theta_{21}(\eta,k)|},  &k\in\bar{D}_{3},\\
					&|\partial_{x}^{j}(\hat{\rho}_{1,a}(x,t,k))|\le \frac{C}{1+|k|}e^{\frac{t}{4}|\mathrm{Re}\theta_{21}(\eta,k)|}, &k\in\bar{D}_{3},
				\end{align}
			\end{subequations}
			where $j=0,1$ and $C$ is a constant independent of $\eta,t$ and $k$.
		\end{enumerate}

	\end{lemma}

	Since $\rho_2=\rho_{2,a}+\rho_{2,r}$, there are decompositions
	\begin{equation*}
		\mathsf{V}_1=\mathsf{V}_{1,a}^U\mathsf{V}_{1,r}\mathsf{V}_{1,a}^L,\quad \mathsf{V}_3=\mathsf{V}_{3,a}^L\mathsf{V}_{3,r}\mathsf{V}_{3,a}^U,\quad \mathsf{V}_5=\mathsf{V}_{5,a}^U\mathsf{V}_{5,r}\mathsf{V}_{5,a}^L,
	\end{equation*}
	where
	\begin{equation*}
		\begin{aligned}
			\mathsf{V}_{1,a}^U&=\begin{pmatrix}
				1 &  -\rho_{2,a}^{*}(k)e^{-t\theta_{21}} & 0 \\
				0 &1 & 0\\
				0 & 0 &1
			\end{pmatrix}, \qquad
			\mathsf{V}_{1,a}^L=\begin{pmatrix}
				1 & 0& 0\\
				\rho_{2,a}(k)e^{t\theta_{21}} & 1 & 0\\
				0 &  0  & 1
			\end{pmatrix},\\
			\mathsf{V}_{3,a}^L&=\begin{pmatrix}
				1 & 0 & 0 \\
				0 & 1 & 0\\
				-\rho_{2,a}^*(\omega^2k)e^{t\theta_{31}} & 0 & 1
			\end{pmatrix},\qquad
			\mathsf{V}_{3,a}^U=\begin{pmatrix}
				1 & 0 & \rho_{2,a}(\omega^2k)e^{-t\theta_{31}}\\
				0 & 1 & 0\\
				0 &  0  & 1
			\end{pmatrix},\\
			\mathsf{V}_{5,a}^U&=\begin{pmatrix}
				1 &  0 & 0 \\
				0 & 1 & -\rho_{2,a}^*(\omega k)e^{-t\theta_{32}}\\
				0 & 0 &1
			\end{pmatrix},\qquad
			\mathsf{V}_{5,a}^L=\begin{pmatrix}
				1 & 0 & 0 \\
				0 & 1 & 0\\
				0 &  \rho_{2,a}(\omega k)e^{t\theta_{32}}  & 1
			\end{pmatrix},
		\end{aligned}
	\end{equation*}
	and
	\begin{equation*}
		\begin{aligned}
			\mathsf{V}_{1,r}&=\begin{pmatrix}
				1-\rho_{2,r}(k)\rho_{2,r}^*(k) & -\rho_{2,r}(k)^*e^{-t\theta_{21}} &0\\
				\rho_{2,r}(k)e^{t\theta_{21}} & 1 & 0\\
				0 & 0 & 1
			\end{pmatrix},\\
			\mathsf{V}_{3,r}&=\begin{pmatrix}
				1 & 0 & \rho_{2,r}(\omega^2k)e^{-t\theta_{31}} \\
				0 & 1 & 0\\
				-\rho_{2,r}^*(\omega^2k)e^{t\theta_{31}} & 0 & 1-\rho_{2,r}(\omega^2k)\rho_{2,r}^{*}(\omega^2k)
			\end{pmatrix},\\
			\mathsf{V}_{5,r}&=\begin{pmatrix}
				1 & 0 & 0\\
				0 & 1-\rho_{2,r}(\omega k)\rho_{2,r}^{*}(\omega k) & -\rho_{2,r}^*(\omega k)e^{-t\theta_{32}}\\
				0 & \rho_{2,r}(\omega k)e^{t\theta_{32}} & 1
			\end{pmatrix}.
		\end{aligned}
	\end{equation*}
	\par	
	Then define a sectionally analytic function $\mathcal{M}^{(1)}$ by the transformation
	\begin{equation}\label{H(x,t,k)}
		\mathcal{M}^{(1)}(x,t,k)=\mathcal{M}(x,t,k)H(x,t,k),
	\end{equation}
	where
	\begin{equation*}
		H(x,t,k)=\begin{cases}
			(\mathsf{V}_{1,a}^L)^{-1},\qquad  &k\in \Omega_{1},\\
			\mathsf{V}_{3,a}^L& k\in \Omega_{2},\\
			(\mathsf{V}_{3,a}^U)^{-1},\qquad  &k\in \Omega_{3},\\
			\mathsf{V}_{5,a}^U & k\in \Omega_{4},\\
			(\mathsf{V}_{5,a}^L)^{-1},\qquad  &k\in \Omega_{5},\\
			\mathsf{V}_{1,a}^U & k\in \Omega_{6}.\\
		\end{cases}
	\end{equation*}

	\begin{lemma}\label{H}
		For $k\in\mathbb{C}\setminus \Gamma$, $t>0$ and $\eta\in \mathcal{K}$, the functions $H(x,t,k)$ and $H(x,t,k)^{-1}$ are uniformly bounded . Additionally, $H(x,t,k)=I+O(k^{-1})$ as $k\to\infty$.
	\end{lemma}
	
	Lemma \ref{H} indicates that $\mathcal{M}^{(1)}(x,t,k)$ meets RH problem \ref{RHP for M} with $j=1$, where $\Gamma^{(1)}=\Gamma$ and the jump matrix $\mathsf{V}^{(1)}(x,t,k)$ is provided by
	\begin{equation*}
		\mathsf{V}^{(1)}=\begin{cases}
			\mathsf{V}_{1,r} ,\qquad & k\in\Gamma_{1},\\
			\mathsf{V}_{1,a}^L\mathsf{V}_2\mathsf{V}_{3,a}^L,  & k\in\Gamma_{2},\\
			\mathsf{V}_{3,r} ,\qquad & k\in\Gamma_{3},\\
			\mathsf{V}_{3,a}^U\mathsf{V}_4\mathsf{V}_{5,a}^U, & k\in\Gamma_{4},\\
			\mathsf{V}_{5,r} ,\qquad & k\in\Gamma_{5},\\
			\mathsf{V}_{5,a}^L\mathsf{V}_6\mathsf{V}_{1,a}^U, & k\in\Gamma_{6}.\\
		\end{cases}
	\end{equation*}

	\subsection{The second transformation}

	Let $\Gamma^{(2)}=\cup_{j=1}^{9}\Gamma_{j}^{(2)} $ represent the new contour shown in Figure
	\ref{figure of gamma(2)}. Then seek an analytic function $\mathcal{D}_{4}(\eta,k)$ satisfying
	\begin{equation*}
		\mathcal{D}_{4+}(\eta,k)=\mathcal{D}_{4-}(\eta,k)(1-|\rho_1(k)|^2),\quad \eta\in\mathcal{K},\quad k\in\Gamma_4^{(2)},
	\end{equation*}
	and
	\begin{equation*}
		\mathcal{D}_{4}(\eta,k)=1+O(k^{-1}),\quad k\to -\infty .
	\end{equation*}
	It is inferred from the Plemelj formulas that
	\begin{equation}\label{HD}
		\mathcal{D}_{4}(\eta,k)=\exp\left\{\frac{1}{2\pi {\rm i}}\int_{\Gamma_{4}^{(2)}}\frac{\ln(1-|\rho_1(s)|^2)}{s-k}ds\right\},\qquad k\in \mathbb{C}\setminus \Gamma_{4}^{(2)}.
	\end{equation}
	Denote $\ln_{\pi}(k)$ as the logarithm of $k$ with branch cut along $\arg k=\pi $, i.e., $\ln_{\pi}(k)=\ln|k|+{\rm i}\arg_{\pi}k$ where $\arg_{\pi}k\in(-\pi,\pi)$. The properties of the function $\mathcal{D}_{4}(\eta,k)$ are listed in the following lemma.

	\begin{figure}[htp]
		\centering
		\begin{tikzpicture}[>=latex]
			\draw[very thick] (-4,0) to (4,0) node[right]{$\Gamma^{(2)}$};
			\draw[very thick] (-2,-1.732*2) to (2,1.732*2);
			\draw[very thick] (-2,1.732*2) to (2,-1.732*2);
			\draw[<->,very thick] (-3,0)node[below]{$4$} to (2,0)node[below]{$1$};
			\draw[->,very thick] (0,0) to (-1,0) node[below]{$8$};
			\draw[<->,very thick] (-1,-1.732)node[right=1mm]{$5$} to (1.5,1.732*1.5)node[right=1mm]{$2$};
			\draw[->,very thick] (0,0) to (.5,1.732*0.5) node[right=1mm]{$7$};
			\draw[<->,very thick] (-1,1.732)node[right=1mm]{$3$} to (1.5,-1.732*1.5)node[right=1mm]{$6$};
			\draw[->,very thick] (0,0) to (.5,-1.732*0.5) node[right=1mm]{$9$};
			\filldraw[black] (-1.6,0) node[below=1mm]{$k_{0}$} circle (1.5pt);
			\filldraw[black] (.8,1.732*0.8) node[right=1mm]{$\omega^{2}k_{0}$} circle (1.5pt);
			\filldraw[black] (.8,-1.732*0.8) node[right=1mm]{$\omega k_{0}$} circle (1.5pt);
		\end{tikzpicture}
		\caption{The jump contour $\Gamma^{(2)}$ for RH problem $\mathcal{M}^{(2)}(x,t,k)$.}
		\label{figure of gamma(2)}
	\end{figure}

	\begin{lemma}\label{property of HD}
		The function $\mathcal{D}_{4}(\eta,k)$ has properties as follow:
		\begin{enumerate}[$(1)$]
			\item $\mathcal{D}_{4}(\eta,k)$ can be expressed as
			\begin{equation}\label{the expansion of D}
				\mathcal{D}_{4}(\eta,k)=e^{-{\rm i}\mu \ln_{\pi }(k-k_{0})}e^{-\kappa_{4}(\eta,k)},
			\end{equation}
			where
			\begin{equation*}
				\mu(\eta)=-\frac{1}{2\pi}\ln(1-|\rho_1(k_0)|^2),\qquad \eta\in \mathcal{K},
			\end{equation*}
			and
			\begin{equation*}
				\kappa_{4}(\eta,k)=\frac{1}{2\pi {\rm i}}\int^{-\infty}_{k_0}\ln_{\pi}(k-s)d\ln(1-|\rho_1(s)|^2).
			\end{equation*}
		\end{enumerate}
		
		\begin{enumerate}[$(2)$]
			\item The function $\mathcal{D}_{4}(\eta,k)$ allows the symmetry
			\begin{equation*}
				\mathcal{D}_4(\eta,k)=\overline{\mathcal{D}_4(\eta,\bar{k})}^{-1}, \qquad \eta\in \mathcal{K},\quad k\in \mathbb{C}\setminus \Gamma_{4}^{(2)} .
			\end{equation*}
		\end{enumerate}
		
		\begin{enumerate}[$(3)$]
			\item $\mathcal{D}_4(\eta,k)$ and $\mathcal{D}_4(\eta,k)^{-1}$ are analytic on $k\in\mathbb{C}\setminus\Gamma_4^{(2)}$ and
			\begin{equation*}
				\sup_{\eta\in\mathcal{K}}\sup_{k\in\mathbb{C}\setminus\Gamma_4^{(2)}}|\mathcal{D}_4(\eta,k)^{\pm1}|<\infty.
			\end{equation*}
		\end{enumerate}
		
		\begin{enumerate}[$(4)$]
			\item There is a positive constant $C$ such that
			\begin{equation*}
				\begin{aligned}
					&|\kappa_{4}(\eta,k)-\kappa_{4}(\eta,k_0)|\le C|k-k_0|(1+|\ln|k-k_0||),\\
					&|\partial_x(\kappa_{4}(\eta,k)-\kappa_{4}(\eta,k_0))|\le \frac{C}{t}(1+|\ln|k-k_0||),
				\end{aligned}
			\end{equation*}
			as $k$ approaches $k_0$ along a path that is nontangential to $(-\infty,k_0)$. Moreover, there are
			\begin{equation*}
				|\partial_x\kappa_{4}(\eta,k_0)|=\frac{1}{t}\bigg|\partial_u\kappa_{4}(u,v)|_{(u,v)=(\eta,k_0)}+\frac{1}{2}\partial_v\kappa_{4}(u,v)|_{(u,v)=(\eta,k_0)}\bigg|\le \frac{C}{t},
			\end{equation*}
			and
			\begin{equation*}
				\partial_x(\mathcal{D}_4(\eta,k)^{\pm1})=\frac{\pm i\mu}{2t(k-k_0)}\mathcal{D}_4(\eta,k)^{\pm1}.
			\end{equation*}
		\end{enumerate}
		
	\end{lemma}

	In the same way, define $\mathcal{D}_{6}$ and $\mathcal{D}_{2}$ by
	\begin{equation*}
		\begin{aligned}
			\mathcal{D}_6(\eta,k)&=\mathcal{D}_4(\eta,\omega^2k),\qquad & k\in\mathbb{C}\setminus\Gamma_6^{(2)},\\
			\mathcal{D}_2(\eta,k)&=\mathcal{D}_4(\eta,\omega k),\qquad & k\in\mathbb{C}\setminus\Gamma_2^{(2)},
		\end{aligned}
	\end{equation*}	
	which satisfy
	\begin{equation*}
		\begin{aligned}
			\mathcal{D}_{6+}(\eta,k)&=\mathcal{D}_{6-}(\eta,k)(1-|\rho_1(\omega^2k)|^2),\quad &k\in\Gamma_6^{(2)},\\
			\mathcal{D}_{2+}(\eta,k)&=\mathcal{D}_{2-}(\eta,k)(1-|\rho_1(\omega k)|^2),\quad &k\in\Gamma_2^{(2)}.
		\end{aligned}
	\end{equation*}

	In the second transformation, take
	\begin{equation}\label{M2}
		\mathcal{M}^{(2)}(x,t,k)=\mathcal{M}^{(1)}(x,t,k)\mathcal{D}(\eta,k),
	\end{equation}
	where
	\begin{equation*}
		\mathcal{D}(\eta,k)=\begin{pmatrix}
			\frac{\mathcal{D}_4(\eta,k)}{\mathcal{D}_6(\eta,k)} & 0 & 0 \\
			0 & \frac{\mathcal{D}_2(\eta,k)}{\mathcal{D}_4(\eta,k)} & 0 \\
			0 & 0 & \frac{\mathcal{D}_6(\eta,k)}{\mathcal{D}_2(\eta,k)}
		\end{pmatrix}.
	\end{equation*}
	\par
	It's infered from Lemma \ref{property of HD} that $\mathcal{D}(\eta,k)$ and $\mathcal{D}^{-1}(\eta,k)$ are uniformly bounded for $\eta\in \mathcal{K}$ and $k\in \mathbb{C}\setminus (\Gamma_{}^{(2)}\cup\Gamma_{4}^{(2)}\cup\Gamma_{6}^{(2)})$ and
	\begin{equation*}
		\mathcal{D}(\eta,k)=I+O(k^{-1})  \qquad  {\rm as } \quad k\to \infty.
	\end{equation*}

	After simple computation, the transformation (\ref{M2}) indicates that the jump matrix for
	$\mathcal{M}^{(2)}(x,t,k)$ is $\mathsf{V}^{(2)}=\mathcal{D}_{-}^{-1}\mathsf{V}^{(1)}\mathcal{D}_{+}$, i.e.,
	\begin{equation*}
		\begin{aligned}
			\mathsf{V}_1^{(2)}&=\begin{pmatrix}
				1-\rho_{2,r}(k)\rho_{2,r}^*(k) & -\frac{\mathcal{D}_{6}\mathcal{D}_{2}}{\mathcal{D}_{4}^2}\rho_{2,r}^*(k)e^{-t\theta_{21}} &0\\
				\frac{\mathcal{D}_{4}^2}{\mathcal{D}_{6}\mathcal{D}_{2}}\rho_{2,r}(k)e^{t\theta_{21}} & 1 & 0\\
				0 & 0 & 1
			\end{pmatrix},\\
			\mathsf{V}_2^{(2)}&=\begin{pmatrix}
				1 & 0 & 0\\
				\frac{\mathcal{D}_{4}^{2}}{\mathcal{D}_{2-}\mathcal{D}_{6}}\beta(\omega k)e^{-t\theta_{21}} & 1 -\rho_1(\omega k)\rho_1^*(\omega k) & -\frac{\mathcal{D}_{4}\mathcal{D}_{6}}{\mathcal{D}_{2-}^2}\frac{\rho_1(\omega k)e^{-t\theta_{32}}}{1-\rho_1(\omega k)\rho_1^*(\omega k)}\\
				-\frac{\mathcal{D}_{4}\mathcal{D}_{2+}}{\mathcal{D}_{6}^{2}}\rho_{2,a}^*(\omega^2 k)e^{t\theta_{31}} & \frac{\mathcal{D}_{2+}^2}{\mathcal{D}_{4}\mathcal{D}_{6}}\frac{\rho_1^*(\omega k)e^{t\theta_{32}}}{1-\rho_1(\omega k)\rho_1^*(\omega k)} & 1
			\end{pmatrix},\\
			\mathsf{V}_3^{(2)}&=\begin{pmatrix}
				1 & 0 & \frac{\mathcal{D}_{6}^{2}}{\mathcal{D}_{2}\mathcal{D}_4}\rho_{2,r}(\omega^2k)e^{-t\theta_{31}} \\
				0 & 1 & 0\\
				-\frac{\mathcal{D}_{2}\mathcal{D}_4}{\mathcal{D}_{6}^{2}}\rho_{2,r}^*(\omega^2k)e^{t\theta_{31}} & 0 & 1-\rho_{2,r}(\omega^2k)\rho_{2,r}^*(\omega^2k)
			\end{pmatrix},
		\end{aligned}
	\end{equation*}
	\begin{equation*}
		\begin{aligned}
			\mathsf{V}_4^{(2)}&=\begin{pmatrix}
				1-\rho_1(k)\rho_1^*(k)  & -\frac{\mathcal{D}_{2}\mathcal{D}_{6}}{\mathcal{D}_{4-}^2}\frac{\rho_1(k)e^{-t\theta_{21}}}{1-\rho_1(k)\rho_1^*(k)} & \frac{\mathcal{D}_{6}^{2}}{\mathcal{D}_{4-}\mathcal{D}_{2}}\beta(k)e^{-t\theta_{31}}\\
				\frac{\mathcal{D}_{4+}^2}{\mathcal{D}_{6}\mathcal{D}_{2}}\frac{\rho_1^*(k)e^{t\theta_{21}}}{1-\rho_1(k)\rho_1^*(k)} & 1 & -\frac{\mathcal{D}_{4+}\mathcal{D}_{6}}{\mathcal{D}_{2}^2}\rho_{2,a}^*(\omega k)e^{-t\theta_{32}}\\
				0 & 0 & 1
			\end{pmatrix},\\
			\mathsf{V}_5^{(2)}&=\begin{pmatrix}
				1 & 0 & 0\\
				0 & 1-\rho_{2,r}(\omega k)\rho_{2,r}^*(\omega k) & -\frac{\mathcal{D}_{4}\mathcal{D}_{6}}{\mathcal{D}_{2}^{2}}\rho_{2,r}^*(\omega k)e^{-t\theta_{32}}\\
				0 & \frac{\mathcal{D}_{2}^{2}}{\mathcal{D}_{4}\mathcal{D}_{6}}\rho_{2,r}(\omega k)e^{t\theta_{32}} & 1
			\end{pmatrix},\\
			\mathsf{V}_6^{(2)}&=\begin{pmatrix}
				1 & -\frac{\mathcal{D}_{2}\mathcal{D}_{6+}}{\mathcal{D}_{4}^{2}}\rho_{2,a}^*(k)e^{-t\theta_{21}} & \frac{\mathcal{D}_{6+}^{2}}{\mathcal{D}_{2}\mathcal{D}_{4}}\frac{\rho_1^*(\omega^2k)e^{-t\theta_{31}}}{1-\rho_1(\omega^2k)\rho_1^*(\omega^2k)} \\
				0 & 1 & 0\\
				-\frac{\mathcal{D}_{2}\mathcal{D}_{4}}{\mathcal{D}_{6-}^{2}}\frac{\rho_1(\omega^2k)e^{t\theta_{31}}}{1-\rho_1(\omega^2k)\rho_1^*(\omega^2k)} & \frac{\mathcal{D}_2^2}{\mathcal{D}_4\mathcal{D}_{6-}}\beta(\omega^2k)e^{t\theta_{32}} & 1-\rho_1(\omega^2k)\rho_1^*(\omega^2k)
			\end{pmatrix},
		\end{aligned}
	\end{equation*}
	\begin{equation*}
		\begin{aligned}
			\mathsf{V}_7^{(2)}&=\begin{pmatrix}
				1 & 0 & 0\\
				\frac{\mathcal{D}_{4}^{2}}{\mathcal{D}_{2}\mathcal{D}_{6}}\beta(\omega k)e^{t\theta_{21}} & 1 & -\frac{\mathcal{D}_{4}\mathcal{D}_{6}}{\mathcal{D}_{2}^2}\rho_1(\omega k)e^{-t\theta_{32}}\\
				\frac{\mathcal{D}_{4}\mathcal{D}_{2}}{\mathcal{D}_{6}^2}\alpha(\omega k)e^{t\theta_{31}} & \frac{\mathcal{D}_{2}^{2}}{\mathcal{D}_{4}\mathcal{D}_{6}}\rho_1^*(\omega k)e^{t\theta_{32}} & 1-\rho_1(\omega k)\rho_1^*(\omega k)
			\end{pmatrix},\\
			\mathsf{V}_8^{(2)}&=\begin{pmatrix}
				1 & -\frac{\mathcal{D}_{6}\mathcal{D}_{2}}{\mathcal{D}_{4}^2}\rho_1(k)e^{-t\theta_{21}} & \frac{\mathcal{D}_{6}^{2}}{\mathcal{D}_{2}\mathcal{D}_4}\beta(k)e^{-t\theta_{31}}\\
				\frac{\mathcal{D}_{4}^{2}}{\mathcal{D}_{2}\mathcal{D}_{6}}\rho_1^*(k)e^{\theta_{21}} & 1-\rho_1(k)\rho_1^*(k) & \frac{\mathcal{D}_{4}\mathcal{D}_{6}}{\mathcal{D}_{2}^2}\alpha(k)e^{-t\theta_{32}}\\
				0 & 0 & 1
			\end{pmatrix},\\
			\mathsf{V}_9^{(2)}&=\begin{pmatrix}
				1-\rho_1(\omega^2k)\rho_1^*(\omega^2k) & \frac{\mathcal{D}_{6}\mathcal{D}_{2}}{\mathcal{D}_{4}^2}\alpha(\omega^2 k)e^{-t\theta_{21}} & \frac{\mathcal{D}_{6}^{2}}{\mathcal{D}_{2}\mathcal{D}_{4}}\rho_1^*(\omega^2k)e^{-t\theta_{31}} \\
				0 & 1 & 0\\
				-\frac{\mathcal{D}_{4}\mathcal{D}_{2}}{\mathcal{D}_{6}^2}\rho_1(\omega^2k)e^{t\theta_{31}} & \frac{\mathcal{D}_{2}^{2}}{\mathcal{D}_{4}\mathcal{D}_{6}}\beta(\omega^2 k)e^{t\theta_{32}}  & 1
			\end{pmatrix},
		\end{aligned}
	\end{equation*}
	where
	\begin{equation*}
		\begin{aligned}
			\alpha(k)&=\rho_{2,a}^*(\omega k)(\rho_1(k)\rho_1^*(k)-1),\qquad &k\in\mathbb{R}^{-},\\
			\beta(k)&=\rho_{2,a}(\omega^2 k)+\rho_{2,a}^*(\omega k)\rho_1(k),\quad &k\in\mathbb{R}^{-}.
		\end{aligned}
	\end{equation*}

	\subsection{The third transformation} Following the procedure of Deift-Zhou nonlinear steepest-descent strategy \cite{Deift-Zhou}, open lenses of the contour $\Gamma^{(2)}$ in Figure \ref{figure of gamma(2)} by introducing an appropriate transformation.
	\par
	Note that the $(11)$-entry of $	\mathsf{V}_4^{(2)}$ can be rewritten as
	\begin{equation*} (\mathsf{V}_4^{(2)})_{11}=1-|\rho_1(k)|^2=1-\frac{\mathcal{D}_{4+}}
		{\mathcal{D}_{4-}}\hat{\rho_1}(k)\frac{\mathcal{D}_{4+}}{\mathcal{D}_{4-}}\hat{\rho_1^*}(k),
	\end{equation*}
	then for $k\in\Gamma_{4}^{(2)}$, the jump matrix $\mathsf{V}_4^{(2)}$ can be factorized as
	\begin{equation*}
		\mathsf{V}_4^{(2)}=\mathsf{V}_4^{(2)U}\mathsf{V}_{4,r}^{(2)}\mathsf{V}_4^{(2)L},
	\end{equation*}
	where
	\begin{equation*}
		\begin{aligned}
			\mathsf{V}_4^{(2)U}&=\begin{pmatrix}
				1 &  -\frac{\mathcal{D}_{2}\mathcal{D}_{6}}{\mathcal{D}_{4-}^2}\hat{\rho}_{1,a}(k)e^{-t\theta_{21}} & \frac{\mathcal{D}_{6}^{2}}{\mathcal{D}_{4-}\mathcal{D}_{2}}\rho_{2,a}(\omega^2 k)e^{-t\theta_{31}}\\
				0 & 1 & 0\\
				0 & 0 & 1
			\end{pmatrix},\\
			\mathsf{V}_{4,r}^{(2)}&=\begin{pmatrix}
				1-\frac{\mathcal{D}_{4+}^2}{\mathcal{D}_{4-}^{2}} \hat{\rho}_{1,r}(k)\hat{\rho}_{1,r}^*(k)  & -\frac{\mathcal{D}_{2}\mathcal{D}_{6}}{\mathcal{D}_{4-}^2}\hat{\rho}_{1,r}(k)e^{-t\theta_{21}} & 0\\
				\frac{\mathcal{D}_{4+}^2}{\mathcal{D}_{6}\mathcal{D}_{2}}\hat{\rho}_{1,r}^*(k)e^{t\theta_{21}} & 1 & 0\\
				0 & 0 & 1
			\end{pmatrix},\\
			\mathsf{V}_4^{(2)L}&=\begin{pmatrix}
				1 &0 & 0\\
				\frac{\mathcal{D}_{4+}^2}{\mathcal{D}_{2}\mathcal{D}_{6}}\hat{\rho}_{1,a}^{*}(k)e^{t\theta_{21}} & 1 & -\frac{\mathcal{D}_{6}\mathcal{D}_{4+}}{\mathcal{D}_{2}^2}\rho_{2,a}^*(\omega k)e^{-t\theta_{32}}\\
				0 & 0 & 1
			\end{pmatrix}.
		\end{aligned}
	\end{equation*}

	For $k\in\Gamma_{8}^{(2)}$, $\mathsf{V}_8^{(2)}$ can be factorized in a similar way as
	\begin{equation*}
		\mathsf{V}_8^{(2)}=\mathsf{V}_8^{(2)L}\mathsf{V}_{8,r}^{(2)}\mathsf{V}_8^{(2)U},
	\end{equation*}
	where
	\begin{equation*}
		\begin{aligned}
			&\mathsf{V}_8^{(2)L}=\begin{pmatrix}
				1 & 0 & 0\\
				\frac{\mathcal{D}_{4}^{2}}{\mathcal{D}_{2}\mathcal{D}_{6}}\rho_{1,a}^*(k)e^{t\theta_{21}} & 1 & -\frac{\mathcal{D}_{4}\mathcal{D}_{6}}{\mathcal{D}_{2}^2}(\rho_{2,a}(\omega k)+\rho_{1}^{*}(k)\rho_{2,a}(\omega^2k))e^{-t\theta_{32}}\\
				0 & 0 & 1
			\end{pmatrix},\\
			&\mathsf{V}_{8,r}^{(2)}=\begin{pmatrix}
				1 & -\frac{\mathcal{D}_{6}\mathcal{D}_{2}}{\mathcal{D}_{4}^2}\rho_{1,r}(k)e^{-t\theta_{21}} & \frac{\mathcal{D}_{6}^{2}}{\mathcal{D}_{2}\mathcal{D}_4}\beta_{r}(k)e^{-t\theta_{31}}\\
				\frac{\mathcal{D}_{4}^{2}}{\mathcal{D}_{2}\mathcal{D}_{6}}\rho_{1,r}^*(k)e^{t\theta_{21}} & 1-|\rho_{1,r}(k)|^2 & \frac{\mathcal{D}_{4}\mathcal{D}_{6}}{\mathcal{D}_{2}^2}\rho_{1,r}^*(k)(\rho_{1,r}\rho_{2,a}^{*}(\omega k)-\rho_{2,a}(\omega^2k))e^{-t\theta_{32}}\\
				0 & 0 & 1
			\end{pmatrix},\\
			&\mathsf{V}_8^{(2)U}=\begin{pmatrix}
				1 & -\frac{\mathcal{D}_{6}\mathcal{D}_{2}}{\mathcal{D}_{4}^2}\rho_{1,a}(k)e^{-t\theta_{21}} & \frac{\mathcal{D}_{6}^{2}}{\mathcal{D}_{2}\mathcal{D}_4}\beta_{a}(k)e^{-t\theta_{31}}\\
				0 & 1 & 0\\
				0 & 0 & 1
			\end{pmatrix},
		\end{aligned}
	\end{equation*}
	and
	\begin{equation*}
		\beta_{r}(k):=\rho_{1,r}\rho_{2,a}^{*}(\omega k), \qquad \beta_{a}(k):=\rho_{2,a}(\omega^2k)+\rho_{1,a}(k)\rho_{2,a}^{*}(\omega k).
	\end{equation*}
	
	Let $A_{j}~(j=1,2,3,4)$ denote the region shown in Figure \ref{gamma (3)} that is the jump contour $\Gamma^{(3)}$ for RH problem $\mathcal{M}^{(3)}(x,t,k)$. Introduce the transformation
	\begin{equation}\label{G(x,t,k)}
		\mathcal{M}^{(3)}(x,t,k)=\mathcal{M}^{(2)}(x,t,k)G(x,t,k),
	\end{equation}
	where
	\begin{equation*}
		G(x,t,k)=\begin{cases}
			\mathsf{V}_{8}^{(2)L},\quad &k\in A_{1},\\
			\mathsf{V}_{4}^{(2)U},\quad &k\in A_{2},\\
			(\mathsf{V}_{4}^{(2)L})^{-1}, &k\in A_{3},\\
			(\mathsf{V}_{8}^{(2)U})^{-1}, &k\in A_{4},\\
			I,  & {\rm elsewhere} \quad in \quad \Omega_{3}\cup\Omega_{4}.
		\end{cases}
	\end{equation*}
	
	\begin{figure}[htp]
		\centering
		\begin{tikzpicture}[>=latex]
			\filldraw[fill=yellow!15!white,draw=white] (-2,0)--(-0.8,1.732*0.8)--(0,0)--(-2,0);
			\draw[fill=yellow!15!white,draw=white] (-2,0)--(-0.8,-1.732*0.8)--(0,0)--(-2,0);
			\draw[fill=yellow!15!white,draw=white] (-2,0)--(-3.2,1.732*0.8)--(-3.8,0)--(-2,0);
			\draw[fill=yellow!15!white,draw=white] (-2,0)--(-3.2,-1.732*0.8)--(-3.8,0)--(-2,0);
			\draw[very thick] (-4,0) to (4,0) node[right]{$\Gamma^{(3)}$};
			\draw[very thick] (-2,-1.732*2) to (2,1.732*2);
			\draw[very thick] (-2,1.732*2) to (2,-1.732*2);
			\draw[<->,very thick] (-3.2,0)node[below]{$8$} to (2.5,0);
			\draw[->,very thick] (0,0) to (-1,0) node[below]{$7$};
			\draw[<->,very thick] (-1.5,-1.732*1.5)node[right=1mm]{$5$} to (1.5,1.732*1.5);
			\draw[->,very thick] (0,0) to (.5,1.732*0.5);
			\draw[->,very thick] (0,0) to (-.5,-1.732*0.5)node[right]{$6$};
			\draw[<->,very thick] (-1.5,1.732*1.5) to (1.5,-1.732*1.5);
			\draw[->,very thick] (0,0) to (.5,-1.732*0.5) ;
			\draw[->,very thick] (0,0) to (-.5,1.732*0.5) ;
			\draw[->,very thick] (0,0) to (1,0);
\draw[very thick] (-3.2,-1.732*0.8)--(-.8,1.732*0.8);
			\draw[<->,very thick] (-2.6,-1.732*0.4)node[right=1mm]{$3$}--(-1.4,1.732*0.4)node[left=1mm]{$1$};
			\draw[very thick] (-3.2,1.732*0.8)--(-.8,-1.732*0.8);
			\draw[<->,very thick] (-2.6,1.732*0.4)node[right=1mm]{$2$}--(-1.4,-1.732*0.4)node[left=1mm]{$4$};
			\draw[very thick] (-.8,-1.732*0.8)--(2.8,-1.2*1.732);
			\draw[<->,very thick] (.1,-1.732*0.9)--(1.9,-1.1*1.732);
			\draw[very thick] (-.8,1.732*0.8)--(2.8,1.2*1.732);
			\draw[<->,very thick] (.1,1.732*0.9)--(1.9,1.1*1.732);
			\draw[very thick] (.4,1.732*2)--(1.6,0);
			\draw[<->,very thick] (.7,1.732*1.5)--(1.3,0.5*1.732);
			\draw[very thick] (.4,-1.732*2)--(1.6,0);
			\draw[<->,very thick] (.7,-1.732*1.5)--(1.3,-0.5*1.732);
			\filldraw[black] (-2,0) node[below=1mm]{$k_{0}$} circle (1.5pt);
			\filldraw[black] (1,1.732)  circle (1.5pt);
			\filldraw[black] (1,-1.732)  circle (1.5pt);
			\node at (1.5,1.6){$\omega^{2}k_{0}$};
			\node at (1.5,-1.6){$\omega k_{0}$};
			\node[blue] at (-.8,0.7){$A_{1}$};
			\node[blue] at (-3.2,0.7){$A_{2}$};
			\node[blue] at (-3.2,-0.7){$A_{3}$};
			\node[blue] at (-.8,-0.7){$A_{4}$};
		\end{tikzpicture}
		\caption{The jump contour $\Gamma^{(3)}$ for RH problem $\mathcal{M}^{(3)}(x,t,k)$.}
		\label{gamma (3)}
	\end{figure}

	\begin{lemma}\label{property of G}
		For $k\in\mathbb{C}\setminus \Gamma^{(3)}$, $t>0$ and $\eta\in \mathcal{K}$, $G(x,t,k)$ is uniformly bounded. Moreover, $G(x,t,k)=I+O(k^{-1})$ as $k\to\infty$.
	\end{lemma}
	
	From Lemma \ref{property of G}, it is inferred that the function $\mathcal{M}^{(3)}(x,t,k)$ solves RH problem \ref{RHP for M} with $j=3$, where for $k\in\Omega_{3}\cup\bar{\Omega}_{4}$ the jump matrix $\mathsf{V}^{(3)}$ is listed as follows
	\begin{equation*}
		\begin{aligned}
			\mathsf{V}_{1}^{(3)}&=(\mathsf{V}_8^{(2)L})^{-1}=\begin{pmatrix}
				1 & 0 & 0\\
				-\frac{\mathcal{D}_{4}^{2}}{\mathcal{D}_{2}\mathcal{D}_{6}}\rho_{1,a}^*(k)e^{t\theta_{21}} & 1 & \frac{\mathcal{D}_{4}\mathcal{D}_{6}}{\mathcal{D}_{2}^2}(\rho_{2,a}(\omega k)+\rho_{1}^{*}(k)\rho_{2,a}(\omega^2k))e^{-t\theta_{32}}\\
				0 & 0 & 1
			\end{pmatrix},\\
			\mathsf{V}_{2}^{(3)}&=\mathsf{V}_4^{(2)U}=\begin{pmatrix}
				1 &  -\frac{\mathcal{D}_{2}\mathcal{D}_{6}}{\mathcal{D}_{4}^2}\hat{\rho}_{1,a}(k)e^{-t\theta_{21}} & \frac{\mathcal{D}_{6}^{2}}{\mathcal{D}_{4}\mathcal{D}_{2}}\rho_{2,a}(\omega^2 k)e^{-t\theta_{31}}\\
				0 & 1 & 0\\
				0 & 0 & 1
			\end{pmatrix},\\
			\mathsf{V}_{3}^{(3)}&=\mathsf{V}_4^{(2)L}=\begin{pmatrix}
				1 & 0 & 0\\
				\frac{\mathcal{D}_{4}^2}{\mathcal{D}_{2}\mathcal{D}_{6}}\hat{\rho}_{1,a}^{*}(k)e^{t\theta_{21}} & 1 & -\frac{\mathcal{D}_{6}\mathcal{D}_{4}}{\mathcal{D}_{2}^2}\rho_{2,a}^*(\omega k)e^{-t\theta_{32}}\\
				0 & 0 & 1
			\end{pmatrix},
		\end{aligned}
	\end{equation*}
	\begin{equation*}
		\begin{aligned}
			&\mathsf{V}_{4}^{(3)}=(\mathsf{V}_8^{(2)U})^{-1}=\begin{pmatrix}
				1 & \frac{\mathcal{D}_{6}\mathcal{D}_{2}}{\mathcal{D}_{4}^2}\rho_{1,a}(k)e^{-t\theta_{21}} & -\frac{\mathcal{D}_{6}^{2}}{\mathcal{D}_{2}\mathcal{D}_4}\beta_{a}(k)e^{-t\theta_{31}}\\
				0 & 1 & 0\\
				0 & 0 & 1
			\end{pmatrix},\\
			&\mathsf{V}_{5}^{(3)}=\mathsf{V}_5^{(2)}=\begin{pmatrix}
				1 & 0 & 0\\
				0 & 1-\rho_{2,r}(\omega k)\rho_{2,r}^*(\omega k) & -\frac{\mathcal{D}_{4}\mathcal{D}_{6}}{\mathcal{D}_{2}^{2}}\rho_{2,r}^*(\omega k)e^{-t\theta_{32}}\\
				0 & \frac{\mathcal{D}_{2}^{2}}{\mathcal{D}_{4}\mathcal{D}_{6}}\rho_{2,r}(\omega k)e^{t\theta_{32}} & 1
			\end{pmatrix},\\ &\mathsf{V}_{6}^{(3)}=\mathsf{V}_{8}^{(2)U}\mathsf{V}_{5}^{(2)}\mathcal{A}^{-1}\mathsf{V}_{8}^{(2)L}\mathcal{A}=\begin{pmatrix}
				1 & -\frac{\mathcal{D}_{2}\mathcal{D}_{6}}{\mathcal{D}_{4}^{2}}g(k)e^{-t\theta_{21}} & \frac{\mathcal{D}_{6}^{2}}{\mathcal{D}_{2}\mathcal{D}_{4}}f(k)e^{-t\theta_{31}}\\
				0 & 1-\rho_{2,r}(\omega k)\rho_{2,r}^{*}(\omega k) & -\frac{\mathcal{D}_{4}\mathcal{D}_{6}}{\mathcal{D}_{2}^{2}}\rho_{2,r}^{*}(\omega k)e^{-t\theta_{32}}\\
				0 & \frac{\mathcal{D}_{2}^{2}}{\mathcal{D}_{4}\mathcal{D}_{6}}\rho_{2,r}(\omega k)e^{t\theta_{32}} & 1
			\end{pmatrix},
		\end{aligned}
	\end{equation*}
	\begin{equation*}
		\begin{aligned}
			\mathsf{V}_{7}^{(3)}&=\mathsf{V}_{8,r}^{(2)}=\begin{pmatrix}
				1 & -\frac{\mathcal{D}_{6}\mathcal{D}_{2}}{\mathcal{D}_{4}^2}\rho_{1,r}(k)e^{-t\theta_{21}} & \frac{\mathcal{D}_{6}^{2}}{\mathcal{D}_{2}\mathcal{D}_4}\beta_{r}(k)e^{-t\theta_{31}}\\
				\frac{\mathcal{D}_{4}^{2}}{\mathcal{D}_{2}\mathcal{D}_{6}}\rho_{1,r}^*(k)e^{t\theta_{21}} & 1-\rho_{2,r}(k)\rho_{2,r}^*(k) & \frac{\mathcal{D}_{4}\mathcal{D}_{6}}{\mathcal{D}_{2}^2}\rho_{1,r}^*(k)(\rho_{1,r}\hat{\rho}_{2,a}^{*}(\omega k)-\rho_{2,a}(\omega^2k))e^{-t\theta_{32}}\\
				0 & 0 & 1
			\end{pmatrix},\\
			\mathsf{V}_{8}^{(3)}&=\mathsf{V}_{4,r}^{(2)}=\begin{pmatrix}
				1-\frac{\mathcal{D}_{4+}^2}{\mathcal{D}_{4-}^{2}} \hat{\rho}_{1,r}(k)\hat{\rho}_{1,r}^*(k)  & -\frac{\mathcal{D}_{2}\mathcal{D}_{6}}{\mathcal{D}_{4-}^2}\hat{\rho}_{1,r}(k)e^{-t\theta_{21}} & 0\\
				\frac{\mathcal{D}_{4+}^2}{\mathcal{D}_{6}\mathcal{D}_{2}}\hat{\rho}_{1,r}^*(k)e^{t\theta_{21}} & 1 & 0\\
				0 & 0 & 1
			\end{pmatrix},
		\end{aligned}
	\end{equation*}
	where
	\begin{equation*}
		\begin{aligned}
			f(k)&=\rho_{2,a}(\omega^2k)+\rho_{1,a}(k)\rho_{2}^{*}(\omega k)+\rho_{1,a}^{*}(\omega^2k),\\
			g(k)&=\rho_{2,r}(\omega k)(\rho_{2,a}(\omega^2k)+\rho_{1,a}(k)\rho_{2,a}^{*}(\omega k))-\rho_{1,a}(1-\rho_{2,r}(\omega k)\rho_{2,r}^{*}(\omega k))\\
			&\quad -\rho_{2,a}^{*}(k)-\rho_{1,a}^{*}(\omega^2k)\rho_{2,a}(\omega k).
		\end{aligned}
	\end{equation*}
	\par
	The properties of $f(k)$ and $g(k)$ are proposed in the following lemma.
	
	\begin{lemma}\label{property of f,g}
		There is a positive constant $C$ such that
		\begin{equation}
			|\partial_{x}^{l}f(k)|\le C|k|e^{\frac{t}{4}|{\rm Re}{\theta_{21}}(\eta,k)|},\qquad |\partial_{x}^{l}g(k)|\le C|k|e^{\frac{t}{4}|{\rm Re}{\theta_{21}}(\eta,k)|},
		\end{equation}
		for $k\in\Gamma^{(3)}$ and $l=0,1$.
	\end{lemma}
	
	\begin{lemma}
		For $k\in\Gamma^{(3)}$ and $\eta\in\mathcal{K}$, except the critical points $\{k_0,\omega k_0,\omega^2k_0\}$, the jump matrix $\mathsf{V}^{(3)}$ uniformly converges to $I$. Furthermore, the matrices $\mathsf{V}^{(3)}_{j}~(j=5,6,7,8)$ possess the following estimations
		\begin{subequations}
			\begin{align}
				&||(1+|\cdot|)\partial_{x}^{i}(\mathsf{V}^{(3)}-I)||_{(L^1\cap L^{\infty })(\Gamma_{5}^{(3)})}\le Ct^{-3/2},\\
				&||(1+|\cdot|)\partial_{x}^{i}(\mathsf{V}^{(3)}-I)||_{L^1(\Gamma_{6}^{(3)})}\le Ct^{-3/2},\\
				&||(1+|\cdot|)\partial_{x}^{i}(\mathsf{V}^{(3)}-I)||_{L^{\infty }(\Gamma_{6}^{(3)})}\le Ct^{-1},\\
				&||(1+|\cdot|)\partial_{x}^{i}(\mathsf{V}^{(3)}-I)||_{(L^1\cap L^{\infty })(\Gamma_{7}^{(3)}\cap \Gamma_{8}^{(3)})}\le Ct^{-3/2},
			\end{align}
		\end{subequations}
		uniformly for $\eta\in\mathcal{K}$ and $i=0,1$.
	\end{lemma}

	\subsection{Local parametrix at the critical point $k_0$}
	Based on analysis above, a RH problem for $\mathcal{M}^{(3)}(x,t,k)$ is achieved, where the jump matrix $\mathsf{V}^{(3)}$ decays to $I$ as $t$ tends to infty everywhere except the region close to three points $\{k_{0}, \omega k_{0}, \omega^2k_{0}\}$. Therefore, as $t$ tends to infty, we only need to consider the neighborhoods of these critical points. So replace the RH problem for $\mathcal{M}^{(3)}(x,t,k)$ with a RH problem for $\mathcal{M}^{k_{0}}(x,t,k)$ that approximates $\mathcal{M}^{(3)}$ near $k_{0}$ as $t\to\infty$ and can be solved exactly.
	
	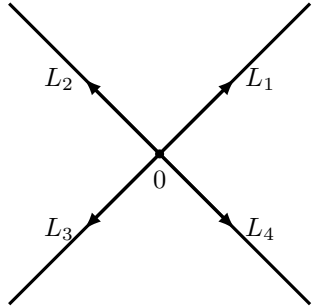
\begin{figure}[htp]
		\centering
		\begin{tikzpicture}[>=latex]
			\draw[very thick] (-2,-2)--(2,2);
			\draw[very thick] (-2,2)--(2,-2);
			\filldraw[black] (0,0) node[below=1mm]{${0}$} circle (1.5pt);
			\draw[<->,very thick] (-1,-1)node[left]{$L_{3}$}--(1,1)node[right]{$L_{1}$};
			\draw[<->,very thick] (-1,1)node[left]{$L_{2}$}--(1,-1)node[right]{$L_{4}$};
		\end{tikzpicture}
		\caption{The contour $\mathbb{L}=L_{1}\cup L_{2}\cup L_{3}\cup L_{4}$.}
		\label{L figure}
	\end{figure}
	
	\par
	Take $\epsilon=\epsilon(\eta)=k_0/2$ and denote $B_{\epsilon}(k_0)$ as an open disk with center $k_0$ and radius $\epsilon$. Set $\mathbb{B}=B_{\epsilon}(k_0)\cup \omega B_{\epsilon}(k_0)\cup\omega^2B_{\epsilon}(k_0)$ and let $L=L_{1}\cup L_{2}\cup L_{3}\cup L_{4}$ be the contour defined in Figure \ref{L figure}. Further, set $\mathbf{L}=k_{0}+L$, $\mathbf{L}^{\epsilon}=\mathbf{L}\cap B_{\epsilon}(k_0)$ and $\mathbf{L}^{\epsilon}_{j}=(k_0+L_{j})\cap B_{\epsilon}(k_0)$ for $j=1,2,3,4$.
	
	\par
	Firstly, let us introduce a new variable $z=z(\eta,k)$ by
	\begin{equation}\label{z}
		z=3^{1/4}\sqrt{2t}(k-k_{0}).
	\end{equation}
	Then the function $\theta_{21}$ can be rewritten as
	\begin{equation*}
		\theta_{21}(\eta,k)=\theta_{21}(\eta,k_0)+\sqrt{3}i(k-k_0)^{2},
	\end{equation*}
	where $\theta_{21}(\eta,k_0)=-\sqrt{3}ik_0^{2}$. It's easy to find
	\begin{equation*}
		t(\theta_{21}(\eta,k)-\theta_{21}(\eta,k_{0}))=\frac{\mathrm{i}}{2}z^{2}.
	\end{equation*}
	For $\eta\in\mathcal{K}$ and $k\in B_{\epsilon}(k_{0})\setminus(-\infty,k_{0}]$, it follows from (\ref{the expansion of D}) and (\ref{z}) that
	\begin{equation*}	\frac{\mathcal{D}_{2}\mathcal{D}_{6}}{\mathcal{D}_{4}^{2}}=e^{2\mathrm{i}\mu\ln_{\pi}(z)}(2\sqrt{3}t)^{-\mathrm{i}\mu}e^{2\kappa_{4}(\eta,k)}\mathcal{D}_{2}\mathcal{D}_{6}=e^{2\mathrm{i}\mu\ln_{\pi}(z)}d_{0}(\eta,t)d_{1}(\eta,k),
	\end{equation*}
	where
	\begin{equation*}
		\begin{aligned} d_{0}(\eta,t)&=(2\sqrt{3}t)^{-\mathrm{i}\mu}e^{2\kappa_{4}(\eta,k_0)}\mathcal{D}_{2}(\eta,k_{0})\mathcal{D}_{6}(\eta,k_0),\\
			d_{1}(\eta,k)&=e^{2\kappa_{4}(\eta,k)-2\kappa_{4}(\eta,k_{0})}\frac{\mathcal{D}_{2}(\eta,k)\mathcal{D}_{6}(\eta,k)}{\mathcal{D}_{2}(\eta,k_{0})\mathcal{D}_{6}(\eta,k_0)}.
		\end{aligned}
	\end{equation*}
	\par
	Introduce the transformation for $\mathcal{M}^{(3)}(x,t,k)$ with $k$ near $k_{0}$ by
	\begin{equation*}
		\widehat{\mathcal{M}}(x,t,k)=\mathcal{M}^{(3)}(x,t,k)T(\eta,t),\qquad k\in B_{\epsilon}(k_{0}),
	\end{equation*}
	where
	\begin{equation*}
		T(\eta,k)=\begin{pmatrix}
			d_0^{1/2}(\eta,t)e^{-\frac{t}{2}\theta_{21}(\eta,k_0)} & 0 & 0\\
			0 & d_0^{-1/2}(\eta,t)e^{\frac{t}{2}\theta_{21}(\eta,k_0)} & 0\\
			0 & 0 & 1
		\end{pmatrix}.
	\end{equation*}
	Then the jump matrix $\widehat{\mathsf{V}}(x,t,k)$ for the function $\widehat{\mathcal{M}}(x,t,k)$ across the contour $\mathbf{L}^{\epsilon}$ is given by
	\begin{equation*}
		\begin{aligned}
			\widehat{\mathsf{V}}_{1}&=\begin{pmatrix}
				1 & 0 & 0\\
				-e^{-2i\nu\ln_{\pi}(z)}d_{1}^{-1}\rho_{1,a}^*(k)e^{\frac{iz^{2}}{2}} & 1 & \frac{\mathcal{D}_{4}\mathcal{D}_{6}}{\mathcal{D}_{2}^2}d_0^{1/2}(\rho_{2,a}(\omega k)+\rho_{1}^{*}(k)\rho_{2,a}(\omega^2k))e^{-\frac{t}{2}\theta_{21}(\eta,k_0)}e^{-t\theta_{32}}\\
				0 & 0 & 1
			\end{pmatrix},\\
			\widehat{\mathsf{V}}_{2}&=\begin{pmatrix}
				1 &  -e^{2i\mu\ln_{\pi}(z)}d_{1}\hat{\rho}_{1,a}(k)e^{-\frac{iz^{2}}{2}} & \frac{\mathcal{D}_{6}^{2}}{\mathcal{D}_{4}\mathcal{D}_{2}}d_0^{-1/2}\rho_{2,a}(\omega^2 k)e^{\frac{t}{2}\theta_{21}(\eta,k_0)}e^{-t\theta_{31}}\\
				0 & 1 & 0\\
				0 & 0 & 1
			\end{pmatrix},\\
			\widehat{\mathsf{V}}_{3}&=\begin{pmatrix}
				1 & 0 & 0\\
				e^{-2i\mu\ln_{\pi}(z)}d_{1}^{-1}\hat{\rho}_{1,a}^{*}(k)e^{\frac{iz^{2}}{2}} & 1 & -\frac{\mathcal{D}_{6}\mathcal{D}_{4}}{\mathcal{D}_{2}^2}d_0^{1/2}\rho_{2,a}^*(\omega k)e^{-\frac{t}{2}\theta_{21}(\eta,k_0)}e^{-t\theta_{32}}\\
				0 & 0 & 1
			\end{pmatrix},\\
			\widehat{\mathsf{V}}_{4}&=\begin{pmatrix}
				1 & e^{2i\mu\ln_{\pi}(z)}d_{1}\rho_{1,a}(k)e^{-\frac{iz^{2}}{2}} & -\frac{\mathcal{D}_{6}^{2}}{\mathcal{D}_{2}\mathcal{D}_4}d_0^{-1/2}\beta_{a}(k)e^{\frac{t}{2}\theta_{21}(\eta,k_0)}e^{-t\theta_{31}}\\
				0 & 1 & 0\\
				0 & 0 & 1
			\end{pmatrix},
		\end{aligned}
	\end{equation*}
	where $\widehat{\mathsf{V}}_{j}$ represents the restriction of $\widehat{\mathsf{V}}$ to $\mathbf{L}^{\epsilon}_{j}$ for $j=1,2,3,4$.
	\par
	Set $p(\eta)=\rho_1(k_0)$.
	Fixing the parameter $z$, we have $\rho_{1,a}(k)\to p$, $\hat{\rho}_{1,a}^{*}\to \frac{\bar{p}}{1-|p|^2}$ and $d_1\to 1$ as $t\to\infty$. Further, define $\mu:~\mathbb{D}\to(0,\infty )$ by $\mu(p)=-\frac{1}{2\pi}\ln(1-|p|^{2})$, where $\mathbb{D}$ denotes an open unit disk on $\mathbb{C}$.
	Then the jump matrix $\widehat{\mathsf{V}}(x,t,k)$ for the function $\widehat{\mathcal{M}}(x,t,k)$ across $\mathbf{L}^{\epsilon}$ decays to $\mathsf{V}^{L}(p,z)$, which is expressed by
	\begin{equation}\label{Jump-for-ML}
		\begin{aligned}
			\mathsf{V}^{L}_{1}&=\begin{pmatrix}
				1 & 0 & 0\\
				-\bar{p}z^{-2i\mu(p)}e^{\frac{iz^{2}}{2}} & 1 & 0\\
				0 & 0 & 1
			\end{pmatrix}, \quad
			\mathsf{V}^{L}_{2}=\begin{pmatrix}
				1 &  -\frac{p}{1-|p|^2}z^{2i\mu(p)}e^{-\frac{iz^{2}}{2}} & 0\\
				0 & 1 & 0\\
				0 & 0 & 1
			\end{pmatrix},\\
			\mathsf{V}^{L}_{3}&=\begin{pmatrix}
				1 & 0 & 0\\
				\frac{\bar{p}}{1-|p|^2}z^{-2i\mu(p)}e^{\frac{iz^{2}}{2}} & 1 & 0\\
				0 & 0 & 1
			\end{pmatrix},\quad
			\mathsf{V}^{L}_{4}=\begin{pmatrix}
				1 & pz^{2i\mu(p)}e^{-\frac{iz^{2}}{2}} & 0\\
				0 & 1 & 0\\
				0 & 0 & 1
			\end{pmatrix},
		\end{aligned}
	\end{equation}
	with $z^{2i\mu(p)}=e^{2i\mu(p){\ln_{\pi}(z)}}$ for choosing the branch cut running along $\R_-$, as $t\to\infty$.
	\par
	In what follow, a RH problem for $\mathcal{M}^{L}(p,z)$ is constructed, which can be exactly solved using the parabolic cylinder functions.
	\begin{RH problem}[RH problem for $\mathcal{M}^{L}(p,z)$]\label{RH problem for ML}
		The $3 \times 3$ matrix-valued function $\mathcal{M}^{L}(p,z)$ satisfies the following properties:
		\begin{enumerate}[$(1)$]
			\item $\mathcal{M}^{L}(p, \cdot)$ is analytic for $z \in\C\setminus \mathbf{L}$.
		\end{enumerate}
		
		\begin{enumerate}[$(2)$]
			\item $\mathcal{M}^{L}(p,z)$ is continuous to $\mathbf{L}\setminus\{0\}$ and meets the jump condition below:
			\begin{equation*}
				\mathcal{M}^{L}_{+}(p,z)=\mathcal{M}^{L}_{-}(p,z)\mathsf{V}^{L}_j(p,z),\quad z\in\mathbf{L}\setminus \{0\},
			\end{equation*}
			where the jump functions $\mathsf{V}_j^{L}(p,z)~(j=1,2,3,4)$ are defined in (\ref{Jump-for-ML}).
		\end{enumerate}
		
		\begin{enumerate}[$(3)$]
			\item $\mathcal{M}^{L}(p,z)\to I$  as $z\to\infty$.
		\end{enumerate}
		
		\begin{enumerate}[$(4)$]
			\item $\mathcal{M}^{L}(p,z)\to O(1)$ as $z\to 0$.
		\end{enumerate}
		
	\end{RH problem}

	The following theorem studies the RH problem \ref{RH problem for ML} in a standard way.
	
	\begin{theorem}\label{solution to ML}
		There is a unique solution $\mathcal{M}^{L}(p,z)$ of the RH problem \ref{RH problem for ML} for each $p\in\mathbb{D}$. Moreover, the solution
		$\mathcal{M}^{L}(p,z)$ satisfies
		\begin{equation*}
			\mathcal{M}^{L}(p, z)=I+\frac{\mathcal{M}_{1}^{L}(p)}{z}+O(\frac{1}{z^{2}}),\qquad z\to\infty,\quad p\in\mathbb{D},
		\end{equation*}
		and the matrix $\mathcal{M}_{1}^{L}(p)$ is defined below
		\begin{equation*}
			\mathcal{M}_{1}^{L}(p)=\begin{pmatrix}
				0 & \alpha_{12} &  0\\
				\alpha_{21} & 0 & 0\\
				0 & 0 & 0
			\end{pmatrix},\qquad p\in\mathbb{D},
		\end{equation*}
		where the functions $\alpha_{12}$ and $\alpha_{21}$ are given by
		\begin{equation*}
			\alpha_{12}=\frac{\sqrt{2\pi}e^{-\frac{\pi i}{4}}e^{-\frac{\pi\mu}{2}}}{\bar{p}\Gamma(-i\mu)},\quad \alpha_{21}=\frac{\sqrt{2\pi}e^{\frac{\pi i}{4}}e^{\frac{-\pi\mu}{2}}}{p\Gamma(i\mu)}.
		\end{equation*}

	\end{theorem}
	
	\begin{proof}
		Denote
		\begin{equation*}
			z^{i\mu\tilde{\sigma_3}}=\begin{pmatrix}
				z^{i\mu} & 0 & 0\\
				0 & z^{-i\mu} & 0\\
				0 & 0 & 1
			\end{pmatrix},\quad
			e^{\frac{iz^{2}}{4}\tilde{\sigma_{3}}}=\begin{pmatrix}
				e^{\frac{iz^{2}}{4}} & 0 & 0\\
				0 & e^{-\frac{iz^{2}}{4}} & 0\\
				0 & 0& 1
			\end{pmatrix},
		\end{equation*}
and introduce the transformation
		\begin{equation*}
			P(p,z)=\mathcal{M}^{L}(p,z)\mathcal{H}(p,z),
		\end{equation*}
		where the matrix $\mathcal{H}(p,z)$ is sectionally analytic function of the form
		\begin{equation*}
			\mathcal{H}(p, z)= \begin{cases}
				\begin{pmatrix}
					1 & 0 & 0\\
					-\bar{p}z^{-2i\mu(p)}e^{\frac{iz^{2}}{2}} & 1 & 0\\
					0 & 0 & 1
				\end{pmatrix}z^{i\mu\tilde\sigma_3} , & z \in \Omega_1, \\
				
				z^{i\mu\tilde\sigma_3}, & z \in \Omega_2, \\
				
				\begin{pmatrix}
					1 &  \frac{p}{1-|p|^2}z^{2i\mu(p)}e^{-\frac{iz^{2}}{2}} & 0\\
					0 & 1 & 0\\
					0 & 0 & 1
				\end{pmatrix}z^{i\nu\tilde\sigma_3}, & z \in \Omega_3, \\
				
				\begin{pmatrix}
					1 & 0 & 0\\
					\frac{\bar{p}}{1-|p|^2}z^{-2i\mu(p)}e^{\frac{iz^{2}}{2}} & 1 & 0\\
					0 & 0 & 1
				\end{pmatrix}z^{i\mu\tilde\sigma_3}, & z \in \Omega_4, \\
				z^{i\mu\tilde\sigma_3}, & z\in \Omega_5,\\
				
				\begin{pmatrix}
					1 & -pz^{2i\mu(p)}e^{-\frac{iz^{2}}{2}} & 0\\
					0 & 1 & 0\\
					0 & 0 & 1
				\end{pmatrix}z^{i\mu\tilde\sigma_3}, & z\in \Omega_6.\\
				
			\end{cases}
		\end{equation*}
		\par		
		
		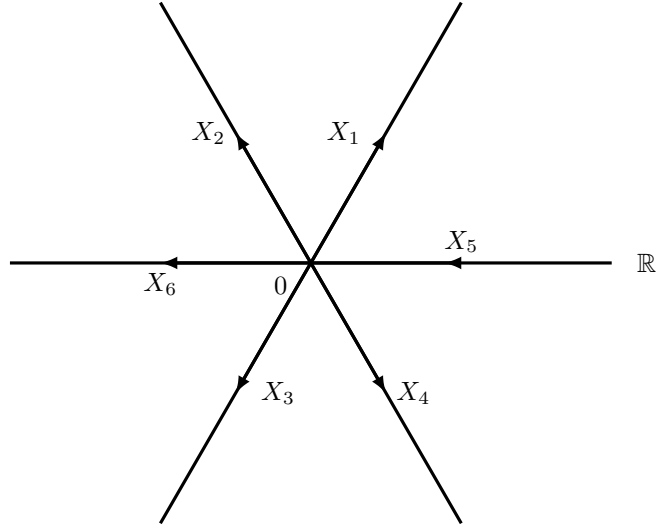
\begin{figure}[htp]
		\centering
		\begin{tikzpicture}[>=latex]
			\draw[<-<,very thick] (-2,0)node[below]{${X_6}$} to (2,0)node[above]{${X_5}$};
			\draw[very thick] (-4,0) to (4,0)node[right=2mm]{$\mathbb{R}$};
			\draw[<->,very thick] (-1,-1.732)node[right=2mm]{${X_3}$} to (1,1.732)node[left=2mm]{${X_1}$};
			\draw[very thick] (-2,-1.732*2) to (2,1.732*2);
			\draw[<->,very thick] (-1,1.732)node[left]{${X_2}$} to (1,-1.732)node[right]{${X_4}$};
			\draw[very thick] (-2,1.732*2) to (2,-1.732*2);
			\node at (-.4,-.3){$0$};
		\end{tikzpicture}
		\caption{The new jump contour $X=X_{1}\cup X_{2}\cup \cdots \cup X_{6}$ for the function $P(p,z)$.}
		\label{Figure-new-jump}
	\end{figure}

Direct computations indicate that
		\begin{equation*}
			P_+(p,z)=P_-(p,z)v^{X}_j(p,z),\quad j=1,2,3,4,
		\end{equation*}
		where the jump matrices $v^{X}_j(p,z)=(\mathcal{H}_-(p,z))^{-1}\mathsf{V}^{L}_j(z,p)\mathcal{H}_+(p,z)~(j=1,2,3,4)$ are
		\begin{equation*}
			v^{X}_j(p,z)=\begin{cases}
				I,\quad & z\in X_1\cup X_2 \cup X_3 \cup X_4,\\
				e^{-\frac{iz^2}{4}\operatorname{ad} \tilde\sigma_3}\left(\begin{array}{ccc}
					1& -p & 0\\
					\bar p & 1-|p|^2 &0\\
					0 & 0 & 1
				\end{array}
				\right),\quad & z\in X_5,\\
				e^{-\frac{iz^2}{4}\operatorname{ad} \tilde\sigma_3}\left(\begin{array}{ccc}
					1& -p & 0\\
					\bar p & 1-|p|^2 &0\\
					0 & 0 & 1
				\end{array}
				\right),\quad & z\in X_6,
			\end{cases}
		\end{equation*}
where the new contour $X=X_{1}\cup X_{2}\cup \cdots \cup X_{6}$ for the function $P(p,z)$ is shown in Figure \ref{Figure-jump}. In fact, the function $P(p,z)$ only has jump  on $\R$. In particular, the RH problem for $P(p,z)$ can be written as
		\begin{equation*}
			\begin{cases}
				P_+(p,z)=P_-(p,z)e^{-\frac{iz^2}{4}\operatorname{ad} \tilde\sigma_3}V(p),&z\in\R,\\
				P\to z^{i\nu\tilde\sigma_3}, &\text{as}\quad z \to \infty.
			\end{cases}
		\end{equation*}
		where
		\begin{equation*}
			V=\left(\begin{array}{ccc}
				1 & -p  & 0 \\
				\bar p  &  1-|p|^2 & 0 \\
				0 & 0 & 1
			\end{array}\right).
		\end{equation*}
		Since $\mathcal H=\left(I+O\left(\frac{1}{z}\right)\right)z^{i\mu\tilde \sigma_3}$, by the procedure in \cite{Deift-Zhou}, it follows
		\begin{equation*}		\Psi(p,z)=P(p,z)e^{-\frac{iz^2\tilde\sigma_3}{4}}=\hat{\Psi}(p,z)z^{i\mu\tilde\sigma_3}e^{-\frac{iz^2\tilde\sigma_3}{4}},
		\end{equation*}
		with the asymptotics
		\begin{equation}\label{Pz-Behavior}
			\begin{aligned}			&P(p,z)=\left(I+\frac{\mathcal{M}_1^L(p)}{z}+O\left(\frac{1}{z^2}\right)\right)z^{i\mu\tilde\sigma_3},\\
			&\hat{\Psi}(p,z)=I+\frac{\mathcal{M}_1^L(p)}{z}+O\left(\frac{1}{z^2}\right),
		\end{aligned}
		\end{equation}
		as $z\to\infty$.
		After some calculations, we have
		\begin{equation*}
			\Psi_+(z)=\Psi_-(z)V(p),\qquad z\in\R,
		\end{equation*}
		and
		\begin{equation*}
			\left(\partial_z\Psi+\frac{iz\tilde \sigma_3}{2}\Psi\right)_+=\left(\partial_z\Psi+\frac{iz\tilde \sigma_3}{2}\Psi\right)_-V(p).
		\end{equation*}
It is claimed that $\left(\partial_z\Psi+\frac{iz\tilde \sigma_3}{2}\Psi\right)\Psi^{-1}$ has no jump along the $\R$. Therefore, it is an entire function on the $\C$. Then one has
		\begin{equation*}
			\begin{aligned}
				\left(\partial_z\Psi+\frac{iz\tilde \sigma_3}{2}\Psi\right)\Psi^{-1}&=(\partial_z\hat\Psi)\hat\Psi^{-1}+\hat\Psi(i\nu\tilde\sigma_3z^{-1})\hat\Psi^{-1}+\hat\Psi\left(-\frac{iz}{2}\tilde\sigma_3\right)\hat\Psi^{-1}+\left(\frac{iz}{2}\tilde\sigma_3\hat\Psi\right)\hat\Psi^{-1}\\
				&=\frac{iz}{2}\left[\tilde\sigma_3,\hat\Psi\right]\hat\Psi^{-1}+O\left(\frac{1}{z}\right)\\
				&=\frac{i}{2}\left[\tilde\sigma_3,\mathcal{M}_1^L\right]+O\left(\frac{1}{z}\right).
			\end{aligned}
		\end{equation*}
		Let
		\begin{equation*}
			\frac{i}{2}\left[\tilde\sigma_3,\mathcal{M}_1^L\right]:=\begin{pmatrix}
				0 &  \tilde\alpha_{12} & 0 \\
				\tilde\alpha_{21} & 0 & 0 \\
				0 & 0 & 0
			\end{pmatrix}
		\end{equation*}
		where $\alpha_{12}=-i\tilde\alpha_{12}$ and $\alpha_{21}=i\tilde\alpha_{21}$.
		Then by Liouville's argument, it yields
		\begin{equation*}
			\partial_z\Psi+\frac{iz\tilde \sigma_3}{2}\Psi=\frac{i}{2}\left[\tilde\sigma_3,\mathcal{M}_1^L\right]\Psi,
		\end{equation*}
which can be rewritten as
		\begin{equation*}
			\frac{\partial \Psi_{11}}{\partial z}+\frac{1}{2} i z \Psi_{11}=\tilde\alpha_{12} \Psi_{21},\\\
			\frac{\partial \Psi_{21}}{\partial z}-\frac{1}{2} i z \Psi_{21}=\tilde\alpha_{21} \Psi_{11}.
		\end{equation*}
		Then we have
		\begin{equation*}
			\frac{\partial^2 \Psi_{11}}{\partial z^2}+\left(\frac{1}{2} i+\frac{1}{4} z^2-\tilde\alpha_{12} \tilde\alpha_{21}\right) \Psi_{11}=0,
		\end{equation*}
		and similarly
		\begin{equation*}
			\frac{\partial^2 \Psi_{22}}{\partial z^2}+\left(-\frac{1}{2} i+\frac{1}{4} z^2-\tilde\alpha_{12} \tilde\alpha_{21}\right) \Psi_{22}=0.
		\end{equation*}
\par
In what follow, let us focus on the upper half plane ${\rm Im}{z}>0$ and denote $z=e^{\frac{3}{4} \pi \mathrm{i}} \xi$, then it follows $\Psi_{11}^{+}(z)=\Psi_{11}^{+}\left(e^{\frac{3}{4} \pi i} \xi\right)=g(\xi)$ with $-\frac{3\pi}{4}<\arg{(\xi)}<\frac{\pi}{4},\frac{\pi}{4}<\arg{(-\xi)}<\frac{5\pi}{4}$. So the Weber equation for variable $g(\xi)$ yields
		\begin{equation}\label{Weber-equation}
			\frac{d^2 g(\xi)}{d \xi^2}+\left(\frac{1}{2}-\frac{1}{4} \xi^2+a\right) g(\xi)=0,
		\end{equation}
where $a=i\tilde\alpha_{12}\tilde\alpha_{21}$. Standard analysis follows the exact solution of the Weber equation
		\begin{equation*}
			g(\xi)=c_1 D_a(\xi)+c_2 D_a(-\xi),
		\end{equation*}
where $D_a(\xi)$ behaves
 		\begin{equation*}
			D_a(\xi)=\left\{\begin{array}{l}
				\xi^a e^{-\frac{1}{4} \xi^2}\left(1+O\left(\frac{1}{\xi^2}\right)\right), \quad|\arg \xi|<\frac{3}{4} \pi, \\
				\xi^a e^{-\frac{1}{4} \xi^2}\left(1+O\left(\frac{1}{\xi^2}\right)\right)-\frac{\sqrt{2 \pi}}{\Gamma(-a)} e^{a \pi \mathrm{i}} \xi^{-a-1} e^{\frac{1}{4} \xi^2}\left(1+O\left(\frac{1}{\xi^2}\right)\right), \quad \frac{\pi}{4}<\arg \xi<\frac{5 \pi}{4}, \\
				\xi^a e^{-\frac{1}{4} \xi^2}\left(1+O\left(\frac{1}{\xi^2}\right)\right)-\frac{\sqrt{2 \pi}}{\Gamma(-a)} e^{-a \pi \mathrm{i}} \xi^{-a-1} e^{\frac{1}{4} \xi^2}\left(1+O\left(\frac{1}{\xi^2}\right)\right), \quad-\frac{5 \pi}{4}<\arg \xi<-\frac{\pi}{4},
			\end{array}\right.
		\end{equation*}
in which $D_a(\xi)$ denotes the parabolic cylinder function. The behavior of $P(p,z)$ in (\ref{Pz-Behavior}) as $z\to\infty$ implies that $g(\xi)\to z^{i\mu}e^{-\frac{iz^2}{4}}$, which indicates that
		\begin{equation*}
			a=\mathrm{i}\mu,\quad c_1=(e^{\frac{3\pi \mathrm{i}}{4}})^{\mathrm{i} \mu}=e^{-\frac{3\pi \mu}{4}},\quad c_2=0
		\end{equation*}
		In other word, one has
		\begin{equation*}
			\Psi_{11}^{+}(z)=e^{-\frac{3}{4} \pi \mu} D_a\left(e^{-\frac{3}{4} \pi \mathrm{i}} z\right), \quad a=\mathrm{i}\mu.
		\end{equation*}
\par
		For the $\mathrm{Im} z<0$, by the same procedure, taking $z=e^{-\frac{1}{4} \pi \mathrm{i}} \xi$, it follows $\Psi_{11}^{-}(z)=\Psi_{11}^{-}\left(e^{-\frac{1}{4} \pi \mathrm{i}} \xi\right)=g(\xi)$ with $-\frac{3\pi}{4}<\arg{(\xi)}<\frac{\pi}{4}$, then we get
		\begin{equation*}
\Psi_{11}^{-}(z)=e^{\frac{1}{4} \pi \mu} D_a\left(e^{\frac{1}{4} \pi \mathrm{i}} z\right),\quad			a=\mathrm{i}\mu,\quad c_2=0,\quad c_1=(e^{-\frac{\pi \mathrm{i}}{4}})^{\mathrm{i} \mu}=e^{\frac{\pi \mu}{4}}.
		\end{equation*}
\par
In conclusion, it is seen that
		\begin{equation*}
			\Psi_{11}(p, z)= \begin{cases}e^{\frac{-3 \pi \mu}{4}} D_{\mathrm{i} \mu}\left(e^{-\frac{3 \pi \mathrm{i}}{4}} z\right), & \mathrm{Im} z>0 ,\\ e^{\frac{ \pi \mu}{4}} D_{\mathrm{i} \mu}\left(e^{\frac{\pi \mathrm{i}}{4}} z\right), & \mathrm{Im} z<0,\end{cases}
		\end{equation*}
which follows that
		\begin{equation*}
			\begin{aligned}
				& \Psi_{21}^{+}(z)=e^{-\frac{3}{4} \pi \mu}\left(\tilde\alpha_{12}\right)^{-1}\left[\partial_z D_a\left(e^{-\frac{3 \pi \mathrm{i}}{4}} z\right)+\frac{\mathrm{i} z}{2} D_a\left(e^{-\frac{3 \pi \mathrm{i}}{4}} z\right)\right], \\
				& \Psi_{21}^{-}(z)=e^{\frac{1}{4} \pi \mu}\left(\tilde\alpha_{12}\right)^{-1}\left[\partial_z D_{a}\left(e^{-\frac{\pi \mathrm{i}}{4}} z\right)+\frac{\mathrm{i} z}{2} D_{a}\left(e^{-\frac{\pi\mathrm{i}}{4}} z\right)\right].
			\end{aligned}
		\end{equation*}
Since $\left(\Psi_-\right)^{-1}\Psi_+=V(p)$ and $\det{\Psi}=1$, we further derive that
		\begin{equation*}
			\bar {q}=\Psi^-_{11}\Psi^+_{21}-\Psi^-_{21}\Psi^+_{11}=\frac{e^{-\frac{\pi\mu}{2}}}{\tilde\alpha_{12}}W\left(D_{\mathrm{i}\mu}(e^{\frac{\pi i}{4}}z),D_{\mathrm{i}\mu}(e^{-\frac{3\pi i}{4}}z)\right)=\frac{\sqrt{2\pi}e^{\frac{\pi i}{4}}e^{-\frac{\pi\mu}{2}}}{\tilde\alpha_{12}\Gamma(-\mathrm{i}\mu)},
		\end{equation*}
which shows that
		\begin{equation*}
			\alpha_{12}=-\mathrm{i}\tilde\alpha_{12}=-i\frac{\sqrt{2\pi}e^{\frac{\pi i}{4}}e^{-\frac{\pi\mu}{2}}}{\bar{p}\Gamma(-i\nu)}=\frac{\sqrt{2\pi}e^{-\frac{\pi \mathrm{i}}{4}}e^{-\frac{\pi\mu}{2}}}{\bar{p}\Gamma(-\mathrm{i}\mu)}.
		\end{equation*}
\par
On the other hand, for the equation
		\begin{equation*}
			\frac{\partial^2 \Psi_{22}}{\partial z^2}+\left(-\frac{1}{2} \mathrm{i}+\frac{1}{4} z^2-\tilde\alpha_{12} \tilde\alpha_{21}\right) \Psi_{22}=0,
		\end{equation*}
when $\operatorname{Im} z>0$, taking $z=e^{\frac{1}{4} \pi i} \xi$ and $\Psi_{22}^{+}(z)=\Psi_{22}^{+}\left(e^{\frac{1}{4} \pi i} \xi\right)=g(\xi)$ for $-\frac{\pi}{4}<\arg{\xi}<\frac{3\pi}{4}$, the other Weber equation is obtained below
		\begin{equation*}
			\frac{d^2 g(\xi)}{d \xi^2}+\left(\frac{1}{2}-\frac{1}{4} \xi^2+a\right) g(\xi)=0,
		\end{equation*}
		where $a=-\mathrm{i}\tilde{\alpha}_{12}\tilde{\alpha}_{21}$ and $g(\xi)\to z^{-\mathrm{i}\mu}e^{\frac{\mathrm{i}z^2}{4}}$ as $z\to \infty$.
		Then it is obvious that the parameters are
		\begin{equation*}
			a=-\mathrm{i}\mu,\quad c_1=(e^{\frac{\pi \mathrm{i}}{4}})^{-\mathrm{i} \mu}=e^{\frac{\pi \mu}{4}},\quad c_2=0.
		\end{equation*}
		When $\operatorname{Im} z<0$, taking $z=e^{-\frac{3}{4} \pi \mathrm{i}} \xi$ and $\Psi_{22}^{-}(z)=\Psi_{22}^{-}\left(e^{-\frac{3}{4} \pi \mathrm{i}} \xi\right)=g(\xi)$, it follows
		\begin{equation*}
			a=-\mathrm{i}\mu,\quad c_1=(e^{-\frac{3\pi \mathrm{i}}{4}})^{-\mathrm{i} \mu}=e^{-\frac{3\pi \mu}{4}},\quad c_2=0.
		\end{equation*}
\par
Finally, it is obtained that
		\begin{equation*}
			\Psi_{22}(p, z)= \begin{cases}e^{\frac{\pi \mu}{4}} D_{-\mathrm{i} \mu}\left(e^{-\frac{\pi \mathrm{i}}{4}} z\right), & \mathrm{Im} z>0 ,\\ e^{-\frac{3 \mu}{4}} D_{-\mathrm{i} \mu}\left(e^{\frac{3 \pi \mathrm{i}}{4}} z\right), & \mathrm{Im} z<0,\end{cases}
		\end{equation*}
which indicates that
\begin{equation*}
			\frac{\partial \Psi_{22}}{\partial z}-\frac{1}{2} \mathrm{i} z \Psi_{22}=\tilde\alpha_{21} \Psi_{12}, \quad \quad
			\frac{\partial \Psi_{12}}{\partial z}+\frac{1}{2} \mathrm{i} z \Psi_{12}=\tilde\alpha_{12} \Psi_{22},
		\end{equation*}
		and
		\begin{equation*}
			\begin{aligned}
				& \Psi_{12}^{+}(z)=e^{\frac{1}{4} \pi \mu}\left(\tilde\alpha_{21}\right)^{-1}\left[\partial_z D_{-a}\left(e^{-\frac{ \pi \mathrm{i}}{4}} z\right)-\frac{\mathrm{i} z}{2} D_{-a}\left(e^{-\frac{ \pi \mathrm{i}}{4}} z\right)\right], \\
				& \Psi_{12}^{-}(z)=e^{-\frac{3}{4} \pi \mu}\left(\tilde\alpha_{21}\right)^{-1}\left[\partial_z D_{-a}\left(e^{\frac{3\pi \mathrm{i}}{4}} k\right)-\frac{\mathrm{i} z}{2} D_{-a}\left(e^{\frac{3\pi\mathrm{i}}{4}} z\right)\right] .
			\end{aligned}
		\end{equation*}
Since $\left(\Psi_-\right)^{-1}\Psi_+=V(p)$ and $\det{\Psi}=1$, it is seen that
		\begin{equation*}		-q=\Psi^-_{22}\Psi^+_{12}-\Psi^-_{12}\Psi^+_{22}=\frac{e^{-\frac{\pi\mu}{2}}}{\tilde\alpha_{21}}W\left(D_{-\mathrm{i}\mu}(e^{\frac{3\pi \mathrm{i}}{4}}z),D_{-\mathrm{i}\mu}(e^{-\frac{\pi \mathrm{i}}{4}}z)\right)=\frac{\sqrt{2\pi}e^{\frac{3\pi \mathrm{i}}{4}}e^{-\frac{\pi\mu}{2}}}{\tilde\alpha_{21}\Gamma(\mathrm{i}\mu)},
		\end{equation*}
which denotes that
		\begin{equation*}
			\alpha_{21}=\mathrm{i}\tilde\alpha_{21}=-\mathrm{i}\frac{\sqrt{2\pi}e^{\frac{3\pi \mathrm{i}}{4}}e^{-\frac{\pi\mu}{2}}}{p\Gamma(\mathrm{i}\mu)}=\frac{\sqrt{2\pi}e^{\frac{\pi \mathrm{i}}{4}}e^{-\frac{\pi\mu}{2}}}{p\Gamma(\mathrm{i}\mu)}.
		\end{equation*}
	\end{proof}

	\begin{remark} In the proof of Theorem \ref{solution to ML}, we have chosen $e^{-\frac{3}{4} \pi \mathrm{i}}$ and $e^{-\frac{\pi \mathrm{i}}{4}}$ since the branch cut from $-\pi$ to $\pi$.
	\end{remark}

	The analysis above suggests that the jump condition for the function $\mathcal{M}^{(3)}(x,t,k)$ tends to that for the function $\widehat{\mathcal{M}}(x,t,k)T^{-1}(\eta,t)$ for $k$ near the critical point $k_0$ as $t\to\infty$. Therefore, introduce the transformation
	\begin{equation}\label{definition ofM k0}
		\mathcal{M}^{k_{0}}(x,t,k)=T(\eta,t)\mathcal{M}^{L}(p(\eta),z(\eta,k))T(\eta,t)^{-1},\quad k\in B_{\epsilon}(k_{0}),
	\end{equation}
	in which the function $T(\eta,t)$ has the estimation below.
	\begin{lemma}
		There is a positive constant $C$ such that
		\begin{equation*}
			\sup_{\eta\in\mathcal{K}}\sup_{t\ge2}|\partial_{x}^{i}T(\eta,t)^{\pm1}|<C,\qquad i=0,1.
		\end{equation*}
		Furthermore, $d_{0}(\eta,t)$ and $d_{1}(\eta,k)$ in matrix $T(\eta,t)$ satisfy
		\begin{equation*}
			\begin{aligned}
				&|d_{0}(\eta,t)|=e^{2\pi\mu},\qquad \eta\in\mathcal{K},\quad t\ge 2,\\
				&|\partial_{x}d_{0}(\eta,t)|\le C \frac{\ln{t}}{t},\qquad \eta\in\mathcal{K},\quad t\ge 2,
			\end{aligned}
		\end{equation*}
		and
		\begin{equation*}
			\begin{aligned}
				&|d_{1}(\eta,k)-1|\le C|k-k_{0}|(1+|\ln{|k-k_{0}|}|), &\eta\in\mathcal{K},k\in\mathbf{L}^{\epsilon},\\
				&|\partial_{x}d_{1}(\eta,k)|\le C \frac{|\ln{|k-k_{0}|}|}{t}, &\eta\in\mathcal{K},k\in\mathbf{L}^{\epsilon}.\\
			\end{aligned}
		\end{equation*}
	\end{lemma}

	\begin{lemma}
		The function $\mathcal{M}^{k_{0}}(x,t,k)$ defined in (\ref{definition ofM k0}) is analytic and bounded for  $k\in B_{\epsilon}(k_{0})\setminus \mathbf{L}^{\epsilon}$ and every $(x,t)$. For $k\in\mathbf{L}^{\epsilon}$, $\mathcal{M}^{k_{0}}(x,t,k)$ follows the jump condition $\mathcal{M}^{k_{0}}_{+}(x,t,k)=\mathcal{M}^{k_{0}}_{-}(x,t,k)\mathsf{V}^{k_0}(x,t,k)$, where the jump matrix $\mathsf{V}^{k_0}$ meets the estimations
		\begin{equation}
			\begin{cases}
				||\partial_{x}^{i}(\mathsf{V}^{(3)}-\mathsf{V}^{k_0})||_{L^{1}(\mathbf{L}^{\epsilon})}\le C t^{-1}\ln{t},\\
				||\partial_{x}^{i}(\mathsf{V}^{(3)}-\mathsf{V}^{k_0})||_{L^{\infty}(\mathbf{L}^{\epsilon})}\le C t^{-\frac{1}{2}}\ln{t},
			\end{cases}
			\qquad \eta\in\mathcal{K},\quad t\ge2,\quad i=0,1.
		\end{equation}
		Moreover, for large $t$, one has
		\begin{equation}
			||\partial_{x}^{i}(\mathcal{M}^{k_{0}}(x,t,\cdot)^{-1}-I)||_{L^{\infty}(\partial B_{\epsilon}(k_0))}=O(t^{-1/2}),\qquad i=0,1,
		\end{equation}
		\begin{equation}\label{equation of mk0}
			\frac{1}{2\pi {\rm i}}\int_{\partial B_{\epsilon}(k_0)}(\mathcal{M}^{k_{0}}(x,t,k)^{-1}-I)dk=-\frac{T(\eta,t)\mathcal{M}_{1}^{L}(p(\eta))T(\eta,t)^{-1}}{3^{1/4}\sqrt{2t}}+O(t^{-1}),
		\end{equation}
		uniformly for $\eta\in\mathcal{K}$.
	\end{lemma}
	
	\begin{figure}[htp]
		\centering
		\begin{tikzpicture}[>=latex]
			\draw[very thick] (-4,0) to (4,0) node[right]{$\widehat{\Gamma}$};
			\draw[very thick] (-2,-1.732*2) to (2,1.732*2);
			\draw[very thick] (-2,1.732*2) to (2,-1.732*2);
			\draw[<->,very thick] (-3.5,0)node[below]{$8$} to (2.5,0);
			\draw[<-<,very thick] (-2.6,0)--(-1.4,0);
			\draw[->,very thick] (0,0) to (-.8,0) node[below]{$7$};
			\draw[->,very thick] (0,0) to (1,0) ;
			\draw[<->,very thick] (-1.5,-1.732*1.5)node[right=1mm]{$5$} to (1.35,1.732*1.35);
			\draw[>->,very thick]  (.3,1.732*0.3) to (.75,1.732*0.75);
			\draw[<->,very thick] (1.8,1.732*1.8) to (-.5,-1.732*0.5)node[right]{$6$};
			\draw[<->,very thick] (-1.5,1.732*1.5) to (1.35,-1.732*1.35);
			\draw[<->,very thick] (-.5,1.732*0.5) to (.35,-1.732*0.35) ;
			\draw[>->,very thick] (.7,-1.732*0.7) to (1.8,-1.732*1.8) ;
			
			\draw[very thick] (-3.2,-1.732*0.8)--(-.8,1.732*0.8);
			\draw[<->,very thick] (-2.6,-1.732*0.4)node[right=1mm]{$3$}--(-1.4,1.732*0.4)node[right=1mm]{$1$};
			\draw[<->,very thick] (-2.99,-1.732*0.66)--(-1.01,1.732*0.66);
			\draw[very thick] (-3.2,1.732*0.8)--(-.8,-1.732*0.8);
			\draw[<->,very thick] (-2.6,1.732*0.4)node[right=1mm]{$2$}--(-1.4,-1.732*0.4)node[right=1mm]{$4$};
			\draw[<->,very thick] (-2.99,1.732*0.66)--(-1.01,-1.732*0.66);
			\draw[very thick] (-.8,-1.732*0.8)--(2.8,-1.2*1.732);
			\draw[<->,very thick] (.1,-1.732*0.9)--(1.9,-1.1*1.732);
			\draw[<->,very thick] (-.5,-1.732*2.5/3)--(2.5,-3.5*1.732/3);
			\draw[very thick] (.4,1.732*2)--(1.6,0);
			\draw[<->,very thick] (.7,-1.732*1.5)--(1.3,-0.5*1.732);
			\draw[<->,very thick] (.5,-1.732*11/6)--(1.5,-1/6*1.732);
			\draw[very thick] (-.8,1.732*0.8)--(2.8,1.2*1.732);
			\draw[<->,very thick] (.1,1.732*0.9)--(1.9,1.1*1.732);
			\draw[<->,very thick] (-.5,1.732*2.5/3)--(2.5,3.5*1.732/3);
			\draw[<->,very thick] (.7,1.732*1.5)--(1.3,0.5*1.732);
			\draw[very thick] (.4,-1.732*2)--(1.6,0);
			\draw[<->,very thick] (.5,1.732*11/6)--(1.5,1/6*1.732);
			\filldraw[black] (-2,0) node[below=1mm]{$k_{0}$} circle (1.5pt);
			\filldraw[black] (1,1.732)  circle (1.5pt);
			\filldraw[black] (1,-1.732)  circle (1.5pt);
			\node at (1.5,1.6){$\omega^{2}k_{0}$};
			\node at (1.5,-1.6){$\omega k_{0}$};
			\draw[->,very thick] (-2,1) arc (90:459:1);
			\draw[->,very thick] (0,-1.732) arc (180:580:1);
			\draw[->,very thick] (1,2.732) arc (90:460:1);
		\end{tikzpicture}
		\caption{The contour $\widetilde{\Gamma}$ of RH problem for $\widetilde{\mathcal{M}}$.}
		\label{Gamma hat figure}
	\end{figure}
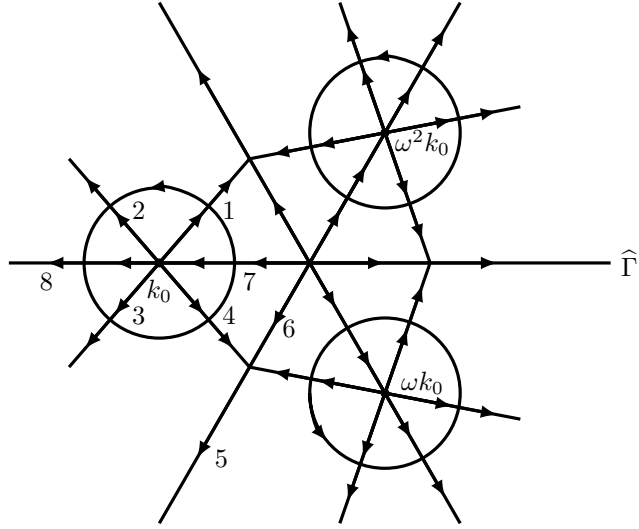

	\subsection{The small-norm RH problem}
	Reminding the symmetry
	\begin{equation*}
		\mathcal{M}^{k_0}(x,t,k)=\mathcal{A}\mathcal{M}^{k_0}(x,t,\omega k)\mathcal{A}^{-1},
	\end{equation*}
	extend the region of definition for the function $\mathcal{M}^{k_0}$ from disk $B_{\epsilon}(k_0)$ to $\mathbb{B}$ by defining
	\begin{equation*}
		\widetilde{\mathcal{M}}=\begin{cases}
			\mathcal{M}^{(3)}(\mathcal{M}^{k_{0}})^{-1},& k\in\mathbb{B}\\
			\mathcal{M}^{(3)}, & {\rm elsewhere}.
		\end{cases}
	\end{equation*}
	Let $\widetilde{\Gamma}=\Gamma^{(3)}\cup\partial\mathbb{B}$ as shown in Figure \ref{Gamma hat figure} and the corresponding jump matrix is given by
	\begin{equation*}
		\widetilde{\mathsf{V}}=\begin{cases}
			\mathsf{V}^{(3)}, & k\in\widetilde{\Gamma}\setminus\bar{\mathbb{B}},\\
			(\mathcal{M}^{k_{0}})^{-1}, & k\in\partial\mathbb{B},\\
			\mathcal{M}_{-}^{k_{0}}\mathsf{V}^{(3)}(\mathcal{M}_{+}^{k_{0}})^{-1}, & k\in \widetilde{\Gamma}\cap\mathbb{B}.
		\end{cases}
	\end{equation*}
	
	\begin{RH problem}[RH problem for $\widetilde{\mathcal{M}}$]\label{RH problem for hatm}
		The function $\widetilde{\mathcal{M}}(x,t,k)$ is a three-order matrix which satisfies $\widetilde{\mathcal{M}}\in I +E^{3}(\mathbb{C}\setminus \widehat{\Gamma})$ and $\widetilde{\mathcal{M}}_{+}(x,t,k)=\widetilde{\mathcal{M}}_{-}(x,t,k)\widetilde{\mathsf{V}}(x,t,k)$ for a.e. $k\in\widetilde{\Gamma}$.
	\end{RH problem}
	
	Let $\widetilde{\mathbf{L}}^{\epsilon}=\mathbf{L}^{\epsilon}\cup \omega \mathbf{L}^{\epsilon}\cup\omega^2 \mathbf{L}^{\epsilon}$. Define the counter
	\begin{equation*}
		\Gamma^{\prime}=\widetilde{\Gamma}\setminus(\Gamma\cup\widetilde{\mathbf{L}}^{\epsilon}\cup\partial\mathbb{B}).
	\end{equation*}
	\begin{lemma}\label{estimate 1}
		Set $\widetilde{\mathsf{W}}=\widetilde{\mathsf{V}}-I$, then for $t\ge2$ and $\eta\in\mathcal{K}$,  the estimates below hold:
		\begin{subequations}
			\begin{align*}
				&||(1+|\cdot|)\partial_{x}^{i}\widetilde{\mathsf{W}}||_{(L^{1}\cap L^{\infty})(\Gamma)}\le\frac{C}{k_{0}t},\\
				&||(1+|\cdot|)\partial_{x}^{i}\widetilde{\mathsf{W}}||_{(L^{1}\cap L^{\infty})(\Gamma^{\prime})}\le Ce^{-ct},\\
				&||\partial_{x}^{i}\widetilde{\mathsf{W}}||_{(L^{1}\cap L^{\infty})(\partial \mathbb{B})}\le C t^{-1/2},\\
				&||\partial_{x}^{i}\widetilde{\mathsf{W}}||_{L^{1}(\widetilde{\mathbf{L}}^{\epsilon})}\le C t^{-1}\ln{t},\\
				&||\partial_{x}^{i}\widetilde{\mathsf{W}}||_{L^{\infty}(\widetilde{\mathbf{L}}^{\epsilon})}\le C t^{-1/2}\ln{t},\\
			\end{align*}
		\end{subequations}
		for $i=0,1$.
	\end{lemma}
	Supposing $f$ is a complex-valued function defined on $\widetilde{\Gamma}$, the Cauchy operator $\widehat{\mathcal{C}}f$ is defined by
	\begin{equation*}
		(\widehat{\mathcal{C}}f)(z)=\frac{1}{2\pi {\rm i}}\int_{\widetilde{\Gamma}}\frac{f(s)ds}{s-z},\qquad z\in\mathbb{C}\setminus\widetilde{\Gamma}.
	\end{equation*}

	Further, define $\tilde{\nu}$ by
	\begin{equation*}		\tilde{\nu}=I+(I-\widehat{\mathcal{C}}_{\widetilde{\mathsf{W}}})^{-1}\widehat{\mathcal{C}}_{\widetilde{\mathsf{W}}}I\in I+\dot{L}^{3}(\widetilde{\Gamma}).
	\end{equation*}
	for $k\in\widetilde{\Gamma}$, $t\ge T$ and $\eta=\frac{x}{t}\in\mathcal{K}$.
	
	\begin{lemma}
		For $t\ge T$ and $\eta\in\mathcal{K}$, there is unique solution $\widetilde{\mathcal{M}}\in I+\dot{E}^{3}(\mathbb{C}\setminus\widetilde{\Gamma})$ of RH problem \ref{RH problem for hatm}, which is given by
		\begin{equation*}			\widetilde{\mathcal{M}}(x,t,k)=I+\widehat{\mathcal{C}}(\tilde{\nu}\widetilde{\mathsf{W}})=I+\frac{1}{2\pi\mathrm{i}}\int_{\widetilde{\Gamma}}\tilde{\nu}(x,t,s)\widetilde{\mathsf{W}}(x,t,s)\frac{ds}{s-k}.
		\end{equation*}
	\end{lemma}
	
	\begin{lemma}\label{estimate 2}
		Assuming $1<p<\infty$, there is an estimate for $\tilde{\nu}$ in the form
		\begin{equation*}
			||\partial_{x}^{i}(\tilde{\nu}-I)||_{L^{p}(\widetilde{\Gamma})}\le C t^{-\frac{1}{2}}(\ln{t})^{\frac{p-1}{p}}, \quad i=0,1,\quad \eta\in\mathcal{K},
		\end{equation*}
		for any sufficiently large $t$.
	\end{lemma}
	
	As $k\to\infty$, there exists the nontangential limit:
	\begin{equation*}
		K(x,t):=\lim\limits_{k\to\infty}^{\angle}k(\widetilde{\mathcal{M}}-I)
		=-\frac{1}{2\pi\mathrm{i}}\int_{\widetilde{\Gamma}}\tilde{\nu}(x,t,k)\widetilde{\mathsf{W}}(x,t,k)dk,
	\end{equation*}
	which has the following asymptotic behavior for large $t$.
	\begin{lemma}
		As $t\to\infty$, the function $K(x,t)$ behaves as
		\begin{equation}\label{K as t to infty}
			K(x,t)=-\frac{1}{2\pi\mathrm{i}}\int_{\partial \mathbb{B}}\widetilde{\mathsf{W}}(x,t,k)dk+O(t^{-1}\ln{t}).
		\end{equation}
	\end{lemma}
	\begin{proof}
		Decompose $K(x,t)$ into
		\begin{equation*}
			K(x,t)=-\frac{1}{2\pi\mathrm{i}}\int_{\partial \mathbb{B}}\widetilde{\mathsf{W}}(x,t,k)dk+K_{1}(x,t)+K_{2}(x,t),
		\end{equation*}
		where
		\begin{equation*}
			K_{1}(x,t)=-\frac{1}{2\pi\mathrm{i}}\int_{\widetilde{\Gamma}\setminus\partial \mathbb{B}}\widetilde{\mathsf{W}}(x,t,k)dk, \quad K_{2}(x,t)=-\frac{1}{2\pi\mathrm{i}}\int_{ \widetilde{\Gamma}}(\tilde{\nu}(x,t,k)-I)\widetilde{\mathsf{W}}(x,t,k)dk.
		\end{equation*}
		Recalling Lemmas \ref{estimate 1} and \ref{estimate 2}, the	asymptotic behavior for	$K(x,t)$ in (\ref{K as t to infty}) is obvious.
	\end{proof}
	\par	
	Define function $F(\eta,t)$ as
	\begin{equation*}
		F(\eta,t)=-\frac{1}{2\pi\mathrm{i}}\int_{\partial B_{\epsilon}(k_{0})}\widetilde{\mathsf{W}}(x,t,k)dk=-\frac{1}{2\pi\mathrm{i}}\int_{\partial B_{\epsilon}(k_{0})}((\mathcal{M}^{k_{0}})^{-1}-I)dk,
	\end{equation*}
	then the equation (\ref{equation of mk0}) implies	
	\begin{equation*}		F(\eta,t)=\frac{T(\eta,t)\mathcal{M}_{1}^{L}(p(\eta))T(\eta,t)^{-1}}{3^{1/4}\sqrt{2t}}+O(t^{-1}\ln{t}),
	\end{equation*}
	as $t\to\infty$. Moreover, the symmetry of $\widetilde{\mathsf{V}}$ indicates that the function $\widetilde{\mathcal{M}}(x,t,k)$ satisfies the symmetry
	\begin{equation*}
		\widetilde{\mathcal{M}}(x,t,k)=\mathcal{A}\widetilde{\mathcal{M}}(x,t,\omega k)\mathcal{A}^{-1}, \qquad k\in\mathbb{C}\setminus \widetilde{\Gamma},
	\end{equation*}
	so do the functions $\tilde{\nu}$ and $\widetilde{\mathsf{W}}$. Taking all these symmetries into the equation (\ref{K as t to infty}), it is seen that
	\begin{equation*}
		\begin{aligned}
			-\frac{1}{2\pi\mathrm{i}}\int_{\partial \mathbb{B}}\widetilde{\mathsf{W}}(x,t,k)dk&=-\frac{1}{2\pi \mathrm{i}}(\int_{\partial\mathcal{B}_{\epsilon}(k_{0})}+\int_{\omega\partial\mathcal{B}_{\epsilon}(k_{0})}+\int_{\omega^2\partial\mathcal{B}_{\epsilon}(k_{0})})\widetilde{\mathsf{W}}(x,t,k)dk\\	&=F(\eta,t)+\omega\mathcal{A}^{-1}F(\eta,t)\mathcal{A}+\omega^2\mathcal{A}^{-2}F(\eta,t)\mathcal{A}^{2}.
		\end{aligned}.
	\end{equation*}
	\par
	Finally, we can infer from the equation (\ref{K as t to infty}) that
	\begin{equation}\label{M as t to infty}
		\partial_{x}^{i}\lim_{k\to\infty}k(\widetilde{\mathcal{M}}(x,t,k)-I)=\partial_{x}^{i}(\frac{\sum_{j=0}^{2}\omega^{j}\mathcal{A}^{-j}T(\eta,t)\mathcal{M}_{1}^{L}(p(\eta))T(\eta,t)^{-1}\mathcal{A}^{j}}{3^{1/4}\sqrt{2t}})+O(\frac{\ln{t}}{t}),
	\end{equation}
	uniformly for $\eta\in\mathcal{K}$ and $i=0,1$ as $t\to\infty$.

	\subsection{Asympototics of $u(x,t)$ and $v(x,t)$ for Hirota-Satsuma equation (\ref{HS eq})}
	Taking the transformations (\ref{H(x,t,k)}),(\ref{M2}) and (\ref{G(x,t,k)}) into consideration,
	the function $\mathcal{M}(x,t,k)$ for RH problem \ref{RH problem for M} can be rewritten as
	\begin{equation}\label{MGDH}
		\mathcal{M}=\widetilde{\mathcal{M}}G^{-1}\mathcal{D}^{-1}H^{-1}
	\end{equation}
	for any $k\in\mathbb{C}\setminus\bar{\mathbb{B}}$. Now we are prepared to prove Theorem \ref{large-time theorem-HS}.

	\begin{proof}  According to the equation (\ref{MGDH}), the reconstruction formula (\ref{fanjie u v}) in Section \ref{Section-2} can be rewritten as
		\begin{equation}\label{reconstruction-new-new}
			\begin{cases}
				u(x,t)&=3\lim\limits_{k\to\infty}^{\angle} k\big[\widetilde{\mathcal{N}}_{3}(x,t,k)-\widetilde{\mathcal{Q}}_{3}(x,t,k)\big]+O(\frac{\ln{t}}{t}),   \\
				v(x,t)&=3\frac{\partial}{\partial x}\lim\limits_{k\to\infty}^{\angle}k\big[\widetilde{\mathcal{N}}_{3}(x,t,k)-1\big]+O(\frac{\ln{t}}{t}),
			\end{cases}
		\end{equation}
		where $\widetilde{\mathcal{N}}=(\omega,\omega^2,1)\widetilde{\mathcal{M}}$ and
		$\widetilde{\mathcal{Q}}=(1,1,1)\widetilde{\mathcal{M}}$.
		\par
		With the equation (\ref{M as t to infty}) and equality $\overline{\Gamma(\mathrm{i}\mu)}=\Gamma(-\mathrm{i}\mu)$ in mind, the reconstruction formula (\ref{reconstruction-new-new}) for $\eta>0$ becomes
		\begin{equation*}
			\begin{aligned}			u(x,t)&=3\lim\limits_{k\to\infty}^{\angle} k\big[\widetilde{\mathcal{N}}_{3}(x,t,k)-\widetilde{\mathcal{Q}}_{3}(x,t,k)\big]+O(\ln{(t)}/t)\\				&=3\big((\omega,\omega^2,1)\frac{\sum_{j=0}^{2}\omega^{j}\mathcal{A}^{-j}T(\eta,t)\mathcal{M}_{1}^{L}(p(\eta))T(\eta,t)^{-1}\mathcal{A}^{j}}{3^{1/4}\sqrt{2t}}\\
				&\quad-(1,1,1)\frac{\sum_{j=0}^{2}\omega^{j}\mathcal{A}^{-j}Y(\eta,t)\mathcal{M}_{1}^{L}
(p(\eta))Y(\eta,t)^{-1}\mathcal{A}^{j}}{3^{1/4}\sqrt{2t}}\big)_{3}+O(\ln{(t)}/t)\\
				&=3\frac{(\omega^2-\omega)\alpha_{21}d_0^{-1}e^{t\Phi_{21}(\eta,k_0)}+(\omega-\omega^2)
\alpha_{12}d_0e^{-t\theta_{21}(\eta,k_0)}}{3^{1/4}\sqrt{2t}}+O(\ln{(t)}/t)\\
				&=\frac{3^{3/4}\times 2}{\sqrt{2t}}\mathrm{Re}[{(\omega^2-\omega)\alpha_{21}d_0^{-1}e^{t\theta_{21}(\eta,k_0)}}]+O(\ln{(t)}/t)\\
				&=-\frac{3^{5/4}\times 2}{\sqrt{2t}}\mathrm{Re}[{e^{\frac{\pi i}{2}}\alpha_{21}d_0^{-1}e^{t\theta_{21}(\eta,k_0)}}]+O(\ln{(t)}/t),\\
			\end{aligned}
		\end{equation*}
		and
		\begin{equation*}
			\begin{aligned}
				v(x,t)&=3\frac{\partial}{\partial x}\lim\limits_{k\to\infty}^{\angle}k\big[\widetilde{\mathcal{N}}_{3}(x,t,k)-1\big]+O(\ln{(t)}/t)\\
				&=3\frac{\partial}{\partial x}\big((\omega,\omega^2,1)\frac{\sum_{j=0}^{2}\omega^{j}\mathcal{A}^{-j}T(\eta,t)
\mathcal{M}_{1}^{L}(p(\eta))T(\eta,t)^{-1}\mathcal{A}^{j}}{3^{1/4}\sqrt{2t}}\big)_{3}+O(\ln{(t)}/t)\\
				&=3\frac{\partial}{\partial x}\frac{\omega^2 d_{0}^{-1}e^{t\theta_{21}(\eta,k_{0})}\alpha_{21}+\omega d_{0}e^{-t\theta_{21}(\eta,k_0)}\alpha_{12}}{3^{1/4}\sqrt{2t}}+O(\ln{(t)}/t)\\
				&=\frac{3^{3/4}\times 2}{\sqrt{2t}}\frac{\partial}{\partial x}\big[\mathrm{Re}(\omega^2 d_{0}^{-1}e^{t\theta_{21}(\eta,k_{0})}\alpha_{21})\big]+O(\ln{(t)}/t),\\
			\end{aligned}
		\end{equation*}
		as $t\to\infty$.
		Utilizing the facts
		\begin{equation*}		|\Gamma(\mathrm{i}\mu)|=\frac{\sqrt{2\pi}}{\sqrt{\mu}\sqrt{e^{\pi\mu}-e^{-\pi\mu}}}=\frac{\sqrt{2\pi}}{\sqrt{\mu}e^{\frac{\pi\mu}{2}}|p|},
		\end{equation*}
		and
		\begin{equation*}
			\mathcal{D}_{2}^{-1}(\eta,k_{0})\mathcal{D}_{6}^{-1}(\eta,k_{0})=\exp{\big[\mathrm{i}\mu\log{(3k_{0}^{2})}+\frac{1}{\pi \mathrm{i}}\int_{k_{0}}^{-\infty}\log{|\omega k_{0}-s|d\ln{(1-|\rho_{1}(s)|^{2})}}}\big],
		\end{equation*}
		the asymptotic solution (\ref{asympotics for k more than 0}) in Theorem \ref{large-time theorem-HS} is derived.
		\par
		Obviously, the Hirota-Satsuma equation (\ref{HS eq}) has symmetry $x\to -x, t\to t, u \to -u, v\to v$.
		Denote $f(x,t)=-u(-x,t),g(x,t)=v(-x,t)$, then the initial condition becomes $f_{0}(x)=-u_{0}(-x),g_{0}(x)=v_{0}(-x)$. So direct calculations imply that the symmetry relations below satisfy
		\begin{equation*}
			Y_{1}(-x,-k;f_{0},g_{0})=Y_{2}(x,k;u_{0},v_{0}),  \qquad Y_{1}^{A}(-x,-k;f_{0},g_{0})=Y_{2}^{A}(x,k;u_{0},v_{0}),
		\end{equation*}
		which results in
		\begin{equation}\label{relation between s(k) and sA}
			s(-k;f_{0},g_{0})=(s^{A}(k;u_{0},v_{0}))^{T}, \qquad s^{A}(-k;f_{0},g_{0})=(s(k;u_{0},v_{0}))^{T}.
		\end{equation}
		The symmetries of $s(k)$ and $s^{A}(k)$	in Section \ref{Section-2} and the equalities in equation (\ref{relation between s(k) and sA}) indicate that
		\begin{equation*}
			|\rho_{2}(k)|=|\rho_{1}(-k)|, \quad k\in (0,\infty).
		\end{equation*}
		\par	
		Combining all the symmetries above along with the large-time asymptotics of the Hirota-Satsuma equation (\ref{HS eq}) for $\eta>0$ in (\ref{asympotics for k more than 0}), the asymptotics (\ref{asympotics for k less than 0}) of the Hirota-Satsuma equation (\ref{HS eq}) for $\eta<0$ is obtained quickly.

	\end{proof}

	\subsection{Remark: asympototics of $w(x,t)$ for the good Boussinesq equation (\ref{good boussinesq})}	Recalling the Miura transformation (\ref{HS-to-GB}) between the Hirota-Satsuma equation (\ref{HS eq}) and the good Boussinesq equation (\ref{good boussinesq}), one can write down the asymptotic solution of equation (\ref{good boussinesq}) for large $t$. To do so, for $\eta>0$, combining the asymptotic expression (\ref{asympotics for k more than 0}) of the Hirota-Satsuma equation (\ref{HS eq}) with Miura transformation $w=-\frac{1}{6}u^2-\frac{1}{2}v$, yields the large-time asymptotics of the good Boussinesq equation (\ref{good boussinesq}) for $\eta=x/t>0$ below
$$
		\begin{aligned}
			w(x,t)&=-\frac{3^{5/4}\times \sqrt{\mu}k_0}{\sqrt{2t}}\sin\big(\mu\ln(6\sqrt{3}tk_{0}^{2})+\frac{1}{\pi}
			\int^{-\infty}_{k_0}\ln{\frac{|s-k_{0}|}{|s-\omega k_{0}|}}d\ln(1-|\rho_1|^2)\\
			&-\sqrt{3}tk_0^{2}+\frac{19\pi}{12}-\arg{p}-\arg{\Gamma(i\mu)}\big)+O(\ln{(t)}/t).\\
		\end{aligned}
$$
This completes the proof of Theorem \ref{large-time theorem-GB}.
The similar asymptotic expression has been derived by Charlier, Lenells and Wang \cite{Charlier-Lenells-Wang-2021}.
\par
Moreover, it is remarked that the large-time asymptotic solution of the modified Boussinesq equation (\ref{mb eq}) in \cite{Wang-Zhu-2022} and the Miura transformation (\ref{mGB-to-HS}) can also result in the large-time asymptotic solution of the Hirota-Satsuma equation (\ref{HS eq}) in Theorem \ref{large-time theorem-HS}.\\
	\par
	\par
	{\bf Acknowledgements}
	\par
	This work was supported by National Natural Science Foundation of China through grant No. 12371247.\\
	\par
	{\bf Conflict of Interest}
	\par
	The authors have no conflicts to disclose.\\
	\par
	\par
	{\bf Data Availability}
	\par
	Data sharing is not applicable to this article as no new data were created or analyzed in this study.

	\bibliographystyle{amsplain}

\end{document}